\documentclass[11pt,a4paper]{article}
\pdfoutput=1
\usepackage[ascii]{inputenc}
\usepackage[T1]{fontenc}
\usepackage[USenglish]{babel}
\usepackage[bookmarksnumbered,hypertexnames=false,colorlinks,colorlinks=true,linkcolor=blue,urlcolor=blue,citecolor=blue,anchorcolor=green,pdfusetitle]{hyperref}
\usepackage{bbm,braket,microtype,mathrsfs,amsmath,amsthm,amssymb,color,mathtools,fullpage,graphicx,enumitem,ae,aecompl,youngtab,dsfont,float,xcolor,csquotes,youngtab}
\usepackage[capitalize]{cleveref}
\usepackage[backend=bibtex8,style=numeric,sorting=none,doi=false,isbn=false,url=false,maxbibnames=7,maxcitenames=2,
noerroretextools
]{biblatex}
\let\etoolboxforlistloop\forlistloop%
\usepackage{autonum}
\let\forlistloop\etoolboxforlistloop%
\renewbibmacro{in:}{}
\usepackage[affil-it]{authblk}
\usepackage{amsbsy}
\newtheorem{thm}{Theorem}\crefname{thm}{Theorem}{Theorems}
\newtheorem{prop}[thm]{Proposition}\crefname{prop}{Proposition}{Propositions}
\newtheorem{lem}[thm]{Lemma}\crefname{lem}{Lemma}{Lemmas}
\newtheorem{cor}[thm]{Corollary}\crefname{cor}{Corollary}{Corollaries}

\newtheorem*{lem-singlet-bound}{\Cref{lem:singlet bound}}\crefname{lem-singlet-bound}{Lemma}{Lemmas}
\newtheorem*{thm-expectation-convergence}{\Cref{thm:expectation-convergence}}\crefname{thm-expectation-convergence}{Theorem}{Theorems}
\newtheorem*{thm-probabilistic}{\Cref{thm:prob-asymptotics}}\crefname{thm-probabilistic}{Theorem}{Theorems}
\newtheorem*{thm-deterministic-standard}{\Cref{thm:det-asymptotics}}\crefname{thm-deterministic-standard}{Theorem}{Theorems}
\newtheorem*{thm-frenchguys}{\Cref{thm:thefrenchguys}}\crefname{thm-frenchguys}{Theorem}{Theorems}
\newtheorem*{thm-frenchguys-summary}{\Cref{thm:thefrenchguys-summary}}\crefname{thm-frenchguys-summary}{Theorem}{Theorems}
\newtheorem*{cor-converse}{\Cref{cor:converse}}\crefname{cor-converse}{Corollary}{Corollaries}

\newtheorem{rem}[thm]{Remark}\crefname{rem}{Remark}{Remarks}
\crefname{exa}{Example}{Examples}
\numberwithin{equation}{section}
\numberwithin{thm}{section}
\allowdisplaybreaks[4]
\DeclareMathOperator{\tr}{tr}
\DeclareMathOperator{\GL}{GL}
\DeclareMathOperator{\poly}{poly}
\DeclareMathOperator{\Sym}{Sym}

\newcommand{\Hi}{\mathcal{H}}
\newcommand{\cP}{\mathcal{P}}
\newcommand{\ot}{\otimes}
\newcommand{\hi}{\Hi}
\newcommand{\cB}{\mathcal{B}}
\newcommand{\cE}{\mathcal{E}}
\newcommand{\C}{\mathbb{C}}
\newcommand{\CC}{\mathbb{C}}
\newcommand{\R}{\mathbb{R}}
\newcommand{\N}{\mathbb{N}}
\newcommand{\Z}{\mathbb{Z}}
\newcommand{\ZZ}{\mathbb{Z}}
\newcommand{\Q}{\mathbb{Q}}
\newcommand{\ketbra}[2]{\left|#1\right\rangle\!\left\langle#2\right|}
\newcommand{\proj}[1]{\ketbra{#1}{#1}}
\newcommand{\GUEzd}{\operatorname{GUE}^0_d}
\newcommand{\GUEd}{\operatorname{GUE}_d}
\newcommand{\RV}{\mathbf}
\newcommand{\bE}{\mathbb{E}}
\newcommand{\tbE}{\tilde{\bE}}
\newcommand{\OS}{\mathrm{OS}}
\newcommand{\bV}{\mathbb{V}}
\newcommand{\sumi}{\sum\nolimits}
\newcommand{\EE}{\mathbb{E}}
\newcommand{\D}{\mathrm{d}}
\DeclareMathOperator{\idch}{id}

\newcommand{\ox}{\otimes}
\newcommand{\eps}{\varepsilon}
\definecolor{mellowred}{rgb}{.6,.1,.1}

\newcommand{\epr}{\mathrm{EPR}}
\newcommand{\std}{\hspace{0.1mm}\mathrm{std}}

\DeclareMathOperator{\vol}{vol}
\newcommand{\ssum}{\sideset{}{'}\sum}

\let\originalleft\left
\let\originalright\right
\renewcommand{\left}{\mathopen{}\mathclose\bgroup\originalleft}
\renewcommand{\right}{\aftergroup\egroup\originalright}
\usepackage{tikz}
\usepackage{pgfplots}
\usepgfplotslibrary{groupplots}
\definecolor{green}{HTML}{009000}
\makeatletter
\g@addto@macro\bfseries{\boldmath}
\makeatother

\begin{document}
\title{Asymptotic performance of port-based teleportation}
\author[1]{Matthias Christandl}%
\author[2,3]{Felix Leditzky\thanks{Email: \texttt{felix.leditzky@jila.colorado.edu}}}
\author[4,5]{Christian Majenz\thanks{Email: \texttt{c.majenz@uva.nl}}}
\author[2,3,6]{Graeme Smith}%
\author[4,7]{Florian Speelman}%
\author[4,5,8,9]{Michael Walter}%

\affil[1]{\small QMATH, Department of Mathematical Sciences, University of Copenhagen, Denmark}
\affil[2]{\small JILA, University of Colorado/NIST, USA}
\affil[3]{\small Center for Theory of Quantum Matter, University of Colorado Boulder, CO, USA}
\affil[4]{\small QuSoft, Amsterdam, The Netherlands}
\affil[5]{\small Institute for Logic, Language and Computation, University of Amsterdam, The Netherlands}
\affil[6]{\small Department of Physics, University of Colorado Boulder, CO, USA}
\affil[7]{\small CWI, Amsterdam, The Netherlands}
\affil[8]{\small Korteweg-de Vries Institute for Mathematics, University of Amsterdam, The Netherlands}
\affil[9]{\small Institute for Theoretical Physics, University of Amsterdam, The Netherlands}

\date{}
\maketitle
\begin{abstract}
Quantum teleportation is one of the fundamental building blocks of quantum Shannon theory.
While ordinary teleportation is simple and efficient, port-based teleportation (PBT) enables applications such as universal programmable quantum processors, instantaneous non-local quantum computation and attacks on position-based quantum cryptography.
In this work, we determine the fundamental limit on the performance of PBT:
for arbitrary fixed input dimension and a large number $N$ of ports, the error of the optimal protocol is proportional to the inverse square of $N$.
We prove this by deriving an achievability bound, obtained by relating the corresponding optimization problem to the lowest Dirichlet eigenvalue of the Laplacian on the ordered simplex.
We also give an improved converse bound of matching order in the number of ports.
In addition, we determine the leading-order asymptotics of PBT variants defined in terms of maximally entangled resource states.
The proofs of these results rely on connecting recently-derived representation-theoretic formulas to random matrix theory.
Along the way, we refine a convergence result for the fluctuations of the Schur-Weyl distribution by Johansson, which might be of independent interest.
\end{abstract}

\tableofcontents

\section{Introduction}
\subsection{Port-based teleportation}
\begin{figure}[ht]
  \centering
  \includegraphics[width=.7\textwidth]{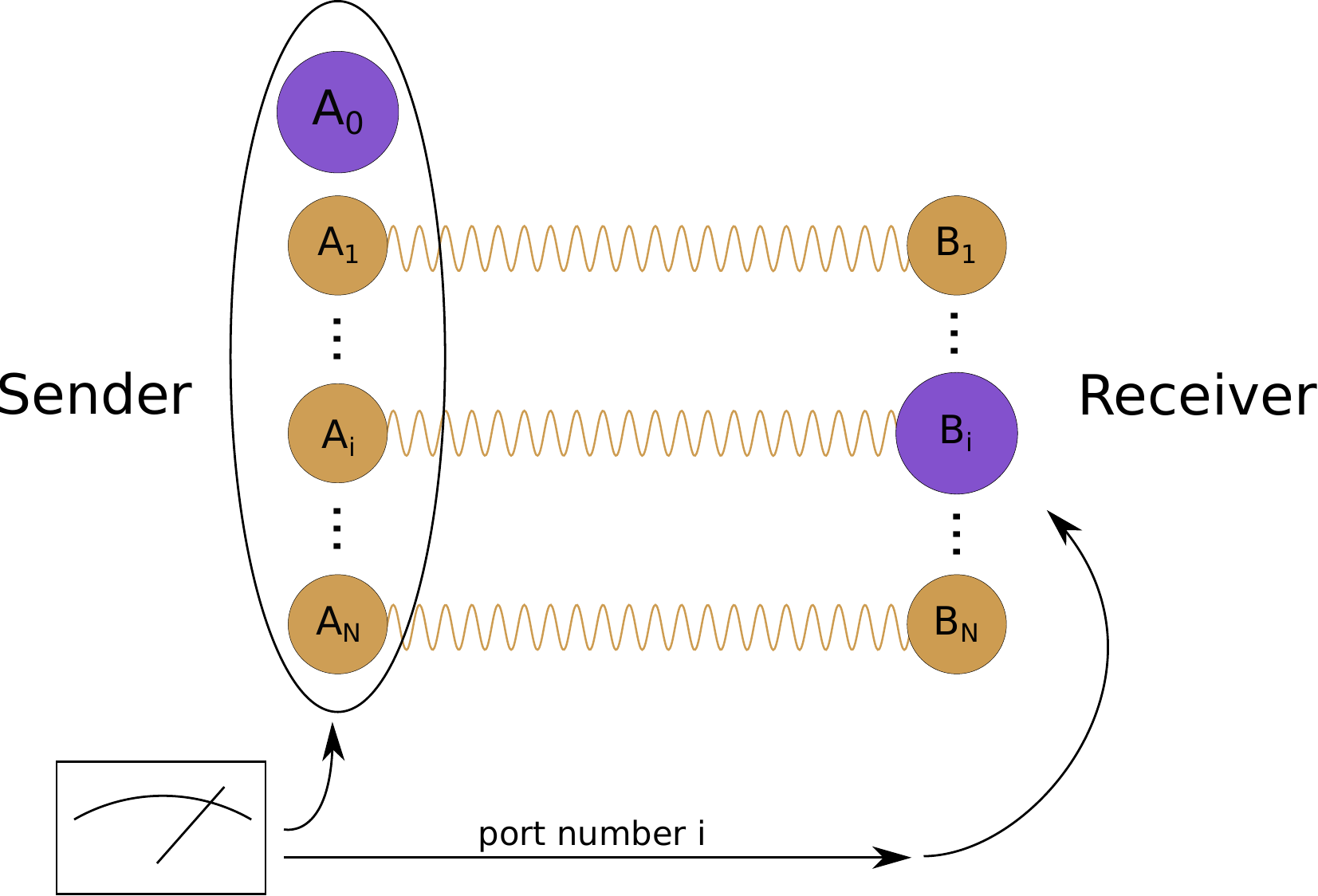}
  \caption{Schematic representation of port-based teleportation (PBT). Like in ordinary teleportation, the sender applies a joint measurement to her input system $A$ and her parts of the entangled resource, $A_i, i=1,\dots,N$, and sends the outcome to the receiver, who applies a correction operation. In PBT, however, this correction operation merely consists of choosing one of the subsystems $B_i$, the \emph{ports}, of the entangled resource.  A PBT protocol cannot implement a perfect quantum channel with a finite number of ports. There are different variants of PBT. The four commonly studied ones are characterized by whether failures are announced, or heralded (probabilistic PBT) or go unnoticed (deterministic PBT), and whether simplifying constraints on the resource state and the sender's measurement are enforced.}
  \label{fig:pbt}
\end{figure}
Port-based teleportation (PBT)~\cite{ishizaka2008asymptotic,ishizaka2009quantum} is a variant of the ubiquitous quantum teleportation protocol~\cite{Bennett1993}, where the receiver's correction operation consists of merely picking the right subsystem, called \emph{port}, of their part of the entangled resource state.
\Cref{fig:pbt} provides a schematic description of the protocol (see \Cref{sec:pbt} for a more detailed explanation).
While being far less efficient than the ordinary teleportation protocol, the simple correction operation allows the receiver to apply a quantum operation to the output of the protocol before receiving the classical message.
This \emph{simultaneous unitary covariance property} enables all known applications that require PBT instead of just ordinary quantum teleportation, including the construction of universal programmable quantum processors~\cite{ishizaka2008asymptotic}, quantum channel discrimination~\cite{Pirandola2018} and instantaneous non-local quantum computation (INQC)~\cite{beigi2011simplified}.

In the INQC protocol, which was devised by Beigi and K\"onig~\cite{beigi2011simplified}, two spatially separated parties share an input state and wish to perform a joint unitary on it.
To do so, they are only allowed a single simultaneous round of communication.
INQC provides a generic attack on any quantum position-verification scheme~\cite{Buhrman2014}, a protocol in the field of position-based cryptography~\cite{beigi2011simplified,Chandran2009,Malaney2016,Unruh2014}.
It is therefore of great interest for cryptography to characterize the resource requirements of INQC\@: it is still open whether a computationally secure quantum position-verification scheme exists, as all known generic attacks require an exponential amount of entanglement.
Efficient protocols for INQC are only known for special cases~\cite{yu2011fast,yu2012fast,broadbent2016popescu,speelman2016instantaneous}.
The best lower bounds for the entanglement requirements of INQC are, however, linear in the input size~\cite{beigi2011simplified,tomamichel2013monogamy,ribeiro2015tight}, making the hardness of PBT, the corner stone of the best known protocol, the only indication for a possible hardness of INQC\@.

PBT comes in two variants, deterministic and probabilistic, the latter being distinguished from the former by the fact that the protocol implements a perfect quantum channel whenever it does not fail (errors are ``heralded''). In their seminal work~\cite{ishizaka2008asymptotic,ishizaka2009quantum}, Ishizaka and Hiroshima completely characterize the problem of PBT for qubits.
They calculate the performance of the standard and optimized protocols for deterministic and the EPR and optimized protocols for probabilistic PBT, and prove the optimality of the `pretty good' measurement in the standard deterministic case.
They also show a lower bound for the standard protocol for deterministic PBT, which was later reproven in~\cite{beigi2011simplified}.
Further properties of PBT were explored in~\cite{strelchuk2013generalized}, in particular with respect to recycling part of the resource state.
Converse bounds for the probabilistic and deterministic versions of PBT have been proven in~\cite{Pitalua-Garcia2013a} and~\cite{Ishizaka2015}, respectively.
In~\cite{Wang2016}, exact formulas for the fidelity of the standard protocol for deterministic PBT with $N=3$ or $4$ in arbitrary dimension are derived using a graphical algebra approach.
Recently, exact formulas for arbitrary input dimension in terms of representation-theoretic data have been found for all four protocols, and the asymptotics of the optimized probabilistic case have been derived~\cite{studzinski2016port,Mozrzymas2017}.

Note that, in contrast to ordinary teleportation, a protocol obtained from executing several PBT protocols is not again a PBT protocol.
This is due to the fact that the whole input system has to be teleported to the same output port for the protocol to have the mentioned simultaneous unitary covariance property.
Therefore, the characterization of protocols for any dimension $d$ is of particular interest.
The mentioned representation-theoretic formulas derived in~\cite{studzinski2016port,Mozrzymas2017} provide such a characterization.
It is, however, not known how to evaluate these formulas efficiently for large input dimension.

\subsection{Summary of main results}\label{sec:summary}

In this paper we provide several characterization results for port-based teleportation. As our main contributions, we characterize the leading-order asymptotic performance of fully optimized deterministic port-based teleportation (PBT), as well as the standard protocol for deterministic PBT and the EPR protocol for probabilistic PBT. In the following, we provide a detailed summary of our results.

Our first, and most fundamental, result concerns deterministic PBT and characterizes the leading-order asymptotics of the optimal fidelity for a large number of ports.
\begin{thm}\label{thm:main}
  For arbitrary but fixed local dimension $d$, the optimal entanglement fidelity for deterministic port-based teleportation behaves asymptotically as
  \begin{align}
  F_d^*(N) = 1 - \Theta(N^{-2}).
  \end{align}
\end{thm}
\Cref{thm:main} is a direct consequence of \cref{thm:thefrenchguys-summary} below.
Prior to our work, it was only known that $F_d^*(N) = 1 - \Omega(N^{-2})$ as a consequence of an explicit converse bound~\cite{Ishizaka2015}.
We prove that this asymptotic scaling is in fact achievable, and we also provide a converse with improved dependency on the local dimension, see \cref{cor:converse}.

For deterministic port-based teleportation using a maximally entangled resource and the pretty good measurement, a closed expression for the entanglement fidelity was derived in~\cite{studzinski2016port}, but its asymptotics for fixed $d>2$ and large $N$ remained undetermined.
As our second result, we derive the asymptotics of deterministic port-based teleportation using a maximally entangled resource and the pretty good measurement, which we call the \emph{standard protocol}.

\newcommand{\restateDeterministicStandard}{%
	For arbitrary but fixed $d$ and any $\delta>0$, the entanglement fidelity of the standard protocol of PBT is given by
	\begin{align}
	F^{\std}_d(N)=1-\frac{d^2-1}{4N}+O(N^{-\frac 3 2+\delta}).
	\end{align}
	}

\begin{thm}\label{thm:det-asymptotics}
  \restateDeterministicStandard
\end{thm}
Previously, the asymptotic behavior given in the above theorem was only known for~$d=2$ in terms of an exact formula for finite $N$; for $d>2$, it was merely known that $F^{\std}_d(N)=1-O\left(N^{-1}\right)$~\cite{ishizaka2009quantum}.
In \Cref{fig:det-pbt-standard} we compare the asymptotic formula of \Cref{thm:det-asymptotics} to a numerical evaluation of the exact formula derived in~\cite{studzinski2016port} for $d\leq 5$.

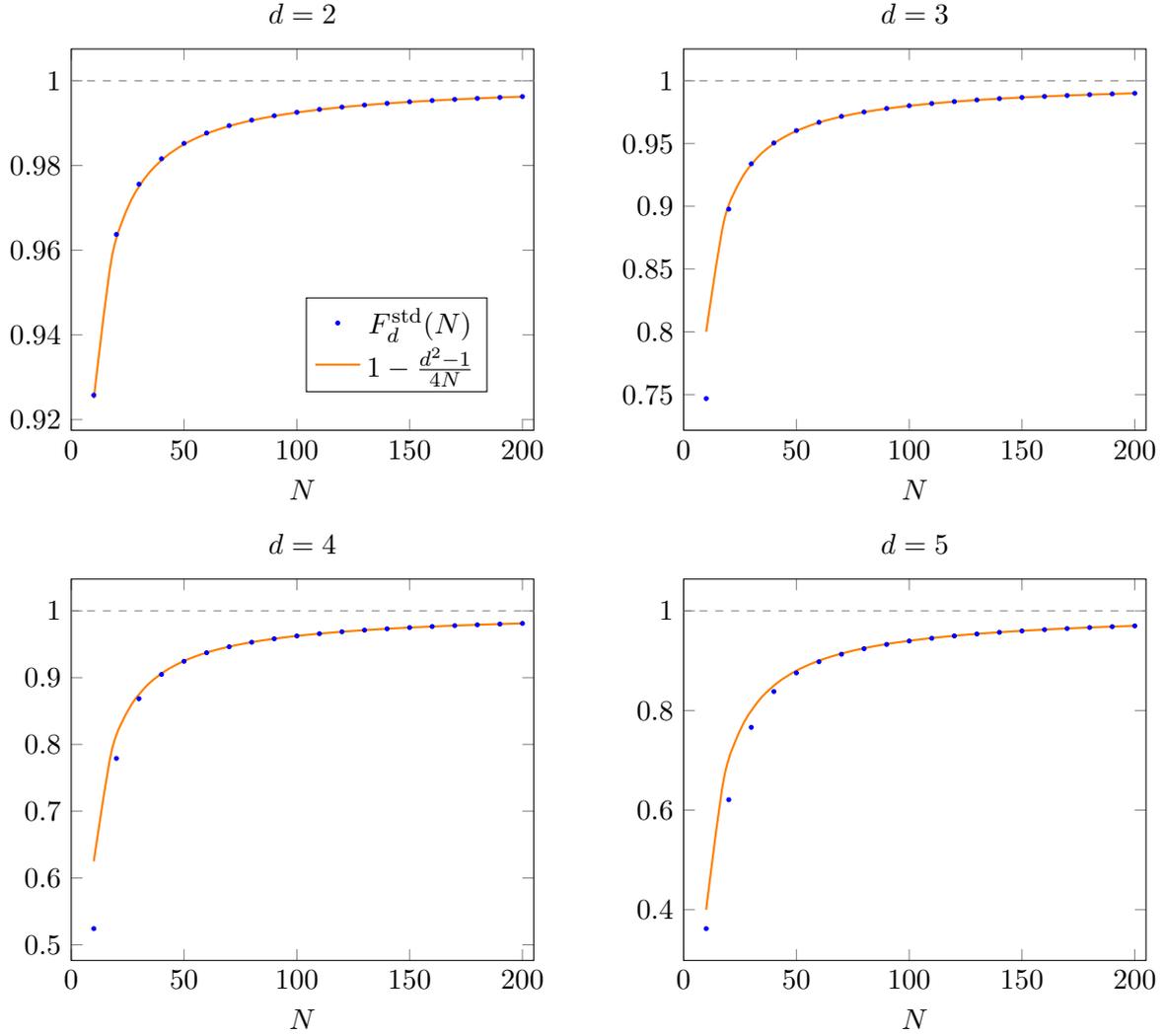
\begin{figure}
	\centering
	\begin{tikzpicture}
	\begin{groupplot}[scale=.9,group style = {group size=2 by 2, horizontal sep = 2cm, vertical sep=2cm}]
	\nextgroupplot[
	xmin = 0,
	xmax = 205,
	xlabel = $N$,
	legend style = {at = {(0.9,0.1)},anchor = south east,cells={align=left}},
	title = {$d=2$}
	]
	\addplot[only marks, mark size = 0.75pt, color=blue] table[x=N,y=F] {numerics/pbt-det-std-d2.txt};
	\addplot[smooth,thick,domain=10:200,color=orange] {1-3/(4*x)};
	\addplot[dashed,domain=0:210,color=gray] {1};
	\legend{{$F_d^{\std}(N)$},$1-\frac{d^2-1}{4N}$};

	\nextgroupplot[
	xmin = 0,
	xmax = 205,
	xlabel = $N$,
	legend style = {at = {(0.95,0.1)},anchor = south east},
	title = {$d=3$},
	ytick = {0.75,0.8,0.85,0.9,0.95,1},
	]
	\addplot[only marks, mark size = 0.75pt, color=blue] table[x=N,y=F] {numerics/pbt-det-std-d3.txt};
	\addplot[smooth,thick,domain=10:200,color=orange] {1-2/x};
	\addplot[dashed,domain=0:210,color=gray] {1};

	\nextgroupplot[
	xmin = 0,
	xmax = 205,
	xlabel = $N$,
	legend style = {at = {(0.95,0.1)},anchor = south east},
	title = {$d=4$}
	]
	\addplot[only marks, mark size = 0.75pt, color=blue] table[x=N,y=F] {numerics/pbt-det-std-d4.txt};
	\addplot[smooth,thick,domain=10:200,color=orange] {1-15/(4*x)};
	\addplot[dashed,domain=0:210,color=gray] {1};

	\nextgroupplot[
	xmin = 0,
	xmax = 205,
	xlabel = $N$,
	legend style = {at = {(0.95,0.1)},anchor = south east},
	title = {$d=5$}
	]
	\addplot[only marks, mark size = 0.75pt, color=blue] table[x=N,y=F] {numerics/pbt-det-std-d5.txt};
	\addplot[smooth,thick,domain=10:200,color=orange] {1-6/x};
	\addplot[dashed,domain=0:210,color=gray] {1};
	\end{groupplot}
	\end{tikzpicture}
	\caption{Entanglement fidelity of the standard protocol for deterministic port-based teleportation in local dimension $d=2,3,4,5$ using $N$ ports~\cite{git}. We compare the exact formula~\eqref{eq:cambridge} for $F_d^{\std}$ (blue dots) with the first-order asymptotics obtained from \cref{thm:det-asymptotics} (orange curve).}
	\label{fig:det-pbt-standard}
\end{figure}

For probabilistic port-based teleportation, \textcite{Mozrzymas2017} obtained the following expression for the success probability $p^*_d$ optimized over arbitrary entangled resources:
\begin{align}
p^*_d(N)=1-\frac{d^2-1}{d^2-1+N},
\end{align}
valid for all values of $d$ and $N$ (see the detailed discussion in \cref{sec:pbt}).
In the case of using $N$ maximally entangled states as the entangled resource, an exact expression for the success probability in terms of representation-theoretic quantities was also derived in~\cite{studzinski2016port}.
We state this expression in \eqref{eq:cambridge epr prob} in \cref{sec:pbt}.
However, its asymptotics for fixed $d>2$ and large $N$ have remained undetermined to date.
As our third result, we derive the following expression for the asymptotics of the success probability of the optimal protocol among the ones that use a maximally entangled resource, which we call the \emph{EPR protocol}.

\newcommand{\restateProbabilistic}{%
	For probabilistic port-based teleportation in arbitrary but fixed dimension $d$ with EPR pairs as resource states,
	\begin{align}
	p^{\epr}_d(N) = 1 -  \sqrt{\frac d {N-1}} \mathbb E[\lambda_{\max}(\mathbf G)] + o\left(N^{-1/2}\right),
	\end{align}
	where $\mathbf{G}\sim \GUEzd$.
	}

\begin{thm}\label{thm:prob-asymptotics}
	\restateProbabilistic
\end{thm}
The famous Wigner semicircle law~\cite{Wigner1993} provides an asymptotic expression for the expected maximal eigenvalue, $ \mathbb E[\lambda_{\max}(\mathbf G)]\sim 2\sqrt d$ for $d\to\infty$.
Additionally, there exist explicit upper and lower bounds for all $d$, see the discussion in \cref{sec:prob}.

\begin{figure}
\def\lambdamaxII{1.12838}
\def\lambdamaxIII{1.90414}
\def\lambdamaxIV{2.52811}
\def\lambdamaxV{3.06311}

	\centering
	\begin{tikzpicture}
	\begin{groupplot}[scale=.9,group style = {group size=2 by 2, horizontal sep = 2cm, vertical sep=2cm}]
	\nextgroupplot[
	xmin = 0,
	xmax = 510,
	xlabel = $N$,
	legend style = {at = {(0.9,0.1)},anchor = south east,cells={align=left}},
	title = {$d=2$, $c_2 = \lambdamaxII$}
	]
	\addplot[only marks, mark size = 0.75pt, color=blue] table[x=N,y=p] {numerics/pbt-prob-epr-d2.txt};
	\addplot[smooth,thick,domain=10:500,color=orange] {1-\lambdamaxII*sqrt(2/(x-1))};
	\addplot[dashed,domain=0:510,color=gray] {1};
	\legend{{$p_d^{\epr}(N)$\hspace{1.5cm}},{$1-c_d \sqrt{d/(N-1)}$}};

	\nextgroupplot[
	xmin = 0,
	xmax = 510,
	ymin = 0,
	xlabel = $N$,
	legend style = {at = {(0.95,0.1)},anchor = south east},
	title = {$d=3$, $c_3 = \lambdamaxIII$}
	]
	\addplot[only marks, mark size = 0.75pt, color=blue] table[x=N,y=p] {numerics/pbt-prob-epr-d3.txt};
	\addplot[smooth,thick,domain=10:500,color=orange] {1-\lambdamaxIII*sqrt(3/(x-1))};
	\addplot[dashed,domain=0:510,color=gray] {1};

	\nextgroupplot[
	xmin = 0,
	xmax = 510,
	ymin = 0,
	xlabel = $N$,
	legend style = {at = {(0.95,0.1)},anchor = south east},
	title = {$d=4$, $c_4 = \lambdamaxIV$}
	]
	\addplot[only marks, mark size = 0.75pt, color=blue] table[x=N,y=p] {numerics/pbt-prob-epr-d4.txt};
	\addplot[smooth,thick,domain=10:500,color=orange] {1-\lambdamaxIV*sqrt(4/(x-1))};
	\addplot[dashed,domain=0:510,color=gray] {1};

	\nextgroupplot[
	xmin = 0,
	xmax = 470,
	ymin = 0,
	xlabel = $N$,
	legend style = {at = {(0.95,0.1)},anchor = south east},
	title = {$d=5$, $c_5 = \lambdamaxV$}
	]
	\addplot[only marks, mark size = 0.75pt, color=blue] table[x=N,y=p] {numerics/pbt-prob-epr-d5.txt};
	\addplot[smooth,thick,domain=10:460,color=orange] {1-\lambdamaxV*sqrt(5/(x-1))};
	\addplot[dashed,domain=0:470,color=gray] {1};
	\end{groupplot}
	\end{tikzpicture}
	\caption{Success probability of the EPR protocol for probabilistic port-based teleporation in local dimension $d=2,3,4,5$ using $N$ ports~\cite{git}. We compare the exact formula~\eqref{eq:cambridge epr prob} for $p_d^{\epr}$ (blue dots) with the first-order asymptotic formula obtained from \cref{thm:prob-asymptotics} (orange curve). The first-order coefficient $c_d \equiv \mathbb E[\lambda_{\max}(\mathbf G)]$ appearing in the formula in \cref{thm:prob-asymptotics} was obtained by numerical integration from the eigenvalue distribution of $\GUEd$.}
	\label{fig:prob-pbt-epr}
\end{figure}
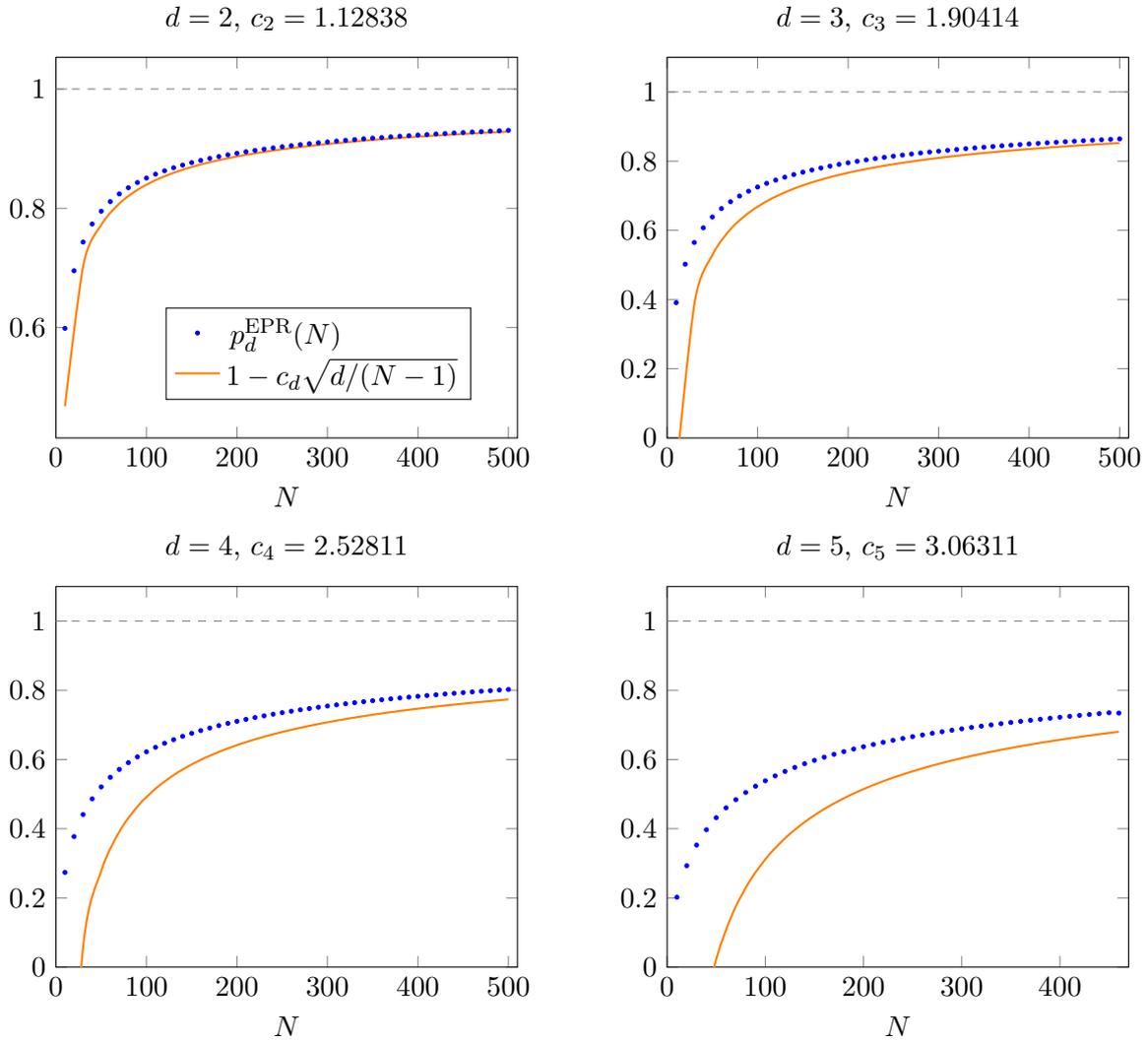

To establish \cref{thm:det-asymptotics,thm:prob-asymptotics}, we analyze the asymptotics of the Schur-Weyl distribution, which also features in other fundamental problems of quantum information theory including spectrum estimation, tomography, and the quantum marginal problem~\cite{keyl2001estimating,hayashi2002quantum,Christandl2006,o2015quantum,haah2017sample,o2016efficient,OW17,Christandl2012,christandl2014eigenvalue,christandl2018recoupling,Christandl2007a}.
Our main technical contribution is a new convergence result for its fluctuations that strengthens a previous result by Johansson~\cite{johansson2001discrete}.
This result, which may be of independent interest, is stated as \cref{thm:expectation-convergence} in \Cref{sec:schur weyl dist}.

\Cref{thm:main} is proved by giving an asymptotic lower bound for the optimal fidelity of deterministic PBT, as well as an upper bound that is valid for any number of ports and matches the lower bound asymptotically. For the lower bound, we again use an expression for the entanglement fidelity of the optimal deterministic PBT protocol derived in~\cite{Mozrzymas2017}.
The asymptotics of this formula for fixed $d$ and large $N$ have remained undetermined so far.
We prove an asymptotic lower bound for this entanglement fidelity in terms of the lowest Dirichlet eigenvalue of the Laplacian on the ordered $(d-1)$-dimensional simplex.

\newcommand{\restateFrenchguys}{%
	The optimal fidelity for deterministic port-based teleportation is bounded from below by
	\begin{align}
	F_d^*(N) &\geq 1-\frac{\lambda_1(\OS_d)}{dN^2}-O(N^{-3}),
	\end{align}
	where
	\begin{align}
	\OS_{d-1}=\left\{x\in \R^d\bigg|\sum\nolimits_i x_i=1, x_i\ge x_{i+1}, x_d\ge 0\right\}
	\end{align}
	is the $(d-1)$-dimensional simplex of ordered probability distributions with $d$ outcomes and
	$\lambda_1(\Omega)$ is the first eigenvalue of the Dirichlet Laplacian on a domain $\Omega$.
	}

\begin{thm}\label{thm:thefrenchguys}
  \restateFrenchguys
\end{thm}
Using a bound from~\cite{Freitas2008} for $\lambda_1(\OS_d)$, we obtain the following explicit lower bound.

\newcommand{\restateFrenchguysSummary}{%
	For the optimal fidelity of port-based teleportation with arbitrary but fixed input dimension $d$ and $N$ ports, the following bound holds,
	\begin{align}
	F^*_d(N)\ge 1-\frac{ d^5+O(d^{9/2})}{4\sqrt 2 N^2}+O(N^{-3}).
	\end{align}
	}

\begin{thm}\label{thm:thefrenchguys-summary}
  \restateFrenchguysSummary
\end{thm}

\smallskip

As a complementary result, we give a strong upper bound for the entanglement fidelity of any deterministic port-based teleportation protocol.
While valid for any finite number $N$ of ports, its asymptotics for large $N$ are given by $1-O(N^{-2})$, matching \cref{thm:thefrenchguys-summary}.

\newcommand{\restateConverse}{%
	For a general port-based teleportation scheme with input dimension $d$ and $N$ ports, the entanglement fidelity $F_d^*$ and the diamond norm error $\varepsilon_d^*$ can be bounded as
	\begin{align}
	F_d^*(N) &\leq \begin{cases}
	\frac{\sqrt N}{d}&\text{ if } N\le \frac{d^2}2\\
	1-\frac{d^2-1}{16N^2}&\text{ otherwise}
	\end{cases}
	&
	\varepsilon_d^*(N)\geq\begin{cases}
	2\bigl(1-\frac{\sqrt N}{d}\bigr)&\text{ if } N\le \frac{d^2}2\\
	2\frac{d^2-1}{16N^2}&\text{ otherwise.}
	\end{cases}
	\end{align}
	}

\begin{cor}\label{cor:converse}
	\restateConverse
\end{cor}
Previously, the best known upper bound on the fidelity~\cite{Ishizaka2015} had the same dependence on $N$, but was increasing in $d$, thus failing to reflect the fact that the task becomes harder with increasing~$d$. Interestingly, a lower bound from~\cite{Kubicki2018} on the program register size of a universal programmable quantum processor also yields a converse bound for PBT that is incomparable to the one from~\cite{Ishizaka2015} and weaker than our bound.

Finally we provide a proof of the following `folklore' fact that had been used in previous works on port-based teleportation. The unitary and permutation symmetries of port-based teleportation imply that the entangled resource state and Alice's measurement can be chosen to have these symmetries as well. Apart from simplifying the optimization over resource states and POVMs, this implies that characterizing the entanglement fidelity is sufficient to give worst-case error guarantees. Importantly, this retrospectively justifies the use of the entanglement fidelity~$F$ in the literature about deterministic port-based teleportation in the sense that any bound on $F$ implies a bound on the diamond norm error without losing dimension factors. This is also used to show the diamond norm statement of \cref{cor:converse}.
\begin{prop}[\Cref{prop:symPOVM,prop:symsuffice,cor:PBT-FvdG}, informal]
  There is an explicit transformation between port-based teleportation protocols that preserves any unitarily invariant distance measure on quantum channels, and maps an arbitrary port-based teleportation protocol with input dimension $d$ and $N$ ports to a protocol that
  \begin{enumerate}[label=(\roman*),font=\normalfont]
    \item\label{item:symmetries} has a resource state and a POVM with $U(d)\times S_N$ symmetry, and
    \item\label{item:unitary-covariance} implements a unitarily covariant channel.
  \end{enumerate}
  In particular, the transformation maps an arbitrary port-based teleportation protocol to one with the symmetries \ref{item:symmetries} and \ref{item:unitary-covariance} above, and entanglement fidelity no worse than the original protocol.
  Point \ref{item:unitary-covariance} implies that
  \begin{align}
 \varepsilon_d^*=2\left(1-F_d^*\right),
  \end{align}
  where $F_d^*$ and $\varepsilon_d^*$ denote the optimal entanglement fidelity and optimal diamond norm error for deterministic port-based teleportation.
\end{prop}

\subsection{Structure of this paper}
In \cref{sec:preliminaries} we fix our notation and conventions and recall some basic facts about the representation theory of the symmetric and unitary groups.
In \cref{sec:pbt} we define the task of port-based teleportation (PBT) in its two main variants, the probabilistic and deterministic setting.
Moreover, we identify the inherent symmetries of PBT, and describe a representation-theoretic characterization of the task.
In \cref{sec:schur weyl dist} we discuss the Schur-Weyl distribution and prove a convergence result that will be needed to establish our results for PBT with maximally entangled resources.
Our first main result is proved in \cref{sec:prob}, where we discuss the probabilistic setting in arbitrary dimension using EPR pairs as ports, and determine the asymptotics of the success probability $p^{\epr}_d$ (\cref{thm:prob-asymptotics}).
Our second main result, derived in \cref{sec:standard}, concerns the deterministic setting in arbitrary dimension using EPR pairs, for which we compute the asymptotics of the optimal entanglement fidelity $F^{\std}_d$ (\cref{thm:det-asymptotics}).
Our third result, an asymptotic lower bound on the entanglement fidelity $F_d^*$ of the optimal protocol in the deterministic setting (\cref{thm:thefrenchguys-summary}), is proved in \cref{sec:optimal}.
Finally, in \cref{sec:converse} we derive a general non-asymptotic converse bound on deterministic port-based teleportation protocols using a non-signaling argument (\cref{thm:converse-bound}).
We also present a lower bound on the communication requirements for approximate quantum teleportation (\cref{cor:porttelebound}).
We make some concluding remarks in \cref{sec:conclusion}.
The appendices contain technical proofs.

\section{Preliminaries}\label{sec:preliminaries}
\subsection{Notation and definitions}\label{sec:notation}
We denote by~$A$, $B$, \dots quantum systems with associated Hilbert spaces~$\Hi_A$, $\Hi_B$, \dots, which we always take to be finite-dimensional, and we associate to a multipartite quantum system $A_1\dots A_n$ the Hilbert space~$\Hi_{A_1\dots{}A_n}=\Hi_{A_1}\ox\dots\ox\Hi_{A_n}$.
When the $A_i$ are identical, we also write $A^n=A_1\dots{}A_n$.
The set of linear operators on a Hilbert space $\Hi$ is denoted by $\cB(\Hi)$.
A \emph{quantum state}~$\rho_A$ on quantum system~$A$ is a positive semidefinite linear operator~$\rho_A\in\cB(\Hi_A)$ with unit trace, i.e., $\rho_A\geq 0$ and $\tr(\rho_A)=1$.
We denote by $I_A$ or $1_A$ the identity operator on $\Hi_A$, and by $\tau_A=I_A/|A|$ the corresponding \emph{maximally mixed} quantum state, where $|A|\coloneqq\dim\Hi_A$.
A pure quantum state~$\psi_A$ is a quantum state of rank one.
We can write $\psi_A=\ketbra\psi\psi_A$ for a unit vector $\ket\psi_A\in\Hi_A$.
For quantum systems~$A,A'$ of dimension $\dim\Hi_A=\dim\Hi_{A'}=d$ with bases $\lbrace |i\rangle_A\rbrace_{i=1}^d$ and $\lbrace |i\rangle_{A'}\rbrace_{i=1}^d$, the vector $|\phi^+\rangle_{A'A} = \frac{1}{\sqrt{d}} \sum_{i=1}^d |i\rangle_{A'}\ox |i\rangle_A$ defines the \emph{maximally entangled} state of Schmidt rank~$d$.
The \emph{fidelity}~$F(\rho,\sigma)$ between two quantum states is defined by $F(\rho,\sigma)\coloneqq \|\sqrt{\rho}\sqrt{\sigma}\|_1^2$, where $\|X\|_1=\tr(\sqrt{X^\dagger X})$ denotes the \emph{trace norm} of an operator.
For two pure states $|\psi\rangle$ and $|\phi\rangle$, the fidelity is equal to $F(\psi,\phi)=|\langle\psi|\phi\rangle|^2$.
A \emph{quantum channel} is a completely positive, trace-preserving linear map $\Lambda\colon\cB(\Hi_A)\to\cB(\Hi_B)$.
We also use the notation $\Lambda\colon A\to B$ or $\Lambda_{A\to B}$, and we denote by $\idch_A$ the identity channel on~$A$.
Given two quantum channels $\Lambda_1,\Lambda_2\colon A\to B$, the \emph{entanglement fidelity} $F(\Lambda_1,\Lambda_2)$ is defined as
\begin{align}
F(\Lambda_1,\Lambda_2) \coloneqq F((\idch_{A'}\ox\Lambda_1)(\phi^+_{A'A}),(\idch_{A'}\ox\Lambda_2)(\phi^+_{A'A})),
\label{eq:entanglement-fidelity-definition}
\end{align}
and we abbreviate $F(\Lambda)\coloneqq F(\Lambda,\idch)$. The \emph{diamond norm} of a linear map $\Lambda\colon\cB(\Hi_A)\to\cB(\Hi_B)$ is defined by
\begin{align}
\|\Lambda\|_\diamond \coloneqq \sup_{\|X_{A'A}\|_1\leq 1}\|(\idch_{A'}\ox\Lambda)(X_{A'A}) \|_1.
\label{eq:diamond-distance}
\end{align}
The induced distance on quantum channels is called the \emph{diamond distance}.
A \emph{positive operator-valued measure (POVM)} $E=\lbrace E_x\rbrace$ on a quantum system~$A$ is a collection of positive semidefinite operators $E_x\geq 0$ satisfying $\sum_x E_x = I_A$.

We denote random variables by bold letters ($\mathbf X$, $\mathbf Y$, $\mathbf Z$, \dots) and the valued they take by the non-bold versions ($X,Y,Z,\dots$).
We denote by $\mathbf X \sim \mathbb P$ that $\mathbf X$ is a random variable with probability distribution~$\mathbb P$.
We write $\Pr(\dots)$ for the probability of an event and $\EE\left[\dots\right]$ for expectation values.
The notation $\mathbf X_n\overset{P}{\to}\mathbf X \ (n\to\infty)$ denotes \emph{convergence in probability} and $\mathbf X_n\overset{D}{\to}\mathbf X \ (n\to\infty)$ denotes \emph{convergence in distribution}.
The latter can be defined, e.g., by demanding that $\EE\left[f(\mathbf X_n)\right] \to \EE\left[f(\mathbf X)\right] \ (n\to\infty)$ for every continuous, bounded function~$f$.
The \emph{Gaussian unitary ensemble}~$\GUEd$ is the probability distribution on the set of Hermitian $d\times d$-matrices~$H$ with density $Z_d^{-1} \exp(-\frac{1}{2}\tr H^2)$, where $Z_d$ is the appropriate normalization constant.
Alternatively, for $\mathbf{X}\sim\GUEd$, the entries $\mathbf{X}_{ii}$ for $1\leq i\leq d$ are independently distributed as $\mathbf{X}_{ii}\sim N(0,1)$, whereas the elements $\mathbf{X}_{ij}$ for $1\leq i<j\leq d$ are independently distributed as $\mathbf{X}_{ij}\sim N(0,\frac{1}{2})+iN(0,\frac{1}{2})$.
Here, $N(0,\sigma^2)$ denotes the centered normal distribution with variance~$\sigma^2$.
The \emph{traceless Gaussian unitary ensemble} $\GUEzd$ can be defined as the distribution of the random variable
$\mathbf{Y} \coloneqq \mathbf{X} - \tfrac{\tr\mathbf{X}}{d} I$,
where $\mathbf{X}\sim \GUEd$.

For a complex number $z\in\mathbb{C}$, we denote by $\Re( z)$ and $\Im( z)$ its real and imaginary part, respectively.
We denote by $\mu\vdash_d n$ a partition $(\mu_1,\dots,\mu_d)$ of~$n$ into $d$~parts.
That is, $\mu\in\ZZ^d$ with $\mu_1\geq\mu_2\geq\cdots\geq\mu_d\geq0$ and $\sum_i\mu_i=n$.
We also call $\mu$ a \emph{Young diagram} and visualize it as an arrangement of boxes, with $\mu_i$ boxes in the $i$-th row.
For example, $\mu=(3,1)$ can be visualized as {\Yboxdim6pt\Yvcentermath1$\yng(3,1)$}.
We use the notation $(i,j)\in\mu$ to mean that $(i,j)$ is a box in the Young diagram $\mu$, that is, $1\leq i\leq d$ and $1\leq j\leq\mu_i$.
We denote by $\GL(\mathcal H)$ the general linear group and by $U(\mathcal H)$ the unitary group acting on a Hilbert space $\mathcal H$.
When $\mathcal H=\CC^d$, we write $\GL(d)$ and $U(d)$.
Furthermore, we denote by $S_n$ the symmetric group on $n$ symbols.
A \emph{representation} $\varphi$ of a group $G$ on a vector space $\Hi$ is a map $G\ni g\mapsto \varphi(g)\in \GL(\Hi)$ satisfying $\varphi(gh)=\varphi(g)\varphi(h)$ for all $g,h\in G$.
In this paper all representations are unitary, which means that $\Hi$ is a Hilbert space and $\varphi(g) \in U(\Hi)$ for every $g\in G$.
A representation is irreducible (or an \emph{irrep}) if $\Hi$ contains no nontrivial invariant subspace.

\subsection{Representation theory of the symmetric and unitary group}\label{subsec:schur weyl}
Our results rely on the representation theory of the symmetric and unitary groups and Schur-Weyl duality (as well as their semiclassical asymptotics which we discuss in \cref{sec:schur weyl dist}).
In this section we introduce the relevant concepts and results (see e.g.~\cite{fulton1997young,Simon1996}.

The irreducible representations of $S_n$ are known as \emph{Specht modules} and labeled by Young diagrams with $n$~boxes.
We denote the Specht module of $S_n$ corresponding to a Young diagram $\mu\vdash_d n$ by $[\mu]$ ($d$ is arbitrary).
Its dimension is given by the \emph{hook length formula}~\cite[p.~53--54]{fulton1997young},
\begin{align}\label{eq:specht hook}
d_\mu=\frac{n!}{\prod_{(i,j)\in\mu}h_\mu(i,j)},
\end{align}
where $h_\mu(i,j)$ is the \emph{hook length} of the hook with corner at the box $(i,j)$, i.e., the number of boxes below $(i,j)$ plus the number of boxes to the right of $(i,j)$ plus one (the box itself).

The polynomial irreducible representations of $U(d)$ are known as \emph{Weyl modules} and labeled by Young diagrams with no more than~$d$ rows.
We denote the Weyl module of~$U(d)$ corresponding to a Young diagram $\mu\vdash_d n$ by $V^d_\mu$ ($n$ is arbitrary).
Its dimension can be computed using \emph{Stanley's hook length formula}~\cite[p.~55]{fulton1997young},
\begin{align}\label{eq:stanley hook}
m_{d,\mu}=\prod_{(i,j)\in\mu} \frac{d+c(i,j)}{h_\mu(i,j)},
\end{align}
where $c(i,j) = j-i$ is the so-called \emph{content} of the box $(i,j)$.
This is an alternative to the \emph{Weyl dimension formula}, which states that
\begin{align}\label{eq:weyl dim}
m_{d,\mu}=\prod_{1\leq i<j\leq d} \frac{\mu_i - \mu_j + j - i}{j - i}.
\end{align}
We stress that $m_{d,\mu}$ depends on the dimension~$d$.

Consider the representations of $S_n$ and $U(d)$ on ${\left(\C^d\right)}^{\otimes n}$ given by permuting the tensor factors, and multiplication by $U^{\otimes n}$, respectively.
Clearly the two actions commute.
\emph{Schur-Weyl duality} asserts that the decomposition of ${\left(\C^d\right)}^{\otimes n}$ into irreps takes the form (see, e.g.,~\cite{Simon1996})
\begin{align}\label{eq:schur-weyl}
  {\left(\C^d\right)}^{\otimes n} \cong \bigoplus_{\mu\vdash_d n} [\mu] \otimes V^d_\mu.
\end{align}

\section{Port-based teleportation}\label{sec:pbt}
The original quantum teleportation protocol for qubits (henceforth referred to as ordinary teleportation protocol) is broadly described as follows~\cite{Bennett1993}: Alice (the sender) and Bob (the receiver) share an EPR pair (a maximally entangled state on two qubits), and their goal is to transfer or `teleport' another qubit in Alice's possession to Bob by sending only classical information.
Alice first performs a joint Bell measurement on the quantum system to be teleported and her share of the EPR pair, and communicates the classical measurement outcome to Bob using two bits of classical communication.
Conditioned on this classical message, Bob then executes a correction operation consisting of one of the Pauli operators on his share of the EPR pair.
After the correction operation, he has successfully received Alice's state.
The ordinary teleportation protocol can readily be generalized to qudits, i.e., $d$-dimensional quantum systems.
Note that while the term `EPR pair' is usually reserved for a maximally entangled state on two qubits ($d=2$), we use the term more freely for maximally entangled states of Schmidt rank $d$ on two qudits, as defined in \Cref{sec:preliminaries}.

Port-based teleportation, introduced by Ishizaka and Hiroshima~\cite{ishizaka2008asymptotic,ishizaka2009quantum}, is a variant of quantum teleportation where Bob's correction operation solely consists of picking one of a number of quantum subsystems upon receiving the classical message from Alice.
In more detail, Alice and Bob initially share an entangled resource quantum state $\psi_{A^N B^N}$, where $\hi_{A_i}\cong\hi_{B_i}\cong\C^d$ for $i=1,\dots,N$.
We may always assume that the resource state is pure, for we can give a purification to Alice and she can choose not to use it.
Bob's quantum systems $B_i$ are called \emph{ports}.
Just like in ordinary teleportation, the goal is for Alice to teleport a $d$-dimensional quantum system $A_0$ to Bob.
To achieve this, Alice performs a joint POVM $\{(E_i)_{A_0A^N}\}_{i=1}^N$ on the input and her part of the resource state and sends the outcome $i$ to Bob.
Based on the index $i$ he receives, Bob selects the $i$-th port, i.e.\ the system $B_i$, as being the output register (renaming it to $B_0$), and discards the rest.
That is, in contrast to ordinary teleportation, Bob's decoding operation solely consists of selecting the correct port $B_i$.
The quality of the teleportation protocol is measured by how well it simulates the identity channel from Alice's input register~$A_0$ to Bob's output register~$B_0$.

Port-based teleportation is impossible to achieve perfectly with finite resources~\cite{ishizaka2008asymptotic}, a fact first deduced from the application to universal programmable quantum processors~\cite{nielsen1997programmable}.
There are two ways to deal with this fact: either one can just accept an imperfect protocol, or one can insist on simulating a perfect identity channel, with the caveat that the protocol will fail from time to time.
This leads to two variants of PBT, which are called \emph{deterministic} and \emph{probabilistic} PBT in the literature~\cite{ishizaka2008asymptotic}.\footnote{Alternatively, one could call ``deterministic PBT'' just ``PBT'' and for ``probabilistic PBT'' use the term ``heralded PBT'', which is borrowed from quantum optics terminology as used in, e.g.,~\cite{Nielsen2009}. However, we will stick to the widely used terms.}

\subsection{Deterministic PBT}
A protocol for deterministic PBT proceeds as described above, implementing an imperfect simulation of the identity channel whose merit is quantified by the entanglement fidelity $F_d$ or the diamond norm error $\varepsilon_d$.
We denote by $F_d^*(N)$ and $\eps_d^*(N)$ the maximal entanglement fidelity and the minimal diamond norm error for deterministic PBT, respectively, where both the resource state and the POVM are optimized.
We will often refer to this as the \emph{fully optimized} case.

Let $\psi_{{A}^N{B}^N}$ be the entangled resource state used for a PBT protocol.
When using the entanglement fidelity as a figure of merit, it is shown in~\cite{ishizaka2009quantum} that the problem of PBT for the fixed resource state $\psi_{{A}^N{B}^N}$ is equivalent to the state discrimination problem given by the collection of states
\begin{align}\label{eq:statedisc}
   \eta^{(i)}_{A^NB_0}=\mathrm{id}_{B_i\to B_0}\tr_{B_i^c}\proj\psi_{{A}^N{B}^N}, \qquad i=1,\ldots,N.
\end{align}
with uniform prior (here we trace over all $B$ systems but $B_i$, which is relabeled to $B_0$).
More precisely, the success probability $q$ for state discrimination with some fixed POVM $\{E_i\}_{i=1}^N$ and the entanglement fidelity $F_d$ of the PBT protocol with Alice's POVM equal to $\{E_i\}_{i=1}^N$, but acting on $A^NA_0$, are related by the equation $q=\frac{d^2}{N}F_d$.
This link with state discrimination provides us with the machinery developed for state discrimination to optimize the POVM.
In particular, it suggests the use of the pretty good measurement~\cite{Belavkin1975,Hausladen1994}.

As in ordinary teleportation, it is natural to consider PBT protocols where the resource state is fixed to be $N$~maximally entangled states (or EPR pairs) of local dimension~$d$.
This is because EPR pairs are a standard resource in quantum information theory that can easily be produced in a laboratory.
We will denote by~$F_d^{\epr}(N)$ the optimal entanglement fidelity for any protocol for deterministic PBT that uses maximally entangled resource states.
A particular protocol is given by combining maximally entangled resource states with the pretty good measurement (PGM) POVM~\cite{Belavkin1975,Hausladen1994}.
We call this the \emph{standard protocol} for deterministic PBT and denote the corresponding entanglement fidelity by $F^{\std}_d(N)$.
For qubits ($d=2$), the pretty good measurement was shown to be optimal for maximally entangled resource states~\cite{ishizaka2009quantum}:
\begin{align}\label{eq:d-qubits-epr}
  F^{\std}_{2}(N) = F^{\epr}_{2}(N) = 1 - \frac3{4N} + o(1/N).
\end{align}
According to~\cite{Mozrzymas2017}, the PGM is optimal in this situation for $d>2$ as well.

In~\cite{ishizaka2008asymptotic} it is shown that the entanglement fidelity $F^{\std}_d$ for the standard protocol is at least
\begin{align}\label{eq:achievability-EPR-PGM}
F^{\std}_d(N) \ge 1-\frac{d^2-1}{N}.
\end{align}
\textcite{beigi2011simplified} rederived the same bound with different techniques.
In~\cite{Ishizaka2015}, a converse bound is provided in the fully optimized setting:
\begin{align}\label{eq:lowebound-ishi}
F^*_d(N) \le 1-\frac{1}{4(d-1)N^2}+ O(N^{-3}).
\end{align}
Note that the dimension $d$ is part of the denominator instead of the numerator as one might expect in the asymptotic setting.
Thus, the bound lacks the right qualitative behavior for large values of $d$.
A different, incomparable, bound can be obtained from a recent lower bound on the program register dimension of a universal programmable quantum processor obtained by~\textcite{Kubicki2018},
\begin{align}
\varepsilon^*_d(N) \ge 2\left(1-c \frac{\log d}{d}\left(2N+\frac{2}{3}\right)\right),\label{eq:kubicki}
\end{align}
where $c$ is a constant.
By \cref{cor:PBT-FvdG}, this bound is equivalent to
\begin{align}
F^*_d(N)\le c \frac{\log d}{d}\left(2N+\frac{2}{3}\right).
\end{align}
Earlier works on programmable quantum processors~\cite{perez2006optimality,hillery2006approximate} also yield (weaker) converse bounds for PBT.

Interestingly, and of direct relevance to our work, \emph{exact} formulas for the entanglement fidelity have been derived both for the standard protocol and in the fully optimized case.
In~\cite{studzinski2016port}, the authors showed that
\begin{align}\label{eq:cambridge}
F^{\std}_d(N)=d^{-N-2}\sum_{\alpha\vdash_d N-1}\left(\sum_{\mu=\alpha+\square}\sqrt{d_\mu m_{d,\mu}}\right)^2.
\end{align}
Here, the inner sum is taken over all Young diagrams $\mu$ that can be obtained by adding one box to a Young diagram $\alpha\vdash_d N-1$, i.e., a Young diagram with $N-1$ boxes and at most $d$ rows.
\Cref{eq:cambridge} generalizes the result of~\cite{ishizaka2009quantum} for $d=2$, whose asymptotic behavior is stated in \cref{eq:d-qubits-epr}.

In the fully optimized case, \textcite{Mozrzymas2017} obtained a formula similar to \cref{eq:cambridge} in which the dimension $d_\mu m_{d,\mu}$ of the $\mu$-isotypic component in the Schur-Weyl decomposition is weighted by a coefficient $c_\mu$ that is optimized over all probability densities with respect to the Schur-Weyl distribution (defined in \cref{sec:schur weyl dist}).
More precisely,
\begin{align}\label{eq:cambridgeII}
  F^{*}_d(N) =d^{-N-2}\max_{c_\mu}\sum_{\alpha\vdash_d N-1}\left(\sum_{\mu=\alpha+\square}\sqrt{c_\mu d_\mu m_{d,\mu}}\right)^2,
\end{align}
where the optimization is over all nonnegative coefficients $\{c_\mu\}$ such that
$\sum_{\mu\vdash_d N} c_\mu  \frac{d_\mu m_{d,\mu}}{d^N}=1$.

\subsection{Probabilistic PBT}
In the task of probabilistic PBT, Alice's POVM has an additional outcome that indicates the failure of the protocol and occurs with probability $1-p_d$.
For all other outcomes, the protocol is required to simulate the identity channel perfectly.
We call $p_d$ the probability of success of the protocol.
As before, we denote by $p_d^*(N)$ the maximal probability of success for probabilistic PBT using $N$ ports of local dimension $d$, where the resource state as well as the POVM are optimized.

Based on the no-signaling principle and a version of the no-cloning theorem, \textcite{Pitalua-Garcia2013a} showed that the success probability $p^*_{2^n}(N)$ of teleporting an $n$-qubit input state using a general probabilistic PBT protocol is at most
\begin{align}
p^*_{2^n}(N) \leq 1-\frac{4^n-1}{4^n-1+N}.\label{eq:pitalua-garcia}
\end{align}
Subsequently, \textcite{Mozrzymas2017} showed for a general $d$-dimensional input state that the converse bound in \eqref{eq:pitalua-garcia} is also achievable, establishing that
\begin{align}\label{eq:prob-opt}
p^*_d(N)=1-\frac{d^2-1}{d^2-1+N}.
\end{align}
This fully resolves the problem of determining the optimal probability of success for probabilistic PBT in the fully optimized setting.

As discussed above, it is natural to also consider the scenario where the resource state is fixed to be $N$~maximally entangled states of rank~$d$ and consider the optimal POVM given that resource state.
We denote by $p^{\epr}_d$ the corresponding probability of success.
We use the superscript $\epr$ to keep the analogy with the case of deterministic PBT, as the measurement is optimized for the given resource state and no simplified measurement like the PGM is used.
In~\cite{ishizaka2009quantum}, it was shown for qubits ($d=2$) that
\begin{align}
p^{\epr}_{2}(N)=1-\sqrt{\frac{8}{\pi N}} + o(1/\sqrt N).
\end{align}
For arbitrary input dimension $d$, \textcite{studzinski2016port} proved the exact formula
\begin{align}\label{eq:cambridge epr prob}
  p^{\epr}_d(N)=\frac{1}{d^N}\sum_{\alpha\vdash N-1}m_{d,\alpha}^2\frac{d_{\mu^*}}{m_{d,\mu^*}},
\end{align}
where $\mu^*$ is the Young diagram obtained from $\alpha$ by adding one box in such a way that
\begin{align}\label{eq:gamma}
  \gamma_\mu(\alpha)=N\frac{m_{d,\mu} d_\alpha}{m_{d,\alpha} d_\mu}
\end{align}
is maximized (as a function of $\mu$).

Finally, we note that any protocol for probabilistic PBT with success probability~$p_d$ can be converted into a protocol for deterministic PBT by sending over a random port index to Bob whenever Alice's measurement outcome indicates an error.
The entanglement fidelity of the resulting protocol can be bounded as $F_d\ge p_d+\frac{1-p_d}{d^2}$.
When applied to the fully optimized protocol corresponding to \cref{eq:prob-opt}, this yields a protocol for deterministic PBT with better entanglement fidelity than the standard protocol for deterministic PBT.
It uses, however, an optimized resource state that might be difficult to produce, while the standard protocol uses $N$ maximally entangled states.

\subsection{Symmetries}\label{sec:symmetries}
The problem of port-based teleportation has several natural symmetries that can be exploited.
Intuitively, we might expect a $U(d)$-symmetry and a permutation symmetry, since our figures of merit are unitarily invariant and insensitive to the choice of port that Bob has to select.
For the resource state, we might expect an $S_N$-symmetry, while the POVM elements have a marked port, leaving a possible $S_{N-1}$-symmetry among the non-marked ports.
This section is dedicated to making these intuitions precise.

The implications of the symmetries have been known for some time in the community and used in other works on port-based teleportation (e.g.\ in~\cite{Mozrzymas2017}).
We provide a formal treatment here for the convenience of the interested reader as well as to highlight the fact that the unitary symmetry allows us to directly relate the entanglement fidelity (which a priori quantifies an average error) to the diamond norm error (a worst case figure of merit).
This relation is proved in \cref{cor:PBT-FvdG}.

We begin with a lemma on purifications of quantum states with a given group symmetry (see~\cite{Ren05,CKR09} and~\cite[Lemma 5.5]{gross2017schur}):

\begin{lem}\label{lem:sympur}
  Let $\rho_A$ be a quantum state invariant under a unitary representation $\varphi$ of a group $G$, i.e., $[\rho_A,\varphi(g)]=0$ for all $g\in G$.
  Then there exists a purification $\ket{\rho}_{AA'}$ such that $(\varphi(g)\ot\varphi^*(g)) \ket{\rho}_{AA'}=\ket{\rho}_{AA'}$ for all $g\in G$.
   Here, $\varphi^*$ is the dual representation of $\varphi$, which can be written as $\varphi^*(g)=\overline{\varphi(g)}$.
\end{lem}

Starting from an arbitrary port-based teleportation protocol, it is easy to construct a modified protocol that uses a resource state such that Bob's marginal is invariant under the natural action of $S_N$ as well as the diagonal action of $U(d)$.
In slight abuse of notation, we denote by $\zeta_{B^N}$ the unitary representation of $\zeta\in S_N$ that permutes the tensor factors of $\Hi_B^{\ox N}$.

\begin{lem}\label{lem:symsuffice}
  Let $\rho_{A^NB^N}$ be the resource state of a protocol for deterministic PBT with input dimension $d$. Then there exists another protocol for deterministic PBT with resource state $ \rho'_{{A}^N{B}^N}$ such that $ \rho'_{{B}^N}$ is invariant under the above-mentioned group actions,
  \begin{align}
  \begin{aligned}
  U_B^{\otimes N}\rho'_{{B}^N}\left(U_B^{\otimes N}\right)^\dagger &=\rho'_{{B}^N} \quad \text{for all $U_B\in U(d)$,}\\
  \zeta_{B^N} \rho'_{B^N} \zeta_{B^N}^\dagger &= \rho'_{B^N} \quad \text{for all $\zeta\in S_N$,}
  \end{aligned}
  \label{eq:symmetries}
  \end{align}
  and such that the new protocol has diamond norm error and entanglement fidelity no worse than the original one.
\end{lem}
In fact, \cref{lem:symsuffice} applies not only to the diamond norm distance and the entanglement fidelity, but any convex functions on quantum channels that is invariant under conjugation with a unitary channel.
\begin{proof}[Proof of \cref{lem:symsuffice}]
  Define the resource state
  \begin{align}
    \tilde{\rho}_{A^NB^NI}= \frac{1}{N!} \sum_{\zeta\in S_N}\zeta_{B^N}\rho_{A^NB^N}\zeta_{B^N}^\dagger \otimes\proj{\zeta}_I,\label{eq:permutation-symmetrized}
  \end{align}
  where $\zeta_{B^N}$ is the action of $S_N$ on $\hi_B^{\otimes N}$ that permutes the tensor factors, and $I$ is a classical `flag' register with orthonormal basis $\lbrace |\zeta\rangle\rbrace_{\zeta\in S_N}$. The following protocol achieves the same performance as the preexisting one: Alice and Bob start sharing $\tilde{\rho}_{A^NB^NI}$ as an entangled resource, with Bob holding $B^N$ as usual and Alice holding registers $A^NI$. Alice begins by reading the classical register $I$. Suppose that its content is a permutation $\zeta$. She then continues to execute the original protocol, except that she applies $\zeta$ to the index she is supposed to send to Bob after her measurement, which obviously yields the same result as the original protocol.

  A similar argument can be made for the case of $U(d)$. Let $D\subset U(d)$, $|D|<\infty$ be an exact unitary $N$-design, i.e., a subset of the full unitary group such that taking the expectation value of any polynomial $P$ of degree at most $N$ in both $U$ and $U^\dagger$ over the uniform distribution on $D$ yields the same result as taking the expectation of $P$ over the normalized Haar measure on $U(d)$. Such exact $N$-designs exist for all $N$ (\cite{Seymour1984}; see~\cite{Kane2015} for a bound on the size of exact $N$-designs).
  We now define a further modified resource state $\rho'_{A^NB^NIJ}$ from $\tilde{\rho}_{A^NB^NI}$ in analogy to \eqref{eq:permutation-symmetrized}:
  \begin{align}
  \rho'_{A^NB^NIJ} = \frac{1}{|D|} \sum_{U\in D} U_B^{\otimes N} \tilde{\rho}_{A^NB^NI} (U_B^\dagger)^{\otimes N} \otimes |U\rangle\langle U|_J,
  \end{align}
  where $\lbrace |U\rangle\rbrace_{U\in D}$ is an orthonormal basis for the flag register $J$.
  Again, there exists a modified protocol, in which Bob holds the registers $B^N$ as usual, but Alice holds registers $A^NIJ$. Alice starts by reading the register $J$ which records the unitary $U\in D$ that has been applied to Bob's side. She then proceeds with the rest of the protocol after applying $U^\dagger$ to her input state.
  Note that $\rho'_{B^N}$ clearly satisfies the symmetries in \eqref{eq:symmetries}, and furthermore the new PBT protocol using $\rho'_{A^NB^NIJ}$ has the same performance as the original one using $\rho_{A^NB^N}$, concluding the proof.
\end{proof}

Denote by $\Sym^N(\Hi)$ the symmetric subspace of a Hilbert space $\Hi^{\ox N}$, defined by
\begin{align}
\Sym\nolimits^N(\Hi)\coloneqq \lbrace |\psi\rangle\in\Hi^{\ox N}\colon \pi|\psi\rangle = |\psi\rangle \text{ for all }\pi\in S_N\rbrace.
\end{align}
Using the above two lemmas we arrive at the following result.
\begin{prop}\label{prop:symsuffice}
  Let $\rho_{A^NB^N}$ be the resource state of a PBT protocol with input dimension $d$. Then there exists another protocol with properties as in \cref{lem:symsuffice} except that it has a  resource state $\proj\psi_{{A}^N{B}^N}$ with $\ket{\psi}_{{A}^N{B}^N}\in\Sym^{N}(\hi_A\otimes\hi_B)$ that is a purification of a symmetric Werner state, i.e., it is invariant under the action of $U(d)$ on $\Hi_{A}^{\otimes N}\otimes \Hi_{B}^{\otimes N}$ given by $U^{\otimes N}\otimes \overline{U}^{\otimes N}$.
\end{prop}
\begin{proof}
  We begin by transforming the protocol according to \cref{lem:symsuffice}, resulting in a protocol with resource state $\rho'_{A^NB^N I J}$. By \cref{lem:sympur}, there exists a purification $\ket{\psi}_{{A}^N{B}^N}$ of $\rho'_{B^N}$ that is invariant under $U^{\otimes N}\otimes \overline{U}^{\otimes N}$. But Uhlmann's Theorem ensures that there exists an isometry $V_{A^N\to A^N I J E}$ for some Hilbert space $\hi_E$ such that $V_{A^N\to A^N I J E}\ket{\psi}_{{A}^N{B}^N}$ is a purification of $\rho'_{A^NB^N I J}$. The following is a protocol using the resource state $\ket{\psi}$: Alice applies $V$ and discards $E$. Then the transformed protocol from \cref{lem:symsuffice} is performed.
\end{proof}

Using the symmetries of the resource state, we can show that the POVM can be chosen to be symmetric as well.
In the proposition below, we omit identity operators.

\begin{prop}\label{prop:symPOVM}
  Let $\{\left(E_i\right)_{A_0A^N}\}_{i=1}^N$ be Alice's POVM for a PBT protocol with a resource state $\ket \psi$ with the symmetries from \cref{prop:symsuffice}. Then there exists another POVM $\{\left(E'_i\right)_{A_0A^N}\}_{i=1}^N$ such that the following properties hold:
  \begin{enumerate}[label=(\roman*),font=\normalfont]
  \item\label{item:sym perm} $	\zeta_{A^N}\left(E'_i\right)_{A_0A^N}\zeta_{A^N}^\dagger =\left(E'_{\zeta(i)}\right)_{A_0A^N}$ for all $\zeta\in S_N$;
  \item\label{item:sym unit} $\left(U_{A_0}\otimes \overline{U}_{A}^{\otimes N}\right)\left(E'_i\right)_{A_0A^N} \left(U_{A_0}\otimes \overline{U}_{A}^{\otimes N}\right)^\dagger =\left(E'_i\right)_{A_0A^N}$ for all $U\in U(d)$;
  \item\label{item:uni cov} the channel $\Lambda'$ implemented by the PBT protocol is unitarily covariant, i.e.,
  \begin{align}
  \Lambda'_{A_0\to B_0}(X)=U_{B_0} \Lambda'_{A_0\to B_0}(U_{A_0}^\dagger X U_{A_0})U_{B_0}^\dagger\quad\text{for all $U\in U(d)$};
  \end{align}
  \item\label{item:convex} the resulting protocol has diamond norm distance (to the identity channel) and entanglement fidelity no worse than the original one.
  \end{enumerate}
\end{prop}
\begin{proof}
  Define an averaged POVM with elements
  \begin{align}
  \left(E'_i\right)_{A_0A^N}=\int_{U(\hi_A)} \mathrm dU \frac{1}{N!}\sum_{\zeta\in S_N}
    \left(U_{A_0}\ot \overline{U}_{A}^{\ot N}\zeta_{A^N}\right)
    \left(E_{\zeta^{-1}(i)}\right)_{A_0A^N}
    \left(U_{A_0}^\dagger \ot \zeta_{A^N}^\dagger (U_{A}^T)^{\ot N}\right),
  \end{align}
  which clearly has the symmetries~\ref{item:sym perm} and \ref{item:sym unit}.
  The corresponding channel can be written as
  \begin{align}
    \Lambda'_{A_0\to B_0}
  = \int_{U(\mathcal H_A)} \frac1{N!} \sum_{\zeta\in S_N} \Lambda^{(U,\zeta)}_{A_0\to B_0},
  \end{align}
  where
  \begin{align}
  &\Lambda^{(U,\zeta)}_{A_0\to B_0}(X_{A_0}) \\
  &= \sum_{i=1}^N \tr_{A_0A^NB_i^c}\bigl[
      \bigl(
      (U_{A_0}\ot \overline{U}_{A}^{\ot N}\zeta_{A^N})
      (E_{\zeta^{-1}(i)})_{A_0A^N}
      (U_{A_0}^\dagger \ot \zeta_{A^N}^\dagger (U_{A}^T)^{\ot N})
      \bigr)
      \bigl(X_{A_0}\ot\proj\psi_{A^NB^N}\bigr)
    \bigr] \\
  &= \sum_{i=1}^N \tr_{A_0A^NB_i^c}\bigl[
      (E_{\zeta^{-1}(i)})_{A_0A^N}
      \bigl(U_{A_0}^\dagger X_{A_0} U_{A_0} \\
  &\qquad\qquad {} \ot
      (\zeta_{A^N}^\dagger (U_{A}^T)^{\ot N} \ot I_{B^N})
      \proj\psi_{A^NB^N}
      (\overline{U}_{A}^{\ot N}\zeta_{A^N} \ot I_{B^N})
      \bigr)
    \bigr] \\
  &= \sum_{i=1}^N \tr_{A_0A^NB_i^c}\bigl[
      (E_{\zeta^{-1}(i)})_{A_0A^N}
      \bigl(U_{A_0}^\dagger X_{A_0} U_{A_0}\\
  &\qquad\qquad {} \ot
      (I_{A^N} \ot U_B^{\ot N} \zeta_{B^N})
      \proj\psi_{A^NB^N}
      (I_{A^N} \ot \zeta_{B^N}^\dagger (U_B^\dagger)^{\ot N})
      \bigr)
    \bigr] \\
  &= U_{B_0} \Lambda_{A_0\to B_0}(U_{A_0}^\dagger X_{A_0} U_{A_0}) U_{B_0}^\dagger,
  \end{align}
  where we suppressed $\idch_{B_i\to B_0}$.
  Here we used the cyclicity of the trace and the symmetries of the resource state, and $\Lambda_{A_0\to B_0}$ denotes the channel corresponding to the original protocol.
  It follows at once that $\Lambda'_{A_0\to B_0}$ is covariant in the sense of~\ref{item:uni cov}.
  Finally, since the identity channel is itself covariant, property~\ref{item:convex} follows from the concavity (convexity) and unitary covariance of the entanglement fidelity and the diamond norm distance, respectively.
\end{proof}

Similarly as mentioned below \cref{lem:symsuffice}, the statement in \cref{prop:symPOVM}\ref{item:convex} can be generalized to any convex function on the set of quantum channels that is invariant under conjugation with unitary channels.

The unitary covariance allows us to apply a lemma from~\cite{Pirandola2018} (stated as \cref{lem:pirandola} in \Cref{app:technical}) to relate the optimal diamond norm error and entanglement fidelity of port-based teleportation.
This shows that the achievability results \cref{eq:cambridge,eq:d-qubits-epr,eq:achievability-EPR-PGM,eq:lowebound-ishi} for the entanglement fidelity of deterministic PBT, as well as the ones mentioned in the introduction, imply similar results for the diamond norm error without losing a dimension factor.

\begin{cor}\label{cor:PBT-FvdG}
  Let $F_d^*$ and $\varepsilon_d^*$ be the optimal entanglement fidelity and optimal diamond norm error for deterministic PBT with input dimension $d$. Then, $\varepsilon_d^*=2\left(1-F_d^*\right)$.
\end{cor}
Note that the same formula was proven for the standard protocol in~\cite{Pirandola2018}.

\subsection{Representation-theoretic characterization}
The symmetries of PBT enable the use of representation-theoretic results, in particular Schur-Weyl duality.
This was extensively done in~\cite{studzinski2016port,Mozrzymas2017} in order to derive the formulas \cref{eq:cambridge,eq:cambridgeII,eq:cambridge epr prob}.
The main ingredient used in~\cite{studzinski2016port} to derive \cref{eq:cambridge,eq:cambridge epr prob} was the following technical lemma.
For the reader's convenience, we give an elementary proof in \cref{app:single bound} using only Schur-Weyl duality and the classical Pieri rule.
In the statement below, $B_i^c$ denotes the quantum system consisting of all $B$-systems except the $i$-th one.

\newcommand{\restateSingletBound}{%
  The eigenvalues of the operator
  \begin{align}
  T(N)_{AB^N} = \frac1N \left( \phi^+_{AB_1}\otimes I_{B_1^c} + \dots + \phi^+_{AB_N}\otimes I_{B_N^c}  \right)
  \end{align}
  on $(\C^d)^{\otimes(1+N)}$
  are given by the numbers
  \begin{align}
  \frac1{dN} \gamma_\mu(\alpha) = \frac 1 d \frac {d_\alpha m_{d,\mu}} {d_\mu m_{d,\alpha}},
  \end{align}
  where $\alpha\vdash_{d}N-1$, the Young diagram $\mu\vdash_d N$ is obtained from $\alpha$ by adding a single box, and
  $\gamma_\mu(\alpha)$ is defined in \cref{eq:gamma}.
}

\begin{lem}[\cite{studzinski2016port}]\label{lem:singlet bound}
\restateSingletBound
\end{lem}

Note that the formula in \Cref{lem:singlet bound} above gives \emph{all} eigenvalues of $T(N)_{AB^N}$, i.e., including multiplicities.

The connection to deterministic PBT is made via the equivalence with state discrimination. In particular, when using a maximally entangled resource, $T(N)$ is a rescaled version of the density operator corresponding to the ensemble of quantum states $\eta_i$ from \cref{eq:statedisc},
\begin{align}
  T(N) =\frac{d^{N-1}}N\sum_i\eta_i.
\end{align}

Using the hook length formulas \cref{eq:specht hook,eq:stanley hook}, we readily obtain the following simple expression for the ratio $\gamma_\mu(\alpha)$ defined in \cref{eq:gamma}:

\begin{lem}[\cite{MSH18}]\label{lem:mmm}
  Let $\mu=\alpha+e_i$.
  Then,
  \[ \gamma_\mu(\alpha) = \mu_i - i + d = \alpha_i - i + d + 1, \]
  i.e.,
  \[ \frac {d_\alpha m_{d,\mu}} {d_\mu m_{d,\alpha}} = \frac{\alpha_i - i + d + 1} N. \]
\end{lem}
\begin{proof}
  Using \cref{eq:specht hook,eq:stanley hook}, we find
  \begin{align}
    &\quad \gamma_\mu(\alpha)
    = N \frac{m_{d,\mu} d_\alpha}{m_{d,\alpha} d_\mu}
    = N \prod_{(i,j)\in\mu} \frac{d+c(i,j)}{h_\mu(i,j)} \frac{\prod_{(i,j)\in\mu}h_\mu(i,j)}{N!} \frac{(N-1)!}{\prod_{(i,j)\in\alpha}h_\alpha(i,j)} \prod_{(i,j)\in\alpha} \frac{h_\alpha(i,j)}{d+c(i,j)} \\
    &= \prod_{(i,j)\in\mu} \frac{d+c(i,j)}{1} \prod_{(i,j)\in\alpha} \frac{1}{d+c(i,j)}
    = d + c(i, \mu_i) = d + \mu_i - i,
  \end{align}
  which concludes the proof.
\end{proof}

\begin{rem}
It is clear that $\gamma_\mu(\alpha)$ is maximized for $\alpha=(N-1,0,\dots,0)$ and $i=1$.
Therefore,
\begin{align}
\lVert T(N) \rVert_\infty = \frac{N+d-1}{dN}.
\end{align}
This result can be readily used to characterize the extendibility of isotropic states, providing an alternative proof of the result by Johnson and Viola~\cite{johnson2013compatible}.
\end{rem}

\section{The Schur-Weyl distribution}\label{sec:schur weyl dist}
Our results rely on the asymptotics of the Schur-Weyl distribution, a probability distribution defined below in \eqref{eq:schur-weyl-distribution} in terms of the representation-theoretic quantities that appear in the Schur-Weyl duality~\eqref{eq:schur-weyl}.
These asymptotics can be related to the random matrix ensemble $\GUEzd$.
In this section we explain this connection and provide a refinement of a convergence result (stated in \eqref{eq:johansson}) by Johansson~\cite{johansson2001discrete} that is tailored to our applications.
While representation-theoretic techniques have been extensively used in previous analyses, the connection between the Schur-Weyl distribution and random matrix theory has, to the best of our knowledge, not been previously recognized in the context of PBT (see however~\cite{OW17} for applications in the the context of quantum state tomography).

Recalling the Schur-Weyl duality ${\left(\C^d\right)}^{\otimes n} \cong \bigoplus_{\alpha\vdash_d n} [\alpha] \otimes V^d_\alpha$, we denote by $P_\alpha$ the orthogonal projector onto the summand labeled by the Young diagram $\alpha\vdash_d n$.
The collection of these projectors defines a projective measurement, and hence
\begin{align}\label{eq:schur-weyl-distribution}
p_{d,n}(\alpha) \coloneqq \tr\left(P_\alpha \tau_d^{\ot n}\right) = \frac{d_\alpha m_{d,\alpha}}{d^n}
\end{align}
with $\tau_d=\frac{1}{d}1_{\mathbb{C}^d}$ defines a probability distribution on Young diagrams $\alpha\vdash_d n$, known as the \emph{Schur-Weyl distribution}.
Now suppose that $\boldsymbol{\alpha}^{(n)} \sim p_{d,n}$ for $n\in\mathbb{N}$.
By spectrum estimation~\cite{alicki88,duffield90,keyl2001estimating,hayashi2002quantum,Christandl2006}, it is known that
\begin{align}\label{eq:keylwerner}
\frac {\boldsymbol{\alpha}^{(n)}} n \:\xrightarrow{P}\: (\tfrac1d,\dots,\tfrac1d) \quad\text{as $n\to\infty$}.
\end{align}
This can be understood as a \emph{law of large numbers}.
Johansson~\cite{johansson2001discrete} proved a corresponding \emph{central limit theorem}:
Let $\mathbf{A}^{(n)}$ be the centered and renormalized random variable defined by
\begin{align}\label{eq:A-variables}
\mathbf{A}^{(n)} \coloneqq \frac {\boldsymbol{\alpha}^{(n)} - (\tfrac{n}d,\dots,\tfrac{n}d)} {\sqrt{n/d}}.
\end{align}
Then Johansson~\cite{johansson2001discrete} proved that
\begin{align}\label{eq:johansson}
\mathbf{A}^{(n)} \:\xrightarrow{D}\: \operatorname{spec}(\mathbf G)
\end{align}
for $n\to\infty$, where $\mathbf G \sim \operatorname{GUE}^0_d$.
The result for the first row is by Tracy and Widom~\cite{tracy2001distributions} (cf.~\cite{johansson2001discrete,kuperberg1999random}; see~\cite{OW17} for further discussion).

In the following sections, we would like to use this convergence of random variables stated in \cref{eq:keylwerner,eq:johansson} to determine the asymptotics of \cref{eq:cambridge epr prob,eq:cambridge}.
To this end, we rewrite the latter as expectation values of some functions of Young diagrams drawn according to the Schur-Weyl distribution.
However, in order to conclude that these expectation values converge to the corresponding expectation values of functions on the spectrum of $\GUEzd$-matrices, we need a stronger sense of convergence than what is provided by the former results.
Indeed, we need to establish convergence for functions that diverge polynomially as~$n\to \infty$ when $A_{j}=\omega(1)$ or when $A_{j}=O(n^{-1/2})$.\footnote{Here, $f(n) = \omega(g(n))$ means that $|f(n)/g(n)|$ diverges as $n\to\infty$.}
The former are easily handled using the bounds from spectrum estimation~\cite{Christandl2006},
but for the latter a refined bound on $p_{d,n}$ corresponding to small $A$ is needed.
To this end, we prove the following result, which shows convergence of expectation values of a large class of functions that includes all polynomials in the variables $\mathbf A_i$.

In the following, we will need the cone of sum-free non-increasing vectors in $\R^d$,
\begin{align}
	C^d=\left\{x\in\R^d\colon\sum\nolimits_i x_i=0, \ x_i\ge x_{i+1}\right\},
\end{align}
and its interior int$(C^d)$.

\begin{thm}\label{thm:expectation-convergence}
	Let $g\colon\mathrm{int}(C^d)\to \R$ be a continuous function satisfying the following: There exist constants $\eta_{ij}$ satisfying $\eta_{ij}> -2-\frac1{d-1}$ such that for
	\begin{align}
	\varphi_{\eta}(x) \coloneqq \prod_{i<j}\left(x_i-x_j\right)^{\eta_{ij}}
	\end{align}
there exists a polynomial $q$ with
		\begin{align}
			\frac{g(x)}{\varphi_{\eta}(x)}\le q(\|x\|_1).
		\end{align}
	For every $n$, let $\RV\alpha^{(n)} \sim p_{d,n}$ be drawn from the Schur-Weyl distribution, $\RV A^{(n)} \coloneqq \sqrt{d/n}(\RV\alpha^{(n)}-n/d)$ the corresponding centered and renormalized random variable, and $\tilde{\RV A}^{(n)}=\RV A^{(n)}+\frac{d-i}{\sqrt{\frac n d}}$.
	Then the family of random variables $\left\{g\left(\tilde{\RV A}^{(n)}\right) \right\}_{n\in\N}$ is uniformly integrable and
	\begin{align}
	\lim_{n\to\infty} \EE\left[ g\left(\tilde{\RV A}^{(n)}\right) \right] = \EE\left[ g(\RV A) \right],
	\end{align}
	where $\RV A=\mathrm{spec}(\mathbf G)$ and $\mathbf G \sim \GUEzd$.
\end{thm}

As a special case we recover the uniform integrability of the moments of $\mathbf A$ (\cref{cor:uniform-integrability}), which implies convergence in distribution in the case of an absolutely continuous limiting distribution.
Therefore, \cref{thm:expectation-convergence} is a refinement of the result \eqref{eq:johansson} by Johansson.
The remainder of this section is dedicated to proving %
 \cref{thm:expectation-convergence}.

The starting point for what follows is Stirling's approximation, which states that $\sqrt{2\pi} \sqrt n \left( \frac n e \right)^n \leq n!\leq e \sqrt n \left( \frac n e \right)^n$ for all $n\in\N$.
It will be convenient to instead use the following variant,
\begin{align}
  \frac{\sqrt{2\pi}}e \sqrt{n+1} \left( \frac n e \right)^n \leq n!&\leq e \sqrt n \left( \frac n e \right)^n,
  \label{eq:stirling two}
\end{align}
where the upper bound is unchanged and the lower bound follows using $n!=\frac{(n+1)!}{n+1}$.
 The dimension $d_\alpha$ is equal to the multinomial coefficient up to inverse polynomial factors~\cite{Christandl2006}.
Defining the normalized Young diagram $\bar\alpha=\frac{\alpha}{n}$ for $\alpha\vdash n$, the multinomial coefficient~$\binom{n}{\alpha}$ can be bounded from above using \cref{eq:stirling two} as
\begin{align}
  \binom n\alpha
= \frac{n!}{\alpha_1!\dots\alpha_d!}
\leq C_d \sqrt{\frac n {\prod_{i=1}^d (\alpha_i+1)}} \, \frac{n^n} {\alpha_1^{\alpha_1}\dots\alpha_d^{\alpha_d}},
\end{align}
where $C_d \coloneqq \frac {e^{d+1}} {(2\pi)^{d/2}}$.
Hence,
\begin{align}
  d^{-n} \binom n\alpha
&\leq C_d\sqrt{\frac n {\prod_{i=1}^d (\alpha_i+1)}} \exp\left(-n D(\bar\alpha\Vert\tau)\right)\\
&\leq C_d\sqrt{\frac n {\prod_{i=1}^d (\alpha_i+1)}} \exp\left(-\frac n2 \left\lVert \bar\alpha - \tau\right\rVert_1^2\right)
\\
&= C_dn^{-\frac{d-1}2} \left[ \prod_{i=1}^d \left( \bar\alpha_i+\frac1n \right)^{-\frac12} \right] \exp\left(-\frac n2 \left\lVert \bar\alpha - \tau\right\rVert_1^2\right).\label{eq:multinomialbound}
\end{align}
Here, $D(p\|q)\coloneqq\sum_i p_i \log{p_i}/\!{q_i}$ is the Kullback-Leibler divergence defined in terms of the natural logarithm, $\tau=(1/d,\ldots,1/d)$ is the uniform distribution, and we used Pinsker's inequality~\cite{Pinsker1960} in the second step.

We go on to derive an upper bound on the probability of Young diagrams that are close to the boundary of the set of Young diagrams under the Schur-Weyl distribution. More precisely, the following lemma can be used to bound the probability of Young diagrams that have two rows that differ by less than the generic $O(\sqrt{n})$ in length.
\begin{lem}\label{lem:differences-bounded}
  Let $d\in\N$ and $c_1,\dots,c_{d-1}\geq0$, $\gamma_1,\dots,\gamma_{d-1}\geq0$.
  Let $\alpha\vdash_dn$ be a Young diagram with {\normalfont (a)} $\alpha_i-\alpha_{i+1}\leq c_i n^{\gamma_i}$ for all $i$.
  Finally, set $A \coloneqq \sqrt{d/n}(\alpha-n/d)$.
  Then,
  \begin{align}
    p_{d,n}(\alpha) \leq C n^{-\frac{d^2-1}{2}+2\sum_{i<j}\gamma_{ij}} \left[ \prod_{i=1}^{d}\left(1+\sqrt{\frac{d}{n}}A_i+\frac dn\right)^{i-d-\frac12} \right] \exp\left(-\frac{1}{2d} \left\|A\right\|_1^2\right),
  \end{align}
  where $\gamma_{ij}\coloneqq\max\{\gamma_i,\gamma_{i+1},\ldots,\gamma_{j-1}\}$ and $C=C(c_1,\dots,c_{d-1},d)$ is a suitable constant.
\end{lem}
\begin{proof}
  We need to bound $p_{d,n}(\alpha)=m_{d,\alpha} d_\alpha / d^n$ and begin with $m_{d,\alpha}$.
  By assumption (a), there exist constants $C_{ij}>0$ (depending on~$c_i,\dots,c_{j-1}$ as well as on~$d$) such that the inequality $\alpha_i-\alpha_j+j-i\le C_{ij}n^{\gamma_{ij}}$ holds for all $i<j$.
  Using the Weyl dimension formula~\eqref{eq:weyl dim} and assumption~(a), it follows that
  \begin{align}
    m_{d,\alpha}&= \prod_{i<j}\frac{\alpha_i-\alpha_j+j-i}{j-i} \leq C_1 \, n^{\sum_{i<j}\gamma_{ij}}
    \label{eq:from weyl}
  \end{align}
  for a suitable constant $C_1=C_1(c_1,\dots,c_{d-1},d)>0$.
  Next, consider $d_\alpha$.
  By comparing the hook-length formulas~\eqref{eq:specht hook} and \eqref{eq:stanley hook}, we have
  \begin{align}
    d_\alpha
  &= n! \, m_{d,\alpha}\left[ \prod_{(i,j)\in\alpha}(d+j-i) \right]^{-1}\\
  &= n! \, m_{d,\alpha}\left[ \prod_{i=1}^d\frac{(\alpha_i+d-i)!}{(d-i)!} \right]^{-1}\\
  &\leq n! \, m_{d,\alpha}\left[ \prod_{i=1}^d\frac{(\alpha_i+1)^{d-i} \alpha_i!}{(d-i)!} \right]^{-1}\\
  &= m_{d,\alpha} \left[ \prod_{i=1}^d\frac{(d-i)!}{(\alpha_i+1)^{d-i}} \right] \binom n \alpha\\
  &= C_2 \, m_{d,\alpha} n^{-\frac{d(d-1)}2} \left[ \prod_{i=1}^d \left( \bar\alpha_i + \frac1n \right)^{i-d} \right] \binom n \alpha,
  \label{eq:from hook}
  \end{align}
  where $C_2=C_2(d)>0$, and $\bar\alpha_i = \alpha_i/n$.
  In the inequality, we used that $\alpha_i + d - i \geq \alpha_i + 1$ for $1\leq i\leq d-1$, and for $i=d$, the exponent of $\alpha_i + 1$ on the right hand side is zero.

  Combining \cref{eq:from weyl,eq:from hook,eq:multinomialbound} and setting $C_3=C_1^2 C_2C_d$, we obtain
  \begin{align}
    p_{d,n}(\alpha)
  = \frac {m_{d,\alpha} d_\alpha}{d^n}
  &\leq C_2 m_{d,\alpha}^2 n^{-\frac{d(d-1)}2} \left[ \prod_{i=1}^d \left(\bar\alpha_i+\frac1n\right)^{i-d} \right] d^{-n} \binom n \alpha \\
  &\leq C_3 \, n^{-\frac{d^2-1}2+2\sum_{i<j}\gamma_{ij}} \left[ \prod_{i=1}^d \left(\bar\alpha_i+\frac1n\right)^{i-d-\frac12} \right] \exp\left(-\frac n2 \left\lVert \bar\alpha - \tau\right\rVert_1^2\right).
  \end{align}
  Substituting $\bar\alpha_i = \frac1d + \frac{A_i}{\sqrt{nd}}$ we obtain the desired bound.
\end{proof}

In order to derive the asymptotics of entanglement fidelities for port-based teleportation, we need to compute limits of certain expectation values.
As a first step, the following lemma ensures that the corresponding sequences of random variables are uniformly integrable.
We recall that a family of random variables $\{\RV X^{(n)}\}_{n\in\N}$ is called \emph{uniformly integrable} if, for every $\eps>0$, there exists $K<\infty$ such that $\sup_n \EE\left[\lvert\RV X^{(n)}\rvert \, \mathds1_{\lvert\RV X^{(n)}\rvert \geq K}\right]\leq\eps$.

\begin{lem}\label{lem:UI}
Under the same conditions as for \cref{thm:expectation-convergence}, the family of random variables $\left\{ g\left(\tilde{\RV A}^{(n)}\right) \right\}_{n\in\N}$ is uniformly integrable.
\end{lem}

\begin{proof}
Let $\RV X^{(n)} \coloneqq g\left(\tilde{\RV A}^{(n)}\right)$.
The claimed uniform integrability follows if we can show that
\begin{align}\label{eq:uniform boundedness}
  \sup_n \EE\left[\lvert\RV X^{(n)}\rvert\right]<\infty
\end{align}
for every choice of the $\eta_{ij}$.
Indeed, to show that $\{ \RV X^{(n)} \}$ is uniformly integrable it suffices to show that $\sup_n \EE\left[\lvert \RV X^{(n)} \rvert^{1+\delta}\right]<\infty$ for some $\delta>0$~\cite[Ex.~5.5.1]{Dur10}.
If we choose $\delta>0$ such that $\eta'_{ij} \coloneqq (1+\delta)\eta_{ij} > -2-\frac{1}{d-1}$ for all $1\leq i < j \leq d$, then it is clear that \cref{eq:uniform boundedness} for $\eta'_{ij}$ implies uniform integrability of the original family.

Moreover, we may also assume that $h_{\eta}\equiv g/\varphi_{\eta}=1$, since the general case then follows from the fact that $p_{d,n}(\alpha)$ decays exponentially in $\lVert A\rVert_1$ (see \cref{lem:differences-bounded}). More precisely, for any polynomial $r$ and any constant $\theta_1>0$ there exist constants $\theta_2, \theta_3>0$ such that
\begin{align}
  r(\|x\|_2)\exp\left(-\theta_1\|x\|_1\right)\le \theta_2\exp\left(-\theta_3\|x\|_1\right).
\end{align}
In particular, this holds for the polynomial $q$ bounding $h$ from above by assumption.
When proving the statement $\sup_n \EE\left[\lvert \RV X^{(n)} \rvert\right]<\infty$, the argument above allows us to reduce the general case $h_{\eta} = g/\varphi_{\eta}\neq 1$ to the case $h_{\eta}=1$, or equivalently, to
\begin{align}
g(x) = \varphi_{\eta}(x) = \prod_{i<j}\left(x_i-x_j\right)^{\eta_{ij}}.
\end{align}

Thus, it remains to be shown that
\begin{align}\label{eq:uniform boundedness simplified}
  \sup_n \EE\left[f^{(n)}(\RV A^{(n)})\right]<\infty,
\end{align}
where
\begin{align}
  f^{(n)}(A) \coloneqq \varphi_{\eta}(\tilde{A})=\prod_{i<j}\left(A_i-A_j+\frac{j-i}{\sqrt{n/d}}\right)^{\eta_{ij}},
\end{align}
and we fix an admissible $\eta$ for the rest of this proof.
Define $\Gamma_{ij}\coloneqq A_i-A_j+\frac{j-i}{\sqrt{n/d}}$.
Then we have $f^{(n)}(A) = \prod_{i<j} \Gamma_{ij}^{\eta_{ij}}$, while the Weyl dimension formula~\eqref{eq:weyl dim} becomes
\begin{align}
m_{d,\alpha}
= \left( \frac nd \right)^{\frac{d(d-1)}4} \prod_{i<j} \frac {\Gamma_{ij}} {j-i}.
\end{align}
Together with \cref{eq:from hook,eq:multinomialbound} we obtain
\begin{align}
p_{d,n}(\alpha) \, f^{(n)}(A)
&\leq C n^{-\frac{d-1}2} \left( \prod_{i<j} \Gamma_{ij}^{2+\eta_{ij}} \right) \left[ \prod_{i=1}^d \left( 1 + \sqrt{\frac dn}A_i + \frac dn \right)^{i-d-\frac12} \right] \exp\left(-\frac 1{2d} \left\lVert A\right\rVert_1^2\right)
\label{eq:pf bound}
\end{align}
for some $C=C(d)$.

We now want to bound the expectation value in \cref{eq:uniform boundedness simplified} and begin by splitting the sum over Young diagrams according to whether $\exists i: |A_i|>n^\eps$ for some $\eps\in(0,\frac12)$ to be determined later, or $|A_i|\leq n^\eps$ for all $i$.
We denote the former event by $\mathcal E$ and obtain
\begin{align}\label{eq:expectation-split}
  \EE\left[f^{(n)}(\RV A^{(n)})\right]
= \EE\left[f^{(n)}(\RV A^{(n)}) \mathds 1_\mathcal E\right]
+ \EE\left[f^{(n)}(\RV A^{(n)}) \mathds 1_{\mathcal E^c}\right].
\end{align}
We treat the two expectation values in~\eqref{eq:expectation-split} separately and begin with the first one.
If $|A_i|>n^\eps$ for some $i$, then $\lVert A\rVert_1^2 \geq n^{2\eps}$, so it follows by \cref{eq:pf bound} that
\begin{align}
&\EE\left[f^{(n)}(\RV A^{(n)}) \mathds 1_\mathcal E\right]\\
&\qquad= \sum_{\substack{\alpha\vdash_d n \text{ s.t.~} \\ \exists i: \lvert A_i\rvert > n^\eps}} p_{d,n}(\alpha) f^{(n)}(A) \\
&\qquad\leq C \sum_{\substack{\alpha\vdash_d n \text{ s.t.~} \\ \exists i: \lvert A_i\rvert > n^\eps}} n^{-\frac{d-1}2} \left( \prod_{i<j} \Gamma_{ij}^{2+\eta_{ij}} \right) \left[ \prod_{i=1}^d \left( 1 + \sqrt{\frac dn}A_i + \frac dn \right)^{i-d-\frac12} \right] \exp\left(-\frac 1{2d} \left\lVert A\right\rVert_1^2\right) \\
&\qquad\leq \poly(n) \exp\left(-\frac 1{2d} n^{2\eps}\right).
\end{align}
Here, $\poly(n)$ denotes some polynomial in $n$ and we also used that, for fixed $d$, the number of Young diagrams is polynomial in~$n$.
This shows that the first expectation value in~\eqref{eq:expectation-split} vanishes for $n\to\infty$.

For the second expectation value, note that $|A_i|\leq n^{\varepsilon}=o(\sqrt n)$ for all $i$, and hence there exists a constant $K>0$ such that we have
\begin{align}\label{eq:K bound}
  \prod_{i=1}^d \left( 1 + \sqrt{\frac dn}A_i + \frac dn \right)^{i-d-\frac12} \leq K.
\end{align}
Using \cref{eq:pf bound,eq:K bound}, we can therefore bound
\begin{align}
\EE\left[f^{(n)}(\RV A^{(n)}) \mathds 1_{\mathcal E^c}\right]
&= \sum_{\substack{\alpha\vdash_d n \text{ s.t.~} \\ \forall i: \lvert A_i\rvert \leq n^\eps}} p_{d,n}(\alpha) f^{(n)}(A) \\
 &\leq C K \sum_{A\in\mathcal D_n} n^{-\frac{d-1}2} \left( \prod_{i<j} \Gamma_{ij}^{2+\eta_{ij}} \right) \exp\left(-\frac 1{2d} \left\lVert A\right\rVert_1^2\right),
\label{eq:improper riemann sum}
\end{align}
where we have introduced $\mathcal D_n \coloneqq \{ A\colon \alpha\vdash_d n \}$. The summands are nonnegative, even when evaluated on any point in $\hat{\mathcal D}_n \coloneqq \left\{ A \in \sqrt{\frac{d}{n}}\Z^d \colon \sum_i A_i = 0, A_i \geq A_{i+1} \forall i \right\}\supset\mathcal D_n$, so that we have the upper bound
\begin{align}
\EE\left[f^{(n)}(\RV A^{(n)}) \mathds 1_{\mathcal E^c}\right]
&\leq C K \sum_{A\in\mathcal D_n} n^{-\frac{d-1}2} \left( \prod_{i<j} \Gamma_{ij}^{2+\eta_{ij}} \right) \exp\left(-\frac 1{2d} \left\lVert A\right\rVert_1^2\right)\\
&\leq C K \sum_{A\in\hat{\mathcal D}_n} n^{-\frac{d-1}2} \left( \prod_{i<j} \Gamma_{ij}^{2+\eta_{ij}} \right) \exp\left(-\frac 1{2d} \left\lVert A\right\rVert_1^2\right).
\label{eq:improper riemann sum extended}
\end{align}
Let $x_i=A_i-A_{i+1}, \ i=1,\ldots,d-1$.
Next, we will upper bound the exponential in \cref{eq:improper riemann sum extended}.
For this, define $\tilde x_i=\max(\frac1{d-1},x_i)$ and let $S=\{ i\in\{1,\dots,d-1\} \;|\; x_i\le \frac1{d-1} \}$.
Then, assuming $S^c\neq\emptyset$,
\begin{align}
\sum_{i=1}^{d-1}\tilde x_i
&\le\left(\sum_{i=1}^{d-1}\tilde x_i\right)^2
= \left(\sum_{i\in S}\tilde x_i+\sum_{i\in S^c}\tilde x_i\right)^2
= \left(\frac{|S|}{d-1} + \sum_{i\in S^c} x_i\right)^2
\\&=\left( \frac{|S|}{d-1} \right)^2 + 2 \frac{|S|}{d-1} \left( \sum_{i\in S^c} x_i \right) + \left( \sum_{i\in S^c} x_i \right)^2
\\&\leq \left( \frac{|S|}{d-1} \right)^2 + 2 \frac{|S|}{d-1} \frac{d-1}{|S^c|} \left( \sum_{i\in S^c} x_i \right)^2 + \left( \sum_{i\in S^c} x_i \right)^2
\\&= \left( \frac{|S|}{d-1} \right)^2 + \left(1 + 2 \frac{|S|}{|S^c|} \right) \left( \sum_{i\in S^c} x_i \right)^2
\\&\leq 1 + \left(2d - 1 \right) \left( \sum_{i=1}^{d-1} x_i \right)^2
\end{align}
since $\sum_{i\in S^c} x_i \geq \frac{|S^c|}{d-1}$.
This bounds also holds when $S^c=\emptyset$.
Hence,
\begin{align}\label{eq:exp split}
\exp\left(-\frac 1{2d} \left\lVert A\right\rVert_1^2\right)
\le\exp\left(-\frac 1{2d}\left(\sum_{i=1}^{d-1}x_i\right)^2\right)
\le R\exp\left(-\gamma \sum_{i=1}^{d-1}\tilde x_i \right)
=R\prod_{i=1}^{d-1}\exp\left(-\gamma \tilde x_i\right),
\end{align}
where $\gamma := \frac1{2d(2d-1)}$ and $R := e^{-\gamma}$.
The first inequality follows from $\sum_{i=1}^{d-1}x_i=A_1-A_d=|A_1|+|A_d| \leq \lVert A\rVert_1$.
If we use \cref{eq:exp split} in \cref{eq:improper riemann sum extended} we obtain the upper bound
\begin{align}\label{eq:exp is now split}
  \EE\left[f^{(n)}(\RV A^{(n)}) \mathds 1_{\mathcal E^c}\right]
  \leq
  C' \sum_{A\in\hat{\mathcal D}_n} n^{-\frac{d-1}2} \left( \prod_{i<j} \Gamma_{ij}^{2+\eta_{ij}} \right) \prod_{i=1}^{d-1}\exp\left(-\gamma \tilde x_i\right)
\end{align}
where $C' := CKR$.

Let us first assume that all $\eta_{ij} \leq -2$, so that $2+\eta_{ij}\in(-\frac{1}{d-1},0]$.
Since
\begin{align}\label{eq:bounderamos}
  \Gamma_{ij}
  = \left(\sum_{l=i}^{j-1}x_l\right)+\frac{j-i}{\sqrt{\frac{n}{d}}}
  = \sum_{l=i}^{j-1}\left(x_l+\frac{1}{\sqrt{\frac{n}{d}}}\right)
  \geq x_{i}+\frac{1}{\sqrt{\frac{n}{d}}}
\end{align}
and $\eta_{ij}+2\leq0$, we have that
\begin{align}\label{eq:Gammabound}
\Gamma_{ij}^{2+\eta_{ij}}\le\left(x_{i}+\frac{1}{\sqrt{\frac{n}{d}}}\right)^{2+\eta_{ij}},
\end{align}
as power functions with non-positive exponent are non-increasing.
We can then upper-bound \cref{eq:exp is now split} as follows,
\begin{align}
  \EE\left[f^{(n)}(\RV A^{(n)}) \mathds 1_{\mathcal E^c}\right]
&\leq C' \sum_{A\in\hat{\mathcal D}_n} n^{-\frac{d-1}2} \left( \prod_{i<j} \Gamma_{ij}^{2+\eta_{ij}} \right) \prod_{i=1}^{d-1}\exp\left(-\gamma \tilde x_i\right)
\\&\leq C'\sum_{A\in\hat{\mathcal D}_n} n^{-\frac{d-1}2} \left( \prod_{i<j} \left(x_{i}+\frac{1}{\sqrt{\frac{n}{d}}}\right)^{2+\eta_{ij}} \right) \prod_{i=1}^{d-1}\exp\left(-\gamma \tilde x_i\right),
\\&= C'\sum_{A\in\hat{\mathcal D}_n} n^{-\frac{d-1}2} \left( \prod_{i=1}^{d-1} \left(x_{i}+\frac{1}{\sqrt{\frac{n}{d}}}\right)^{\sum_{j=i+1}^d(2+\eta_{ij})} \right) \prod_{i=1}^{d-1}\exp\left(-\gamma \tilde x_i\right),
\\&= C'\prod_{i=1}^{d-1}\left(n^{-\frac{1}2} \sum_{x_i\in\sqrt{\frac d n}\N} \left(x_{i}+\frac{1}{\sqrt{\frac{n}{d}}}\right)^{\sum_{j=i+1}^d(2+\eta_{ij})}  \exp\left(-\gamma \tilde x_i\right)\right)
\end{align}
where the first inequality is \cref{eq:exp is now split} and in the second inequality we used \cref{eq:Gammabound}.
Since $\eta_{ij} > -2-\frac1{d-1}$ by assumption, it follows that $\sum_{j=i+1}^d (2+\eta_{ij}) > -\frac{d-i}{d-1} \geq -1$.
Thus, each term in the product is a Riemann sum for an improper Riemann integral, as in \cref{lem:Riemann-madness}, which then shows that the expression converges for $n\to\infty$.

The case where some $\eta_{ij}>-2$ is treated by observing that
\begin{align}
\prod_{\substack{i<j:\\ \eta_{ij}>-2}}\Gamma_{ij}^{2+\eta_{ij}}\exp\left(-\frac 1{2d} \left\lVert A\right\rVert_1^2\right)\le c_1\exp\left(-\frac{c_2}{2d} \left\lVert A\right\rVert_1^2\right)
\end{align}
for suitable constants $c_1,c_2>0$.
We can use this bound in \cref{eq:exp is now split} to replace each $\eta_{ij}>-2$ by $\eta_{ij}=-2$, at the expense of modifying the constants $C'$ and $\gamma$, and then proceed as we did before.
This concludes the proof of \cref{eq:uniform boundedness simplified}.
\end{proof}

The uniform integrability result of \cref{lem:UI} implies that the corresponding expectation values converge.
To determine their limit in terms of the expectation value of a function of the spectrum of a $\GUEzd$-matrix, however, we need to show that we can take the limit of the dependencies on~$n$ of the function and the random variable $\mathbf A^{(n)}$ separately.
This is proved in the following lemma, where we denote the interior of a set $E$ by int$(E)$.
\begin{lem}\label{lem:expectation-convergence-generic}
  Let $\lbrace\mathbf A^{(n)}\rbrace_{n\in\N}$ and $\mathbf A$ be random variables on a Borel measure space $E$ such that $\mathbf{A}^{(n)}\overset{D}{\to}\mathbf A$ for $n\to\infty$ and $\mathbf A$ is absolutely continuous.
  Let $f\colon\mathrm{int}(E)\to \R$. Let further $f_n: E\to \R$, $n\in\N$, be a sequence of continuous bounded functions such that $f_n\to f$ pointwise on $\mathrm{int}(E)$ and, for any compact $S\subset \mathrm{int}(E)$, $\{f_n|_S\}_{n\in\N}$ is uniformly equicontinuous and $f_n|_S\to f|_S$ uniformly. Then for any such compact $S\subset \mathrm{int}(E)$, the expectation value $\mathbb E\left[f(\mathbf A)\mathds 1_S(\mathbf A)\right]$ exists and
  \begin{align}
    \lim_{n\to\infty}\mathbb E\left[f_n(\mathbf{A}^{(n)})\mathds 1_S(\mathbf{A}^{(n)})\right]=\mathbb E\left[f(\mathbf A)\mathds 1_S(\mathbf A)\right].
  \end{align}
\end{lem}
\begin{proof}
  For $n,m\in\N\cup\{\infty\}$, define
  \begin{align}
    b_{nm}(S)=\mathbb E\left[f_n(\mathbf{A}^{(m)})\mathds 1_S(\mathbf A^{(m)})\right]
  \end{align}
  with $f_\infty\coloneqq f$, $\mathbf{A}^{(\infty)}\coloneqq \mathbf A$ and  $S\subset \mathrm{int}(E)$ compact. These expectation values readily exist as $f_n$ is bounded for all $n$, and the uniform convergence of $f_n|_S $ implies that $f|_S$ is continuous and bounded as well. The uniform convergence $f_n|_S\to f|_S$ implies that $f_n|_S$ is uniformly bounded, so by Lebesgue's theorem of dominated convergence $b_{\infty m}(S)$ exists for all $m\in\N$ and
  \begin{align}\label{eq:lim1}
    \lim_{n\to\infty}b_{nm}(S)=b_{\infty m}(S)\ \forall m\in\N\cup\{\infty\}.
  \end{align}
  This convergence is even uniform in $m$ which follows directly from the uniform convergence of $f_n|_S$. The sequence $\lbrace\mathbf{A}^{(n)}\rbrace_{n\in\mathbb{N}}$ of random variables converges in distribution to the absolutely continuous $\mathbf A$, so the expectation value of any continuous bounded function converges.
  Therefore,
  \begin{align}\label{eq:lim2}
    \lim_{m\to\infty}b_{nm}(S)=b_{n\infty}(S)\ \forall n\in\N\cup\{\infty\}.
  \end{align}
  An inspection of the proof of Theorem 1, Chapter VIII in~\cite{Feller2008} reveals the following: The fact that the uniform continuity and boundedness of $f_n|_S$ hold uniformly in $n$ implies the uniformity of the above limit. Moreover, since both limits exist and are uniform, this implies that they are equal to each other, and any limit of the form
  \begin{align}
    \lim_{n\to\infty} b_{n m(n)}
  \end{align}
  for $m(n)\xrightarrow{n\to\infty}\infty$ exists and is equal to the limits in \cref{eq:lim1,eq:lim2}.
\end{proof}

Finally, we obtain the desired convergence theorem.
For our applications, $\eta_{ij}\equiv-2$ suffices.
The range of $\eta_{ij}$'s for which the lemma is proven is naturally given by the proof technique.

\begin{proof}[Proof of \cref{thm:expectation-convergence}]
  The uniform integrability of $\mathbf X^{(n)}\coloneqq g\left(\tilde{\mathbf A}^{(n)}\right) $ is the content of \cref{lem:UI}.
  Recall that uniform integrability means that
  \begin{align}
  \lim_{K\to\infty} \sup_{n\in\mathbb{N}} \mathbb{E}\left[\big|\mathbf{X}^{(n)}\big| \mathds{1}_{\cE_K^c}\left(\mathbf{A}^{(n)}\right)\right] = 0,
  \end{align}
  where $\cE_K\coloneqq \lbrace x\in\mathbb{R}^d\colon \|x\|_\infty \leq K\rbrace$.
  Let now $\eps>0$ be arbitrary, and $K<\infty$ be such that the following conditions are true:
  \begin{align}
  \sup_{n\in\mathbb{N}}\mathbb{E}\left[\big|\mathbf{X}^{(n)}\big| \mathds{1}_{\cE_K^c}\left(\mathbf{A}^{(n)}\right)\right] &\leq \frac{\eps}{3} &
  \mathbb{E}\left[g(\RV A)\mathds{1}_{\cE_K^c}\left(\mathbf{A}\right)\right] &\leq \frac{\eps}{3},
  \end{align}
  where $\mathbf{A}$ is distributed as the spectrum of a $\GUEzd$ matrix according to \cref{eq:johansson}.
  By \cref{lem:expectation-convergence-generic}, $\lim_{n\to\infty}\mathbb{E}\left[\mathbf{X}^{(n)}\mathds{1}_{\cE_K}\left(\mathbf{A}^{(n)}\right)\right] = \mathbb{E}\left[g(\RV A)\mathds{1}_{\cE_K}\left(\mathbf{A}\right)\right]$.
  Thus, we can choose $n_0\in\mathbb{N}$ such that for all $n\geq n_0$,
  \begin{align}
  \left|\mathbb{E}\left[\mathbf{X}^{(n)}\mathds{1}_{\cE_K}\left(\mathbf{A}^{(n)}\right)\right] - \mathbb{E}\left[g(\RV A)\mathds{1}_{\cE_K}\left(\mathbf{A}\right)\right] \right| \leq \frac{\eps}{3}.
  \end{align}
  Using the above choices, we then have
  \begin{align}
  \left| \mathbb{E}\big[\mathbf{X}^{(n)}\big] - \mathbb{E}\left[g(\RV A)\right] \right|&\leq \mathbb{E}\left[\big|\mathbf{X}^{(n)}\big| \mathds{1}_{\cE_K^c}\left(\mathbf{A}^{(n)}\right)\right]+ |\mathbb{E}\left[g(\RV A)\mathds{1}_{\cE_K^c}\left(\mathbf{A}\right)\right]|\\
  &+ \left|\mathbb{E}\left[\mathbf{X}^{(n)}\mathds{1}_{\cE_K}\left(\mathbf{A}^{(n)}\right)\right] - \mathbb{E}\left[g(\RV A)\mathds{1}_{\cE_K}\left(\mathbf{A}\right)\right] \right|
  \leq \eps
  \end{align}
  for all $n\geq n_0$, proving the desired convergence of the expectation values.
\end{proof}

From \cref{thm:expectation-convergence} we immediately obtain the following corollary about uniform integrability of the moments of $\mathbf{A}$.

\begin{cor}\label{cor:uniform-integrability}
  Let $k\in\N$, let $j\in\{1,\dots,d\}$, and, for every $n$, let $\RV A^{(n)}$ be the random vector defined in \eqref{eq:A-variables}.
  Then, the sequence of $k$-th moments $\big\lbrace ( \RV A^{(n)}_j)^k \big\rbrace_{n\in\mathbb{N}}$ is uniformly integrable and $\lim_{n\to\infty} \EE\bigl[(\RV A^{(n)}_j)^k\bigr] = \EE[\RV A_j^k]$, where $\RV A \sim \GUEzd$.
\end{cor}

\section{Probabilistic PBT}\label{sec:prob}
Our goal in this section is to determine the asymptotics of $p^{\epr}_d$ using the formula \eqref{eq:cambridge epr prob} and exploiting our convergence theorem, \cref{thm:expectation-convergence}.
The main result is the following theorem stated in \Cref{sec:summary}, which we restate here for convenience.
\begin{thm-probabilistic}[restated]
	\restateProbabilistic
\end{thm-probabilistic}
Previously, such a result was only known for $d=2$ following from an exact formula for~$p^{\epr}_2(N)$ derived in~\cite{ishizaka2009quantum}.
We show in \cref{lem:maxboltz} in \cref{app:maxboltz} that, for $d=2$, $\mathbb E[\lambda_{\max}(\mathbf G)] = \frac{2}{\sqrt{\pi}}$, hence rederiving the asymptotics from~\cite{ishizaka2009quantum}.

While \cref{thm:prob-asymptotics} characterizes the limiting behavior of $p^{\epr}$ for large $N$, it contains the constant $\mathbb E[\lambda_{\max}(\mathbf G)]$, which depends on $d$. As $\mathbb{E}[\mathbf M]=0$ for $\mathbf M\sim \GUEd$, it suffices to analyze the expected largest eigenvalue for $\GUEd$. The famous Wigner semicircle law~\cite{Wigner1993} implies immediately that
\begin{align}
  \lim_{d\to\infty}\frac{\mathbb E[\lambda_{\max}(\mathbf G)]}{\sqrt d}=2,
\end{align}
but meanwhile the distribution of the maximal eigenvalue has been characterized in a much more fine-grained manner.
In particular, according to \cite{ledoux2007deviation}, there exist constants $C$ and $C'$ such that the expectation value of the maximal eigenvalue satisfies the inequalities
\begin{equation}
	1-\frac{1}{C'd^{\frac 2 3}}\le \frac{\mathbb E[\lambda_{\max}(\mathbf G)]}{2\sqrt d}\le 1-\frac{1}{Cd^{\frac 2 3}}.
\end{equation}

This also manifestly reconciles \cref{thm:prob-asymptotics} with the fact that teleportation needs at least $2 \log d$ bits of classical communication (see \cref{sec:converse}), since the amount of classical communication in a port-based teleportation protocol consists of $\log N$ bits.

\begin{proof}[Proof of \cref{thm:prob-asymptotics}]
We start with \cref{eq:cambridge epr prob}, which was derived in~\cite{studzinski2016port}, and which we restate here for convenience:
\begin{align}
  p^{\epr}_d(N)
= \frac{1}{d^N}\sum_{\alpha\vdash N-1}m_{d,\alpha}^2\frac{d_{\mu^*}}{m_{d,\mu^*}},
\end{align}
where $\mu^*$ is the Young diagram obtained from $\alpha\vdash N-1$ by adding one box such that $\gamma_\mu(\alpha) = N\frac{m_{d,\mu} d_\alpha}{m_{d,\alpha},d_\mu}$ is maximal.
By \cref{lem:mmm}, we have $\gamma_\mu(\alpha) = \alpha_i - i + d + 1$ for $\mu=\alpha+e_i$.
This is maximal if we choose $i=1$, resulting in $\gamma_{\mu^*}(\alpha) = \alpha_1+d$.
We therefore obtain:
\begin{align}
  p^{\epr}_d(N) &= \frac{1}{d^N}\sum_{\alpha\vdash N-1}m_{d,\alpha} d_\alpha\frac{m_{d,\alpha} d_{\mu^*}}{m_{d,\mu^*} d_\alpha}\\
&= \frac{1}{d^N}\sum_{\alpha\vdash N-1}m_{d,\alpha} d_\alpha \frac N {\gamma_{\mu^*}(\alpha)}\\
&= \frac 1 d \bE_{\boldsymbol{\alpha}}\left[\frac N {\gamma_{\mu^*}(\boldsymbol{\alpha})}\right]\\
&= \frac 1 d \bE_{\boldsymbol{\alpha}}\left[\frac N {\boldsymbol{\alpha}^{(N-1)}_1+d}\right].
\end{align}

Recall that
\begin{align}
\boldsymbol{\alpha}^{(N-1)} = \left(\boldsymbol{\alpha}^{(N-1)}_1,\dots,\boldsymbol{\alpha}^{(N-1)}_d \right)\sim p_{d,N-1}
\end{align}
is a random vector corresponding to Young diagrams with $N-1$ boxes and at most $d$ rows, where $p_{d,N-1}$ is the Schur-Weyl distribution defined in \eqref{eq:schur-weyl-distribution}. We continue by abbreviating $n=N-1$ and changing to the centered and renormalized random variable $\mathbf A^{(n)}$ from \cref{eq:A-variables}.
\cref{cor:uniform-integrability} implies that
\begin{align}
  \bE\left[\mathbf{A}^{(n)}_{1}\right] &\xrightarrow{N\to\infty} \mathbb E[\lambda_{\max}(\mathbf G)] &
  \bE\left[\big(\mathbf{A}^{(n)}_{1}\big)^2\right] &\xrightarrow{N\to\infty} \mathbb E[\lambda_{\max}(\mathbf G)^2].
  \label{eq:first two moments converge}
\end{align}

Using the $\mathbf{A}^{(n)}$ variables from \eqref{eq:A-variables}, linearity of the expectation value and suitable rearranging, one finds that
\begin{align}
  \sqrt{N - 1} \left(1 - p^{\epr}_d(N)\right) &= \mathbb{E}\left[\sqrt{N-1} - \frac{\sqrt{N-1} N }{\sqrt{d(N-1)} \mathbf{A}^{(n)}_1 + N-1 + d^2}\right] = \mathbb{E}\left[f_{d,N}\left(\mathbf{A}^{(n)}_1\right)\right],
\end{align}
where we set
\begin{align}
f_{d,N}(x) \coloneqq \frac {x \sqrt d+\frac{d^2-1}{\sqrt{N-1}}} {1 + \frac{d^2}{N-1} + \frac{x\sqrt{d}} {\sqrt{N-1}}}.
\end{align}
Note that, for $x\geq 0$,
\begin{align}
\left| f_{d,N}(x) - x \sqrt d \right| \leq \frac{d^2-1}{\sqrt{N-1}} + \frac {x d^{5/2}} {N-1} + \frac {x^2 d} {\sqrt{N-1}} \leq \frac 1{\sqrt{N-1}} \left(K_1 + K_2 x + K_3 x^2\right)
\end{align}
for some constants $K_i$.
Since both $\mathbf{A}^{(n)}_1\geq0$ and $\lambda_{\max}(\mathbf G)\geq0$, and using~\eqref{eq:first two moments converge}, it follows that
\begin{align}
&\left| \mathbb E\left[f_{d,N}\left(\mathbf{A}^{(n)}_1\right)\right] - \sqrt d\, \mathbb E[\lambda_{\max}(\mathbf G)] \right|\\
&\qquad\qquad \leq \left| \mathbb E\left[f_{d,N}\left(\mathbf{A}^{(n)}_1\right)\right] - \sqrt d\, \mathbb E\left[\mathbf{A}^{(n)}_1\right] \right| + \sqrt d \left| \mathbb E\left[\mathbf{A}^{(n)}_1\right] - \sqrt d\, \mathbb E[\lambda_{\max}(\mathbf G)] \right| \\
&\qquad\qquad \leq \frac {K_1 + K_2 \mathbb E\left[\mathbf{A}^{(n)}_1\right] + K_3 \bE\left[\big(\mathbf{A}^{(n)}_{1}\big)^2\right]} {\sqrt{N-1}} + \sqrt d \left\lvert \mathbb E\left[\mathbf{A}^{(n)}_1\right] - \sqrt d\, \mathbb E[\lambda_{\max}(\mathbf G)] \right\rvert\\
&\qquad\qquad \xrightarrow{N\to \infty} 0.
\end{align}
Thus we have shown that, for fixed $d$ and large $N$,
\begin{align}
p^{\epr}_d(N) = 1 -  \sqrt{\frac d {N-1}} \mathbb E[\lambda_{\max}(\mathbf G)] + o\left(N^{-1/2}\right),
\end{align}
which is what we set out to prove.
\end{proof}

\begin{rem}
For the probabilistic protocol with optimized resource state, recall from \cref{eq:prob-opt} that
\begin{align}
p^*_d(N)=1-\frac{d^2-1}{d^2-1+N} = 1 - \frac {d^2-1}N + o(1/N).
\end{align}
For fixed $d$, this converges to unity as $O(1/N)$, i.e., much faster than the $O(1/\sqrt{N})$ convergence in the EPR case proved in \cref{thm:prob-asymptotics} above.
\end{rem}

\section{Deterministic PBT}\label{sec:det}

The following section is divided into two parts.
First, in \Cref{sec:standard} we derive the leading order of the standard protocol for deterministic port-based teleportation (see \Cref{sec:pbt}, where this terminology is explained).
Second, in \Cref{sec:optimal} we derive a lower bound on the leading order of the optimal deterministic protocol.
As in the case of probabilistic PBT, the optimal deterministic protocol converges quadratically faster than the standard deterministic protocol, this time displaying an $N^{-2}$ versus $N^{-1}$ behavior (as opposed to $N^{-1}$ versus $N^{-1/2}$ in the probabilistic case).

\subsection{Asymptotics of the standard protocol}\label{sec:standard}
Our goal in this section is to determine the leading order in the asymptotics of $F^{\std}_d$.
We do so by deriving an expression for the quantity  $\lim_{N\to\infty}N(1 - F^{\std}_d(N))$, that is, we determine the coefficient $c_1=c_1(d)$ in the expansion
\begin{align}
F^{\std}_d(N) = 1 - \frac{c_1}{N} + o(N^{-1}) \,.
\end{align}

 We need the following lemma that states that we can restrict a sequence of expectation values in the Schur-Weyl distribution to a suitably chosen neighborhood of the expectation value and remove degenerate Young diagrams without changing the limit. Let
 \begin{align}
 	H(x)=\begin{cases}
 	0 & x<0\\
 	1 & x\ge 0
 	\end{cases}
 \end{align}
 be the Heaviside step function.
 Recall the definition of the centered and normalized variables
 \begin{align}
 A_i\coloneqq \frac{\alpha_i - n/d}{\sqrt{n/d}},
 \end{align}
 such that $\alpha_i = \sqrt{\frac{n}{d}} A_i + \frac{n}{d}$. In the following it will be advantageous to use both variables, so we use the notation $A(\alpha)$ and $\alpha(A)$ to move back and forth between them.
\begin{lem}\label{lem:doesntmatter}
	Let $C>0$ be a constant and $0<\varepsilon<\frac 1 2(d-2)^{-1}$ (for $d=2$, $\varepsilon>0$ can be chosen arbitrary). Let $f_N$ be a function on the set of centered and rescaled Young diagrams (see \cref{eq:A-variables}) that that grows at most polynomially in $N$, and for $N$ large enough and all arguments $A$ such that  $\|A\|_1\le n^\eps$ fulfills the bound
	\begin{align}
		f_N(A)\le C N.
	\end{align}
	Then the limit of its expectation values does not change when removing degenerate and large deviation diagrams,
	\begin{align}
	\lim_{N\to\infty}\bE_{\boldsymbol{\alpha}}[f_N(\mathbf{A})]=\lim_{N\to\infty}\bE_{\boldsymbol{\alpha}}[f_N(\mathbf{A})H(n^\varepsilon-\|\mathbf{A}\|_1)\mathds 1_{\mathrm{ND}}(\mathbf A)],
	\end{align}
	where $\mathds 1_{\mathrm{ND}}$ is the indicator function that is 0 if two or more entries of its argument are equal, and 1 else. Moreover we have the stronger statement
	\begin{align}
	\left|\bE_{\boldsymbol{\alpha}}[f_N(\mathbf{A})]-\bE_{\boldsymbol{\alpha}}[f_N(\mathbf{A})H(n^\varepsilon-\|\mathbf{A}\|_1)\mathds 1_{\mathrm{ND}}(\mathbf A)]\right|=O(N^{-1/2+(d-2)\varepsilon}).
	\end{align}
\end{lem}
\begin{proof}
	The number of all Young diagrams is bounded from above by a polynomial in $N$. But $p_{d, n}(\alpha(A))=O(\exp(-\gamma\|A\|_1^2))$ for some $\gamma>0$ according to \cref{lem:differences-bounded}, which implies that
	\begin{align}
	\lim_{N\to\infty}\bE_{\boldsymbol{\alpha}}[f_N(\mathbf{A})]=\lim_{N\to\infty}\bE_{\boldsymbol{\alpha}}[f_N(\mathbf{A})H(n^\varepsilon-\|\mathbf{A}\|_1)].
	\end{align}
	Let us now look at the case of degenerate diagrams. Define the set of degenerate diagrams that are also in the support of the above expectation value,
	\begin{align}
	\Xi&=\left\{\alpha\vdash_{d}n\colon \exists\, 1\le i\le d-1 \text{ s.t.~}\alpha_i=\alpha_{i+1}\wedge \left(\frac{n}{d}\right)^{-1/2}\left\|\alpha-\frac n d\mathbf 1\right\|_1\le n^\varepsilon\right\}\\
	&=\mathrm{ND}(d,n)^c\cap\mathrm{supp}(H(n^\varepsilon-\|A\|_1)).
	\end{align}
	Here, $\mathbf 1=(1,\ldots,1)^T\in\R^d$ is the all-one vector. We write
	\begin{align}\label{eq:union}
	\Xi=\bigcup_{k=1}^{d-1}\Xi_k
	\end{align}
	with
	\begin{align}
	\Xi_k=\left\{\alpha\vdash_{d}n\colon \alpha_k=\alpha_{k+1}\wedge \left(\frac{n}{d}\right)^{-1/2}\left\|\alpha-\frac n d\mathbf 1\right\|_1\le n^\varepsilon\right\}.
	\end{align}

	It suffices to show that
	\begin{align}
	\lim_{N\to\infty}\bE_{\boldsymbol{\alpha}}[f(A(\boldsymbol{\alpha}))H(n^\varepsilon-\|A(\boldsymbol{\alpha})\|_1)\mathds 1_{\Xi_k}(\boldsymbol{\alpha})]=0
	\end{align}
	for all $k=1,\ldots,d-1$.
	We can now apply \cref{lem:differences-bounded} to $\Gamma_k$ and choose the constants $\gamma_k=0$ and $\gamma_i=\frac 1 2+\varepsilon$ for $i\neq k$. Using the 1-norm condition on $A$ and bounding the exponential function by a constant we therefore get the bound
	\begin{align}
	p_{d,n}(\alpha(A))\le C_1 n^{-\frac{d+1}{2}}
	\end{align}
	for some constant $C_1>0$. The cardinality of $\Xi_k$ is not greater than the number of integer vectors whose entries are between $n/d-n^{1/2+\varepsilon}$ and $n/d+n^{1/2+\varepsilon}$ and sum to $n$, and for which the $k$-th and $(k+1)$-st entries are equal. It therefore holds that
	\begin{align}
	\left|\Xi_k\right|\le C_2 n^{(d-2)\left(\frac{1}{2} +\varepsilon\right)}.
	\end{align}
	By assumption,
	\begin{align}
	f(A)\le C n \quad \text{for all $A$ such that }\alpha(A)\in\Xi_k.
	\end{align}
	Finally, we conclude that
	\begin{align}
	\bE_{\boldsymbol{\alpha}}[f(A(\boldsymbol{\alpha}))H(n^\varepsilon-\|A(\boldsymbol{\alpha})\|_1)\mathds 1_{\Xi_k}(\boldsymbol{\alpha})]&\le  CC_1C_2n\cdot n^{-\frac{d+1}{2}}n^{(d-2)\left(\frac{1}{2} +\varepsilon\right)}\\
	&\le \tilde C n^{(d-2)\varepsilon-\frac 1 2}.
	\end{align}
	This implies that we have indeed that
	\begin{align}
	\lim_{N\to\infty}\bE_{\boldsymbol{\alpha}}[f(A(\boldsymbol{\alpha}))H(n^\varepsilon-\|A(\boldsymbol{\alpha})\|_1)\mathds 1_{\Xi_k}(\boldsymbol{\alpha})]=0.
	\end{align}
	In fact, we obtain the stronger statement
	\begin{align}
	\left|\bE_{\boldsymbol{\alpha}}[f(A(\boldsymbol{\alpha}))H(n^\varepsilon-\|A(\boldsymbol{\alpha})\|_1)\mathds 1_{\Xi_k}(\boldsymbol{\alpha})]\right|=O(N^{-1/2+(d-2)\varepsilon}).
	\end{align}
	The statement follows now using \cref{eq:union}.
\end{proof}

With \Cref{lem:doesntmatter} in hand, we can now prove the main result of this section, which we stated in \Cref{sec:summary} and restate here for convenience.

\begin{thm-deterministic-standard}[restated]
	\restateDeterministicStandard
\end{thm-deterministic-standard}
\begin{proof}
We first define $n=N-1$ and recall \eqref{eq:cambridge}, which we can rewrite as follows:
\begin{align}
  F^{\std}_d(N) &=d^{-N-2}\sum_{\alpha\vdash_d n}\left(\sum_{\mu=\alpha+\square}\sqrt{d_\mu m_{d,\mu}}\right)^2\\
&=d^{-N-2}\sum_{\alpha\vdash_d n}d_\alpha m_{d,\alpha}\left(\sum_{\mu=\alpha+\square}\frac{m_{d,\mu}}{m_{d,\alpha}}\sqrt{\frac{d_\mu m_{d,\alpha}}{m_{d,\mu} d_\alpha}}\right)^2\\
&=d^{-N-2}\sum_{\alpha\vdash_d n}d_\alpha m_{d,\alpha}\left(\sum_{\mu=\alpha+e_i\text{ YD}}\left[\prod_{j:j\neq i}\frac{\alpha_i-\alpha_j+j-i+1}{\alpha_i-\alpha_j+j-i}\right]\sqrt{\frac{N}{\mu_i-i+d}}\right)^2 \\
&=\frac1{d^3} \mathbb E_{\boldsymbol{\alpha}}\left[\left(\sum_{\boldsymbol{\mu}=\boldsymbol{\alpha}+e_i\text{ YD}}\left[\prod_{j:j\neq i}\frac{\boldsymbol{\alpha}^{(n)}_i-\boldsymbol{\alpha}^{(n)}_j+j-i+1}{\boldsymbol{\alpha}^{(n)}_i-\boldsymbol{\alpha}^{(n)}_j+j-i}\right]\sqrt{\frac{N}{\boldsymbol{\mu}^{(N)}_i-i+d}}\right)^2\right]\\
&=\frac1{d^3} \mathbb E_{\boldsymbol{\alpha}}\left[\left(\sum_{\boldsymbol{\mu}=\boldsymbol{\alpha}+e_i\text{ YD}}\left[\prod_{j:j\neq i} \left( 1 + \frac1{\boldsymbol{\alpha}^{(n)}_i-\boldsymbol{\alpha}^{(n)}_j+j-i} \right) \right]\sqrt{\frac{N}{\boldsymbol{\mu}^{(N)}_i-i+d}}\right)^2\right].\label{eq:reformulated-cambridge}
\end{align}
In the third step, we used \cref{lem:mmm} for the term $\frac{d_\mu m_{d,\alpha}}{m_{d,\mu} d_\alpha}$ and the Weyl dimension formula~\eqref{eq:weyl dim} for the term $\frac{m_{d,\mu}}{m_{d,\alpha}}$.
The expectation value refers to a random choice of $\alpha\vdash_{d}n$ according to the Schur-Weyl distribution~$p_{d,n}$.
The sum over $\mu=\alpha+e_i$ is restricted to only those~$\mu$ that are valid Young diagrams, i.e., where~$\alpha_{i-1}>\alpha_i$, which we indicate by writing `YD'.
Hence, we have
\begin{multline}
  N(1 - F^{\std}_d(N))\\
=\frac N{d^2} \bE_{\boldsymbol{\alpha}}\left[d^2 - \left(\sum_{\boldsymbol{\mu}=\boldsymbol{\alpha}+e_i\text{ YD}}\left[\prod_{j\neq i} \left( 1 + \frac1{\boldsymbol{\alpha}^{(n)}_i-\boldsymbol{\alpha}^{(n)}_j+j-i} \right) \right]\sqrt{\frac{N/d}{\boldsymbol{\mu}^{(N)}_i-i+d}}\right)^2\right].\label{eq:first-order-coefficient}
\end{multline}
In the following, we suppress the superscript indicating $n=N-1$ for the sake of readability.
The random variables $\boldsymbol{\alpha}$, $\mathbf{A}$, and $\boldsymbol{\Gamma}_{ij}$, as well as their particular values $\alpha$, $A$, and $\Gamma_{ij}$, are all understood to be functions of $n=N-1$.

 The function
 \begin{align}
 	f_N( A)\coloneqq \frac N{d^2} \left(d^2 - \left(\sum_{{\mu}={\alpha}( A)+e_i\text{ YD}}\left[\prod_{j:j\neq i} \left( 1 + \frac 1{{\alpha}_i(A)-{\alpha}_j(A)+j-i} \right) \right]\sqrt{\frac{N/d}{{\mu}_i-i+d}}\right)^2\right)
 \end{align}
 satisfies the requirements of \cref{lem:doesntmatter}. Indeed we have that
 \begin{align}
 	\frac 1{{\alpha}_i(A)-{\alpha}_j(A)+j-i} \le 1
 \end{align}
 for all $i\neq j$, and  clearly
 \begin{align}
 	\sqrt{\frac{N/d}{{\mu}_i-i+d}}\le \sqrt{N}.
 \end{align}
 Therefore we get
 \begin{align}
 	f_N( A)\le C N^2
 \end{align}
 for some constant $C$.
 If $\|A\|_1\le n^\varepsilon$, we have that
  \begin{align}
 \sqrt{\frac{N/d}{{\mu}_i-i+d}}\le\sqrt{\frac{N/d}{n/d-n^\varepsilon}}\le \sqrt{N/n}+O\left(n^{-(1-\varepsilon)}\right)
 \end{align}
 and hence
 \begin{align}
 f_N( A)\le C N
 \end{align}
 for $N$ large enough.
 We therefore define, using an $\varepsilon$ in the range given by \cref{lem:doesntmatter}, the modified expectation value
\begin{align}\label{eq:modexpect}
  \tilde{\bE}_{\boldsymbol{\alpha}}[f(\boldsymbol{\alpha})]\coloneqq\bE_{\boldsymbol{\alpha}}[f(\boldsymbol{\alpha})\mathds 1_{\mathrm{ND}(n,d)}(\boldsymbol{\alpha})H(n^\varepsilon-\|A(\boldsymbol\alpha)\|_1)],
\end{align}
 and note that an application of \cref{lem:doesntmatter} shows that the limit that we are striving to calculate does not change when replacing the expectation value with the above modified expectation value, and the difference between the members of the two sequences is $O(n^{-\frac 1 2+\varepsilon(d-2)})$.

For a non-degenerate $\alpha$, adding a box to any row yields a valid Young diagram $\mu$.
Hence, the sum $\sum_{\mu=\alpha+e_i\text{ YD}}$ in \eqref{eq:first-order-coefficient} can be replaced by $\sum_{i=1}^d$, at the same time replacing $\mu_i$ with $\alpha_i+1$.
The expression in \eqref{eq:first-order-coefficient} therefore simplifies to
\begin{align}
R_N
&\coloneqq\frac N{d^2} \tbE_{\boldsymbol{\alpha}}\left[d^2 - \left(\sum_{i=1}^d\left[\prod_{k:k\neq i} \left( 1 + \frac{1}{\boldsymbol{\alpha}_i-\boldsymbol{\alpha}_k+k-i} \right) \right]\sqrt{\frac{N/d}{\boldsymbol{\alpha}_i+1-i+d}}\right)^2\right].\label{eq:first-order-coefficient-non-deg}
\end{align}

Let us look at the square root term, using the variables $\mathbf A_i$. For sufficiently large $n$, we write
\begin{align}
	\sqrt{\frac{N/d}{\boldsymbol{\alpha}_i+1-i+d}}&=\sqrt{N/n}\left(1+\frac{(1-i+d)d}{n}+\sqrt{\frac d n}\mathbf A_i\right)^{-1/2}\\
	&=\sqrt{\frac{N\gamma_{i,d,n}}{n}}\left(1+\gamma_{i,d,n}\sqrt{\frac d n}\mathbf A_i\right)^{-1/2}\\
	&=\sqrt{\frac{N\gamma_{i,d,n}}{n}}\sum_{r=0}^\infty a_r \left(\gamma_{i,d,n}\sqrt{\frac d n}\mathbf A_i\right)^r.
\end{align}
In the second line we have defined
\begin{align}
\gamma_{i,d,n}=\left(1+\frac{(1-i+d)d}{n}\right)^{-1},
\end{align}
and in the third line we have written the inverse square root in terms of its power series around $1$. This is possible as we have $\|\mathbf A\|_1\le n^\varepsilon$ on the domain of $\tbE$, so $\gamma_{i,d,n}\sqrt{\frac d n}\mathbf A_i=O(n^{-1/2+\varepsilon})$, i.e. it is in particular in the convergence radius of the power series, which is equal to $1$. This implies also that the series converges absolutely in that range.
Defining
	\begin{align}
	\boldsymbol\Gamma_{ik}=-\boldsymbol\Gamma_{ki}=\mathbf A_i-\mathbf A_k+\frac{k-i}{\sqrt{\frac{n}{d}}}
	\end{align}
as in \cref{sec:schur weyl dist}, we can write
\begin{align}
R_N
&=\frac{N}{d^2} \tbE_{\boldsymbol{\alpha}}\Bigg[d^2 - \frac N n\sum_{i,j=1}^d\sqrt{\gamma_{i,d,n}\gamma_{j,d,n}}\left[\prod_{k:k< i} \left( 1 -\sqrt{\frac{d}{n}}\mathbf \Gamma_{ki}^{-1} \right) \right]\left[\prod_{k:k> i} \left( 1 +\sqrt{\frac{d}{n}}\mathbf \Gamma_{ik}^{-1} \right) \right]\\
&\quad\quad\quad\quad\quad\quad\left[\prod_{l:l< j} \left( 1 -\sqrt{\frac{d}{n}}\mathbf \Gamma_{lj}^{-1} \right) \right]\left[\prod_{l:l> j} \left( 1 +\sqrt{\frac{d}{n}}\mathbf \Gamma_{jl}^{-1} \right) \right]\\
&\quad\quad\quad\quad\quad\quad\left(\sum_{r=0}^\infty a_r \left(\gamma_{i,d,n}\sqrt{\frac d n}\mathbf A_i\right)^r\right)\left(\sum_{r=0}^\infty a_r \left(\gamma_{j,d,n}\sqrt{\frac d n}\mathbf A_j\right)^r\right)\Bigg]\\
&\eqqcolon \frac{N}{d^2} \tbE_{\boldsymbol{\alpha}}\left[d^2- \frac{N}{n}\sum_{i,j=1}^d\sqrt{\gamma_{i,d,n}\gamma_{j,d,n}}\left(\sum_{s=0}^{2(d-1)}\left(\frac{d}{n}\right)^{\frac{s}{2}}P_{i,j}^{(1,s)}\left(\mathbf \Gamma^{-1}\right)\right)\left(\sum_{r=0}^{\infty}\left(\frac{d}{n}\right)^{\frac{r}{2}}P^{(2,r)}_{i,j}(\tilde{\mathbf A})\right)\right].\label{eq:series-repr}
\end{align}
Here we have defined $\tilde{ \mathbf A}$ by $\tilde{ \mathbf A}_i=\gamma_{i,d,n}\mathbf A_i$ and the polynomials $P_{i,j}^{(1,s)}$, $P_{i,j}^{(2,r)}$, for $s=0,...,2(d-1)$, $r\in\N$, $i,j=1,...,d$, which are homogeneous of degree $r$, and $s$, respectively. In the last equality we have used the absolute convergence of the power series.  We have also abbreviated $\mathbf \Gamma\coloneqq(\mathbf \Gamma_{ij})_{i<j}$, $\mathbf \Gamma^{-1}$ is to be understood elementwise, and $P_{i,j}^{(1,s)}$ has the additional property that for all $k,l\in\{1,...,d\}$ it has degree at most $2$ in each variable $\mathbf \Gamma_{k,l}$.

By the Fubini-Tonelli Theorem,
we can now exchange the infinite sum and the expectation value if the expectation value
\begin{align}
	\tbE_{\boldsymbol{\alpha}}\left[\left(\sum_{s=0}^{2(d-1)}\left(\frac{d}{n}\right)^{\frac{s}{2}}\tilde P_{i,j}^{(1,s)}\left(|\mathbf \Gamma^{-1}|\right)\right)\left(\sum_{r=0}^{\infty}\left(\frac{d}{n}\right)^{\frac{r}{2}}\tilde P^{(2,r)}_{i,j}(|\tilde{ \mathbf A}|)\right)\right]
\end{align}
exists, where the polynomials $\tilde P^{(1,s)}_{i,j}$ and $\tilde P^{(2,r)}_{i,j}$ are obtained from $P^{(1,s)}_{i,j}$ and $P^{(2,r)}_{i,j}$, respectively, by replacing the coefficients with their absolute value, and the absolute values $|\mathbf \Gamma^{-1}|$ and $|\tilde {\mathbf A}|$ are to be understood element-wise. But the power series of the square root we have used converges absolutely on the range of $\mathbf A$ restricted by $\tbE$ (see \cref{eq:modexpect}), yielding a continuous function on an appropriately chosen compact interval. Moreover, if $A$ is in the range of $\mathbf A$ restricted by $\tbE$, then so is $|A|$. The function is therefore bounded, as is $\tilde{ \mathbf A}$ for fixed $N$, and the expectation value above exists. We therefore get
\begin{align}
R_N
&=\frac{N}{d^2}\left[d^2-\frac{N}{n}\sum_{i,j=1}^d\sqrt{\gamma_{i,d,n}\gamma_{j,d,n}}\sum_{s=0}^{2(d-1)}\sum_{r=0}^{\infty}\left(\frac{d}{n}\right)^{\frac{s+r}{2}} \tbE_{\boldsymbol{\alpha}}\left[P_{i,j}^{(1,s)}\left(\mathbf \Gamma^{-1}\right)P^{(2,r)}_{i,j}(\tilde{ \mathbf A})\right]\right].
\end{align}
Now note that the expectation values above have the right form to apply \cref{thm:expectation-convergence}, so we can start calculating expectation values provided that we can exchange the limit $N\to\infty$ with the infinite sum.
We can then split up the quantity $\lim_{N\to\infty}R_N$ as follows,
 \begin{align}
 	\lim_{N\to\infty}R_N&=\lim_{N\to\infty}\frac{N}{d^2}\left[d^2-\frac{N}{n}\sum_{i,j=1}^d\sqrt{\gamma_{i,d,n}\gamma_{j,d,n}}\sum_{s=0}^{2(d-1)}\sum_{r=0}^{\infty}\left(\frac{d}{n}\right)^{\frac{s+r}{2}} \tbE_{\boldsymbol{\alpha}}\left[P_{i,j}^{(1,s)}\left(\mathbf \Gamma^{-1}\right)P^{(2,r)}_{i,j}(\tilde{ \mathbf A})\right]\right]\\
 	&\label{eq:first limit}
  =\lim_{N\to\infty}\frac{N^2}{nd^2}\tbE_{\boldsymbol{\alpha}}\left[ \frac{d^2 n}{N}-\sum_{i,j=1}^d\sqrt{\gamma_{i,d,n}\gamma_{j,d,n}}P_{i,j}^{(1,0)}\left(\mathbf \Gamma^{-1}\right)P^{(2,0)}_{i,j}(\tilde{ \mathbf A})\right]\\
 	&\label{eq:second limit}
  \quad -\sum_{r,s\in\N: 1\le r+s\le 2}\lim_{N\to\infty}\frac{N^2}{nd^2}\left(\frac{d}{n}\right)^{\frac{r+s}{2}} \tbE_{\boldsymbol{\alpha}}\left[ \sum_{i,j=1}^d\sqrt{\gamma_{i,d,n}\gamma_{j,d,n}}P_{i,j}^{(1,s)}\left(\mathbf \Gamma^{-1}\right)P^{(2,r)}_{i,j}(\tilde{ \mathbf A})\right]\\
 	&\label{eq:third limit}
  \quad -\lim_{N\to\infty}\sum_{\substack{r,s\in\N\\  r+s\ge 3\\s\le 2(d-1)}}\frac{N^2}{nd^2}\left(\frac{d}{n}\right)^{\frac{r+s}{2}} \tbE_{\boldsymbol{\alpha}}\left[ \sum_{i,j=1}^d\sqrt{\gamma_{i,d,n}\gamma_{j,d,n}}P_{i,j}^{(1,s)}\left(\mathbf \Gamma^{-1}\right)P^{(2,r)}_{i,j}(\tilde{ \mathbf A})\right],
\end{align}
provided that all the limits on the right hand side exist.
We continue by determining these limits and begin with \cref{eq:third limit}.
First observe that, for fixed $r$ and $s$ such that~$r+s\ge 3$,
\begin{align}\label{eq:2zero}
	\lim_{N\to\infty}\frac{N^2}{nd^2}\left(\frac{d}{n}\right)^{\frac{r+s}{2}} \tbE_{\boldsymbol{\alpha}}\left[ \sum_{i,j=1}^d\sqrt{\gamma_{i,d,n}\gamma_{j,d,n}}P_{i,j}^{(1,s)}\left(\mathbf \Gamma^{-1}\right)P^{(2,r)}_{i,j}(\tilde{ \mathbf A})\right]=0.
\end{align}
This is because the expectation value in \cref{eq:2zero} converges according to \cref{thm:expectation-convergence} and \cref{lem:doesntmatter}, which in turn implies that the whole expression is $O(N^{-1/2})$.
In particular, there exists a constant~$K>0$ such that, for the finitely many values of $r$ and $s$ such that $r\leq r_0\coloneqq \lceil (\frac{1}{2}-\varepsilon)^{-1}\rceil$,
\begin{align}
	\frac{N^2}{nd^2}\left(\frac{d}{n}\right)^{\frac{r+s}{2}} \tbE_{\boldsymbol{\alpha}}\left[ \sum_{i,j=1}^d\sqrt{\gamma_{i,d,n}\gamma_{j,d,n}}P_{i,j}^{(1,s)}\left(\mathbf \Gamma^{-1}\right)P^{(2,r)}_{i,j}(\tilde{ \mathbf A})\right]\le K \quad (\forall N).
\end{align}
Now suppose that $r>r_0$.
On the domain of $\tbE$, we have $\|\mathbf A\|_1\le n^\varepsilon$.
Therefore, we can bound
\begin{align}\label{eq:lowerorderbound}
  \frac{N^2}{nd}\left(\frac{d}{n}\right)^{\frac{s+r}{2}} P_{i,j}^{(1,s)}\left(\mathbf \Gamma^{-1}\right)P^{(2,r)}_{i,j}(\tilde{ \mathbf A})\le C N^{1+r(\varepsilon-1/2)}\le C \, 2^{1+r(\varepsilon-1/2)} \quad (\forall N)
\end{align}
The first step holds because $\left(\frac{d}{n}\right)^{\frac{s}{2}} P_{i,j}^{(1,s)}\left(\mathbf \Gamma^{-1}\right)$ is a polynomial in the variables $\left(\frac{d}{n}\right)^{\frac{1}{2}} \mathbf \Gamma^{-1}_{ij}\le1 $ with coefficients independent of $n$, and in the second step we used that $1+r(\varepsilon-1/2) < 0$.
We can therefore apply the dominated convergence theorem using the dominating function
\begin{align}
	g(r,s)=\begin{cases}
		K& r\leq r_0 = \lceil (\frac{1}{2}-\varepsilon)^{-1}\rceil\\
		C \cdot 2^{1+r(\varepsilon-1/2)} & \text{else}
	\end{cases}
\end{align}
to exchange the limit and the sum in \cref{eq:third limit}.
Thus, \cref{eq:2zero} implies that \cref{eq:third limit} is zero.

It remains to compute the limits \cref{eq:first limit,eq:second limit}, i.e., the terms
\begin{align}
	T_{s,r} \coloneqq \lim_{N\to\infty}\frac{N^2}{nd^2}\left(\frac{d}{n}\right)^{\frac{s+r}{2}} \tbE_{\boldsymbol{\alpha}}\left[\delta_{s0}\delta_{r0} \frac{d^2 n}{N}-\sum_{i,j=1}^d\sqrt{\gamma_{i,d,n}\gamma_{j,d,n}}P_{i,j}^{(1,s)}\left(\mathbf \Gamma^{-1}\right)P^{(2,r)}_{i,j}(\tilde{ \mathbf A})\right].\label{eq:T_sr}
\end{align}
for $r+s=0,1,2$. The first few terms of the power series for the inverse square root are given by
\begin{align}
	(1+x)^{-1/2}=1-\frac x 2+\frac{3x^2}{8}+O(x^3).
\end{align}
The relevant polynomials to calculate the remaining limits are, using the above and $-\boldsymbol \Gamma_{ik}=\mathbf \Gamma_{ki}$,
\begin{align}
	P^{(1,0)}_{i,j}&=P^{(1,0)}_{i,j}\equiv 1\\
	P^{(1,1)}_{i,j}(\boldsymbol{\Gamma}^{-1})&=-\sum_{k:k< i}\mathbf \Gamma_{k:ki}^{-1} +\sum_{k> i} \mathbf \Gamma_{ik}^{-1} -\sum_{l:l< j} \mathbf \Gamma_{lj}^{-1} +\sum_{l:l> j}\mathbf \Gamma_{jl}^{-1} \\
	&=\sum_{k:k\neq i}\mathbf \Gamma_{ik}^{-1}+\sum_{l:l\neq j} \mathbf \Gamma_{jl}^{-1}\\
	P^{(2,1)}_{i,j}(\tilde{\mathbf A})&=a_1a_0(\tilde{\mathbf A}_i+\tilde{\mathbf A}_j)=-\frac 1 2 (\tilde{\mathbf A}_i+\tilde{\mathbf A}_j)\\
	P^{(1,2)}_{i,j}(\boldsymbol{\Gamma}^{-1})&=\sum_{k,l:i\neq k\neq l\neq i} \mathbf \Gamma_{ik}^{-1} \mathbf \Gamma_{il}^{-1}+\sum_{k,l:j\neq k\neq l\neq j}\mathbf  \Gamma_{ik}^{-1} \mathbf \Gamma_{il}^{-1}+\sum_{k,l:i\neq k, l\neq j} \mathbf \Gamma_{ik}^{-1} \mathbf \Gamma_{il}^{-1}\\
	P^{(2,2)}_{i,j}(\tilde{\mathbf A})&=a_2a_0(\tilde{\mathbf A}_i^2+\tilde{\mathbf A}_j^2)+a_1^2 \tilde{\mathbf A}_i\tilde{\mathbf A}_j=\frac 3 8(\tilde{\mathbf A}_i^2+\tilde{\mathbf A}_j^2)+\frac 1 4 \tilde{\mathbf A}_i\tilde{\mathbf A}_j.
\end{align}
We now analyze the remaining expectation values using these explicit expressions for the corresponding polynomials.

\subsubsection*{Evaluating $T_{0,0}$}
Using the power series expansions
\begin{align}
\sqrt{\gamma_{i,d,n}}=1-\frac{d(d+1-i)}{2n}+O(n^{-2})=1-O(n^{-1}),\label{eq:gamma-expansion}
\end{align}
we simplify
\begin{align}
T_{0,0}&= \lim_{N\to\infty}\frac{N^2}{nd^2} \tbE_{\boldsymbol{\alpha}}\left[ \frac{d^2 n}{N}-\sum_{i,j=1}^d\sqrt{\gamma_{i,d,n}\gamma_{j,d,n}}P_{i,j}^{(1,0)}\left(\mathbf \Gamma^{-1}\right)P^{(2,0)}_{i,j}(\tilde{ \mathbf A})\right]\\
&=\lim_{N\to\infty}\frac{N^2}{nd^2}\left(\frac{d^2 n}{N}- \sum_{i,j=1}^d\left(1-\frac{d(d+1-i)}{2n}-\frac{d(d+1-j)}{2n}+O(n^{-2})\right)\right)\\
&=\lim_{N\to\infty}\left(N- \frac{N^2}{n}+\frac{N^2}{nd^2}\sum_{i,j=1}^d\left(\frac{d(2d+2-i-j)}{2n}+O(n^{-2})\right)\right)\\
&=\lim_{N\to\infty}\left(- \frac{n+1}{n}+\frac{N^2}{2n^2}\left(2d^2+2d-d(d+1)\right)+O(n^{-1})\right)\\
&=\frac{d(d+1)}{2}-1.\label{eq:T00}
\end{align}
In the second-to-last line we have replaced $N=n+1$.

\subsubsection*{Evaluating $T_{0,1}$ and $T_{1,0}$}
We first compute
\begin{align}
	\sum_{i,j=1}^d\sqrt{\gamma_{i,d,n}\gamma_{j,d,n}}P_{i,j}^{(1,1)}\left(\mathbf \Gamma^{-1}\right)P^{(2,0)}_{i,j}(\tilde{ \mathbf A})
	&=\sum_{i,j=1}^d\sqrt{\gamma_{i,d,n}\gamma_{j,d,n}}\left(\sum_{k:k\neq i}\mathbf \Gamma_{ik}^{-1}+\sum_{l:l\neq j} \mathbf \Gamma_{jl}^{-1}\right)\\
	&=\sum_{i,j=1}^d(1+O(n^{-1}))\left(\sum_{k:k\neq i}\mathbf \Gamma_{ik}^{-1}+\sum_{l:l\neq j} \mathbf \Gamma_{jl}^{-1}\right)\\
	&=\sum_{i,j=1}^dO(n^{-1})\left(\sum_{k:k\neq i}\mathbf \Gamma_{ik}^{-1}+\sum_{l:l\neq j} \mathbf \Gamma_{jl}^{-1}\right).
\end{align}
In the last equation we have used that
\begin{align}
	\sum_{i\neq k}\mathbf \Gamma_{ik}^{-1}=0,
\end{align}
as the summation domain is symmetric in $i$ and $k$, and $ \mathbf\Gamma_{ik}^{-1}=- \mathbf\Gamma_{ki}^{-1}$. But now we can determine the limit,
\begin{align}
T_{1,0}&= \lim_{N\to\infty}\frac{N^2}{nd^2}\left(\frac{d}{n}\right)^{\frac{1}{2}} \tbE_{\boldsymbol{\alpha}}\left[\sum_{i,j=1}^d\sqrt{\gamma_{i,d,n}\gamma_{j,d,n}}P_{i,j}^{(1,1)}\left(\mathbf \Gamma^{-1}\right)P^{(2,0)}_{i,j}(\tilde{ \mathbf A})\right]\\
&=\lim_{N\to\infty}\frac{N^2}{nd^2}\left(\frac{d}{n}\right)^{\frac{1}{2}} \tbE_{\boldsymbol{\alpha}}\left[\sum_{i,j=1}^dO(n^{-1})\left(\sum_{k:k\neq i}\mathbf \Gamma_{ik}^{-1}+\sum_{l:l\neq j} \mathbf \Gamma_{jl}^{-1}\right)\right]\\
&=\sum_{i,j=1}^d\lim_{N\to\infty}O(n^{-1/2})\tbE_{\boldsymbol{\alpha}}\left[\left(\sum_{k:k\neq i}\mathbf \Gamma_{ik}^{-1}+\sum_{l:l\neq j} \mathbf \Gamma_{jl}^{-1}\right)\right]=0.\label{eq:T10}
\end{align}
Here we have used \cref{thm:expectation-convergence} to see that the sequence of expectation values converges, implying that the expression vanishes due to the $O(n^{-1/2})$ prefactor.

Similarly, to show that $T_{0,1}$ vanishes as well, we calculate
\begin{align}
&\sum_{i,j=1}^d\sqrt{\gamma_{i,d,n}\gamma_{j,d,n}}P_{i,j}^{(1,0)}\left(\mathbf \Gamma^{-1}\right)P^{(2,1)}_{i,j}(\tilde{ \mathbf A})\\
&=-\frac 1 2\sum_{i,j=1}^d\sqrt{\gamma_{i,d,n}\gamma_{j,d,n}} (\tilde{\mathbf A}_i+\tilde{\mathbf A}_j)\\
&=-\frac 1 2\sum_{i,j=1}^d\sqrt{\gamma_{i,d,n}\gamma_{j,d,n}} (\gamma_{i,d,n}{\mathbf A}_i+\gamma_{j,d,n}{\mathbf A}_j)\\
&=-\sum_{i,j=1}^d(1+O(n^{-1}))({\mathbf A}_i+{\mathbf A}_j)\\
&=-\sum_{i,j=1}^dO(n^{-1})({\mathbf A}_i+{\mathbf A}_j).
\end{align}
Here we have used that $\sqrt{\gamma_{i,d,n}}=1-O(n^{-1})=\gamma_{i,d,n}$, and in the last line we used that $\sum_{i=1}^d{\mathbf A}_i=0$. This implies, using the same argument as in \cref{eq:T10}, that
\begin{align}
T_{0,1} = \lim_{N\to\infty}\frac{N^2}{nd^2}\left(\frac{d}{n}\right)^{\frac{1}{2}} \tbE_{\boldsymbol{\alpha}}\left[\sum_{i,j=1}^d\sqrt{\gamma_{i,d,n}\gamma_{j,d,n}}P_{i,j}^{(1,0)}\left(\mathbf \Gamma^{-1}\right)P^{2,1}_{i,j}(\tilde{ \mathbf A})\right]=0.\label{eq:T01}
\end{align}

\subsubsection*{Evaluating $T_{s,r}$ for $s+r=2$}
For $s+r=2$, we first observe that
\begin{align}
&\lim_{N\to\infty}\frac{N^2}{nd^2}\left(\frac{d}{n}\right)^{\frac{s+r}{2}} \tbE_{\boldsymbol{\alpha}}\left[\sum_{i,j=1}^d\sqrt{\gamma_{i,d,n}\gamma_{j,d,n}}P_{i,j}^{(1,s)}\left(\mathbf \Gamma^{-1}\right)P^{(2,r)}_{i,j}(\tilde{ \mathbf A})\right]\\
&= \lim_{N\to\infty}\frac{N^2}{n^2d}\tbE_{\boldsymbol{\alpha}}\left[\sum_{i,j=1}^d\sqrt{\gamma_{i,d,n}\gamma_{j,d,n}}P_{i,j}^{(1,s)}\left(\mathbf \Gamma^{-1}\right)P^{(2,r)}_{i,j}(\tilde{ \mathbf A})\right]\\
&=\frac 1 d\lim_{N\to\infty}\tbE_{\boldsymbol{\alpha}}\left[\sum_{i,j=1}^d\sqrt{\gamma_{i,d,n}\gamma_{j,d,n}}P_{i,j}^{(1,s)}\left(\mathbf \Gamma^{-1}\right)P^{(2,r)}_{i,j}(\tilde{ \mathbf A})\right].
\end{align}
Therefore we can replace all occurrences of $\gamma_{i,d,n}$ by $1$ using the same argument as in \cref{eq:T10}. There are three cases to take care of, $(s,r)=\lbrace (2,0),(1,1),(0,2)\rbrace$. For $(s,r)=(2,0)$, we first look at the term
\begin{align}
	\sum_{i,j}\sum_{k,l:i\neq k\neq l\neq i} \mathbf \Gamma_{ik}^{-1} \mathbf \Gamma_{il}^{-1}&=d\sum_{i,k,l:i\neq k\neq l\neq i} \mathbf \Gamma_{ik}^{-1}\mathbf \Gamma_{il}^{-1}\\
	&=d\sum_{i,k,l:i\neq k\neq l\neq i}\frac{1}{\left(\pmb{\mathscr{A}}_i-\pmb{\mathscr{A}}_k\right)\left(\pmb{\mathscr{A}}_i-\pmb{\mathscr{A}}_l\right)},
\end{align}
where we have defined $\pmb{\mathscr{A}}_i=\mathbf A_i-i\sqrt{\frac{d}{n}}$. For fixed $i_0\neq k_0\neq j_0\neq i_0$, all permutations of these indices appear in the sum. For these terms with $i,j,k\in\{i_0,j_0,k_0\}$,
\begin{align}
&\sum_{\substack{i,j,k\in\{i_0,j_0,k_0\}\\ i\neq k\neq j\neq i}}\frac{1}{\left(\pmb{\mathscr{A}}_i-\pmb{\mathscr{A}}_j\right)\left(\pmb{\mathscr{A}}_i-\pmb{\mathscr{A}}_k\right)}\\
&= \frac{2}{\left(\pmb{\mathscr{A}}_{i_0}-\pmb{\mathscr{A}}_{j_0}\right)\left(\pmb{\mathscr{A}}_{j_0}-\pmb{\mathscr{A}}_{k_0}\right)\left(\pmb{\mathscr{A}}_{k_0}-\pmb{\mathscr{A}}_{i_0}\right)}\left(-\left(\pmb{\mathscr{A}}_{j_0}-\pmb{\mathscr{A}}_{k_0}\right)-\left(\pmb{\mathscr{A}}_{k_0}-\pmb{\mathscr{A}}_{i_0}\right)-\left(\pmb{\mathscr{A}}_{i_0}-\pmb{\mathscr{A}}_{j_0}\right)\right)=0,
\end{align}
implying
\begin{align}
\sum_{i,j}\sum_{k,l:i\neq k\neq l\neq i} \mathbf \Gamma_{ik}^{-1} \mathbf \Gamma_{il}^{-1}&=0
\end{align}
and therefore
\begin{align}
T_{2,0} = \lim_{N\to\infty}\frac{N^2}{n^2d}\tbE_{\boldsymbol{\alpha}}\left[\sum_{i,j=1}^d\sqrt{\gamma_{i,d,n}\gamma_{j,d,n}}P_{i,j}^{(1,2)}\left(\mathbf \Gamma^{-1}\right)P^{(2,0)}_{i,j}(\tilde{ \mathbf A})\right]=0.\label{eq:T20}
\end{align}
Moving on to the case $(s,r)=(0,2)$, we first note that, as in the previous case $(s,r)=(2,0)$ and again replacing all occurrences of $\gamma_{i,d,n}$ by $1$, we have
 \begin{align}
 	\lim_{N\to\infty}\tbE[\sum_{i,j=1}^dP^{(2,2)}_{i,j}(\tilde{\mathbf A})]&=\sum_{i,j=1}^d\lim_{N\to\infty}\tbE[P^{(2,2)}_{i,j}(\mathbf A)]\\
  &=\sum_{i,j=1}^d\bE\left[\frac 3 8({\mathbf S}_i^2+{\mathbf S}_j^2)+\frac 1 4 {\mathbf S}_i{\mathbf S}_j\right]\\
  &=\frac {3d} 4\bE\left[\sum_{i=1}^d{\mathbf S}_i^2\right]+\frac 1 4 \bE\left[\sum_{i,j=1}^d{\mathbf S}_i{\mathbf S}_j\right].\label{eq:E22}
 \end{align}
 Here, $\mathbf S=\mathrm{spec}(\mathbf G)\sim\GUEzd$ and we have used \cref{lem:doesntmatter} to switch back to the unrestricted expectation value and \cref{thm:expectation-convergence} in the second equality. First we observe that $\mathbf G$ is traceless, and hence $\sum_{i=1}^d{\mathbf S}_i=0$ such that the second term in \eqref{eq:E22} vanishes. For the first term in \eqref{eq:E22}, let $\mathbf X\sim\GUEd$ such that $\mathbf G = \mathbf X - \frac{\tr(\mathbf X)}{d} I\sim\GUEzd$. We calculate
 \begin{align}
 \sumi_i \mathbf{S}_i^2 &= \tr(\mathbf{G}^2)\\
 &= \tr(\mathbf{X}^2) - \frac{1}{d} \tr(\mathbf{X})^2\\
 &= \sum_{i,j} |\mathbf{X}_{ij}|^2 - \frac{1}{d}\left(\sumi_i \mathbf{X}_{ii}\right)^2\\
 &= \sum_{i} \mathbf{X}_{ii}^2 + \sum_{i\neq j} |\mathbf{X}_{ij}|^2 - \frac{1}{d}\sum_i \mathbf{X}_{ii}^2 - \frac{2}{d} \sum_{i\neq j} \mathbf{X}_{ii}\mathbf{X}_{jj}.\label{eq:GUE-elements}
 \end{align}
 We have $\bE[\mathbf{X}_{ii}^2] = \bV[\mathbf{X}_{ii}] = 1$, where $\bV[\cdot]$ denotes the variance of a random variable.
 Similarly, for $i\neq j$,
 \begin{align}
 \bE[|\mathbf{X}_{ij}|^2] = \bV[\Re(\mathbf{X}_{ij})] + \bV[\Im(\mathbf{X}_{ij})] = \frac{1}{2} + \frac{1}{2} = 1,
 \end{align}
 and $\bE[\mathbf{X}_{ii}\mathbf{X}_{jj}] = \bE[\mathbf{X}_{ii}]\bE[\mathbf{X}_{jj}] = 0$, since the entries of a $\GUEd$-matrix are independent.
 Hence, taking expectation values in \eqref{eq:GUE-elements} gives
 \begin{align}
 \bE_{\boldsymbol{\alpha}}\left[\sumi_i \mathbf{S}_i^2\right] &= d + d(d-1) - 1 = d^2-1,
 \end{align}
 and we can calculate
 \begin{align}
 T_{0,2} = \lim_{N\to\infty}\frac{N^2}{n^2d}\tbE_{\boldsymbol{\alpha}}\left[\sum_{i,j=1}^d\sqrt{\gamma_{i,d,n}\gamma_{j,d,n}}P_{i,j}^{(1,0)}\left(\mathbf \Gamma^{-1}\right)P^{(2,2)}_{i,j}(\tilde{ \mathbf A})\right]=\frac 3 4 \left(d^2-1\right).\label{eq:T02}
 \end{align}
 We finally turn to the only missing case, $(s,r)=(1,1)$. The polynomial $P^{(1,1)}_{ij}$ is symmetric in $i$ and $j$, therefore we can simplify
 \begin{align}
 	\sum_{i,j}P^{(1,1)}_{ij}(\mathbf \Gamma^{-1})P^{(2,1)}_{ij}(\mathbf \Gamma^{-1})&=-\frac 1 2\sum_{i,j}\left(\sum_{k:k\neq i}\mathbf \Gamma_{ik}^{-1}+\sum_{l:l\neq j} \mathbf \Gamma_{jl}^{-1}\right)\left(\tilde {\mathbf A}_i+\tilde {\mathbf A}_j\right)\\
 	&=-\sum_{i,j}\left(\sum_{k:k\neq i}\mathbf \Gamma_{ik}^{-1}+\sum_{l:l\neq j} \mathbf \Gamma_{jl}^{-1}\right)\tilde {\mathbf A}_i\\
 	&=-\sum_{i,j}\left(\sum_{k:k\neq i}\mathbf \Gamma_{ik}^{-1}+\sum_{l:l\neq j} \mathbf \Gamma_{jl}^{-1}\right){\mathbf A}_i\\
	&=-d\sum_{i,k:k\neq i}\mathbf \Gamma_{ik}^{-1}\mathbf A_i,\label{eq:r1s1-1}
 \end{align}
where we have used in the second-to-last equation that we can replace any occurrence of $\gamma_{i,d,n}$ by one, and the last equation follows by the same reasoning as used in the case $(s,r)=(0,1)$ above. Now observe that for each $i\neq k$, both $\Gamma_{ik}^{-1}$ and $\Gamma_{ki}^{-1}=-\Gamma_{ik}^{-1}$ occur in the sum. Therefore we can simplify
\begin{align}
\sum_{i,k:k\neq i}\mathbf \Gamma_{ik}^{-1}\mathbf A_i&=\sum_{i,k: i<k}\mathbf \Gamma_{ik}^{-1}(\mathbf A_i-\mathbf A_k)\\
&=\sum_{i,k: i<k}\mathbf \Gamma_{ik}^{-1}\left(\mathbf A_i-\mathbf A_k+\frac{k-i}{\sqrt{\frac n d}}\right)+\sqrt{\frac{d}{n}}\sum_{i,k: i<k}(i-k)\mathbf \Gamma_{ik}^{-1}\\
&=\sum_{i,k: i<k}\mathbf \Gamma_{ik}^{-1}\mathbf \Gamma_{ik}+\sqrt{\frac{d}{n}}\sum_{i,k: i<k}(i-k)\mathbf \Gamma_{ik}^{-1}\\
&=\frac{d(d-1)}{2}-O\left(n^{-1/2}\right)\sum_{i,k: i<k}(i-j)\mathbf \Gamma_{ik}^{-1},\label{eq:r1s1-2}
\end{align}
where we have used the definition of $\mathbf \Gamma_{ij}$ in the last equality. Combining \cref{eq:r1s1-1,eq:r1s1-2} we arrive at
\begin{align}
T_{1,1} = \lim_{N\to\infty}\frac{N^2}{n^2d}\tbE_{\boldsymbol{\alpha}}\left[\sum_{i,j=1}^d\sqrt{\gamma_{i,d,n}\gamma_{j,d,n}}P_{i,j}^{(1,1)}\left(\mathbf \Gamma^{-1}\right)P^{(2,1)}_{i,j}(\tilde{ \mathbf A})\right]&=-\frac{d(d-1)}{2}.\label{eq:T11}
\end{align}
Collecting all the terms $T_{r,s}$ for $r+s\leq 2$ that we have calculated in \cref{eq:T00,eq:T10,eq:T01,eq:T20,eq:T02,eq:T11}, we arrive at
\begin{align}
	\lim_{N\to\infty}R_N&= \lim_{N\to\infty}\sum_{\substack{r,s\in{0,1,2}\\r+s\le 2}}\frac{N^2}{nd^2}\left(\frac{d}{n}\right)^{\frac{s+r}{2}} \tbE_{\boldsymbol{\alpha}}\left[\delta_{s0}\delta_{r0} \frac{d^2 n}{N}-\sum_{i,j=1}^d\sqrt{\gamma_{i,d,n}\gamma_{j,d,n}}P_{i,j}^{(1,s)}\left(\mathbf \Gamma^{-1}\right)P^{(2,r)}_{i,j}(\tilde{ \mathbf A})\right]\\
	&= T_{0,0} + T_{0,1} + T_{1,0} + T_{0,2} + T_{2,0} + T_{1,1}\\
	&=\frac{d(d+1)}{2}-1-\frac{3(d^2-1)}{4}+\frac{d(d-1)}{2}\\
	&=\frac{d^2-1}{4},
\end{align}
which implies that
\begin{align}
\lim_{N\to\infty}N\left(1-F^{\mathrm{std}}_d(N)\right)&=\frac{d^2-1}{4}.
\end{align}
To determine the lower order term, note that in all expressions above we have neglected terms of at most  $O(n^{-1/2+\varepsilon(d-2)})$. \cref{eq:lowerorderbound} shows that the terms with $r+s\ge 3$ are $O(n^{-1/2+3/2\eps})$, and the difference between $R_N$ and $N\left(1-F^{\mathrm{std}}_d(N)\right)$ is  $O(n^{-1/2+\varepsilon(d-2)})$ as well. As $\varepsilon\in(0,(d-2)^{-1})$ was arbitrary we conclude that, for all $\delta>0$,
\begin{align}
F^{\mathrm{std}}_d(N)&=1-\frac{d^2-1}{4N}+O(N^{-3/2+\delta}),
\end{align}
which concludes the proof.
\end{proof}

\subsection{Asymptotics of the optimal protocol}\label{sec:optimal}
In this section, our goal is to obtain an asymptotic lower bound on the optimal entanglement fidelity $F_d^*$ of a deterministic PBT protocol with both the entangled resource state and the POVM optimized. This is achieved by restricting the optimization in \cref{eq:cambridgeII} to the class of protocol families that use a density $c_\mu$ such that the probability distribution $q(\mu)=c_\mu p_{N,d}(\mu)$ converges for $N\to \infty$ in a certain sense. We then continue to show that the optimal asymptotic entanglement fidelity within this restricted class is related to the first eigenvalue of the Dirichlet Laplacian on the simplex of ordered probability distributions.

The main result of this section is the following theorem, which we restate from \Cref{sec:summary} for convenience.
\begin{thm-frenchguys}[restated]
	\restateFrenchguys
\end{thm-frenchguys}

For the proof of \Cref{thm:thefrenchguys} it will be convenient to switch back and forth between summation over a lattice and integration, which is the content of \Cref{lem:lattice-integration} below.
Before stating the lemma, we make a few definitions.
For a set $\Omega$ we define $d(x,\Omega)\coloneqq \min_{y\in\Omega} \|x-y\|_2$, and for $\delta\geq 0$ we define
\begin{align}
\partial_\delta \Omega \coloneqq \lbrace x\in\Omega\colon d(x,\partial\Omega)\leq \delta \rbrace.
\end{align}
Let $V_{0}^{d-1}=\{x\in\R^d|\sum_{i=1}^dx_i=0\}$ and $\Z^d_0=\Z^d\cap V_0^{d-1}$.
For a vector subspace $V\subset \mathbb{R}^d$ and lattice $\Lambda\subset\mathbb{R}^d$, we denote by $v+V$ and $v+\Lambda$ the affine space and affine lattice with the origin shifted to $v\in\R^d$, respectively.
We denote by $\lbrace e_i\rbrace_{i=1}^d$ the standard basis in $\mathbb{R}^d$.
For $y\in e_1+\frac 1 N \Z^d_0$, define $U_N(y)\subset e_1+V_0^{d-1}$ by the condition
\begin{align}
x\in U_N(y)\Leftrightarrow \forall y'\in e_1+\frac 1 N \Z^d_0, \|x-y\|_2<\|x-y'\|_2.
\end{align}
In other words, up to sets of measure zero we have tiled $e_1+V_0^{d-1}$ regularly into neighborhoods of lattice points. This also induces a decomposition $\OS_{d-1}\subset e_1+ V_0^{d-1}$ via intersection, $U_N^\OS(y)=U_N(y)\cap\OS_{d-1}$.
We define the function $g_{N}\colon e_1 + V_0^{d-1} \to e_1 + \frac{1}{N} \mathbb{Z}_0^d$ via $g_N(x) = y$ where $y$ is the unique lattice point such that $x\in U_N(y)$, if such a point exists. On the measure-zero set $\left(\bigcup_{y\in e_1 + \frac{1}{N} \mathbb{Z}_0^d}U_N(y)\right)^c$, the function $g_N$ can be set to an arbitrary value.

\begin{lem}\label{lem:lattice-integration}
  Let $f\in C^1(\OS_{d-1})\cap C(\mathbb{R}^d)$ be such that $f(x)= O(d(x,\partial\OS_{d-1})^p)$ for some $p\geq 1$, and $f\equiv 0$ on $\mathbb{R}^d\setminus\OS_{d-1}$.
  Then,
  \begin{align}
  \mathrm{(i)} &\quad \left| \frac{1}{N^{d-1}} \sum_{y\in \OS_{d-1}\cap \frac{1}{N}\mathbb{Z}^d} f(y) - \int_{\OS_{d-1}} f(g_N(x)) \D x \right| \leq O(N^{-p-2}); \label{eq:sum-to-int}\\
  \mathrm{(ii)} &\quad \left| \int_{\OS_{d-1}} f(g_N(x)) \D x - \int_{ \OS_{d-1}} f(x) \D x \right| \leq O(N^{-1}).\label{eq:int-to-int-g}\\
  \intertext{If furthermore $f\in C^2(\OS_{d-1})$, then}
  \mathrm{(iii)} &\quad \int_{\OS_{d-1}} f(g_N(x))(-\Delta f)(g_N(x)) \D x = \int_{\OS_{d-1}} f(x) (-\Delta f)(x) \D x + O(N^{-1}).\label{eq:partial-integration}
  \end{align}
\end{lem}

\begin{proof}
  Throughout the proof we set $\Lambda \coloneqq \OS_{d-1}\cap \frac{1}{N}\mathbb{Z}^d$.
  Observe first that the largest radius of the cell $U_N(y)$ around $y\in \frac{1}{N}\mathbb{Z}^d$ is equal to half the length $\frac{\sqrt{d}}{N}$ of a main diagonal in a $d$-dimensional hypercube of length $\frac{1}{N}$.
  Setting $c\coloneqq \frac{\sqrt{d}}{2}$, it follows that $g_N^{-1}(y)\subseteq\OS_{d-1}$ for all $y\in \Lambda$ with
  \begin{align}
  d(y,\partial\OS_{d-1}) > \frac{c}{N}. \label{eq:distance-from-boundary}
  \end{align}
  Hence, we can write
  \begin{align}
  \int_{ \OS_{d-1}} f(g_N(x)) \D x = \sum_{y\in\Lambda} \omega(y) f(y),
  \end{align}
  where $\omega(y)$ assigns the weight $N^{-d+1}$ to all $y\in\Lambda$ satisfying \eqref{eq:distance-from-boundary}, and $0\leq\omega(y)\leq N^{-d+1}$ for all $y\in\partial_{c/N}\OS_{d-1}$ to compensate for i) the fact that in this region $g_N$ maps some $x\in \OS_{d-1}$ to a lattice point outside of $\OS_{d-1}$, and ii) the fact that for some lattice points in $y\in \OS_{d-1}$, not all of the neighborhood of $y$ is contained in $\OS_{d-1}$, i.e. $U_N(y)\setminus \OS_{d-1}\neq\emptyset$.

  We bound
  \begin{align}
  \left| \sum_{y\in\Lambda} N^{-d+1} f(y) - \int_{\OS_{d-1}} f(g_N(x)) \D x \right| &=
  \left| \sum_{y\in\Lambda} (N^{-d+1}-\omega(y)) f(y)\right|\\
  &\leq \sum_{y\in\partial_{c/N}\OS_{d-1}} N^{-d+1} |f(y)|\\
  &\leq \sum_{y\in\partial_{c/N}\OS_{d-1}} N^{-d+1} \left(\frac{c}{N}\right)^p\\
  & \leq \frac{c}{N} C_d N^{d-2}N^{-d+1} \left(\frac{c}{N}\right)^p \\
  &= O(N^{-p-2}),
  \end{align}
  where in the second inequality we used the assumption $f\in O(d(x,\partial\OS_{d-1})^p)$, and in the third inequality we used that there are at most $\frac{c}{N} C_d N^{d-2}$ lattice points in $\partial_{c/N}\OS_{d-1}$ for some constant $C_d$ that only depends on $d$.
  This proves (i).

  In order to prove (ii), we first develop $f(g_N(x))$ into a Taylor series around a point $x$:
  \begin{align}
  f(g_N(x)) = f(x) + (g_N(x)-x)^T\nabla f(x) + O(N^{-1})
  \end{align}
  where we used the bound $\|g_N(x) - x\|_2\leq \frac{c}{N}$ for some constant $c$ for the remainder term in the Taylor series.
  Hence, we have
  \begin{align}
  \left| \int_{\OS_{d-1}} f(g_N(x)) \D x - \int_{ \OS_{d-1}} f(x) \D x \right| &\leq \int_{ \OS_{d-1} } \left| (g_N(x)-x)^T\nabla f(x) \right|\,\D x + O(N^{-1})\\
  &\leq \int_{ \OS_{d-1} } \|g_N(x)-x\|_2 \|\nabla f(x)\|_2 \,\D x + O(N^{-1})\\
  &\leq \frac{c}{N} K \vol(\OS_{d-1}) + O(N^{-1})\\
  &= O(N^{-1}),
  \end{align}
  where the second inequality follows from the Cauchy-Schwarz inequality, and in the third inequality we used the fact that by assumption $\|\nabla f(x)\|_2$ is a continuous function on the compact domain $\OS_{d-1}$ and therefore bounded by a constant $K$, proving (ii).

  Finally, we prove assertion (iii).
  We denote by $\partial_{ij}f \coloneqq (e_i-e_j)^T \nabla f$ the partial derivative of $f$ in the direction $e_{ij}\coloneqq e_i-e_j$.
  We approximate $\partial_{ij} f(x)$ using a central difference $D_{ij}[f(x)]\coloneqq f(x+\tfrac{h}{2} e_{ij}) - f(x-\tfrac{h}{2} e_{ij})$, where $h>0$ is to be chosen later.
  To this end, consider the Taylor expansions
  \begin{align}
  f(x+\tfrac{h}{2} e_{ij}) &= f(x) + \frac{h}{2} e_{ij}^T \nabla f(x) + O(h^2)\\
  f(x-\tfrac{h}{2} e_{ij}) &= f(x) - \frac{h}{2} e_{ij}^T \nabla f(x) + O(h^2).
  \end{align}
  Subtracting the second expansion from the first and rearranging gives
  \begin{align}
  \partial_{ij} f(x)= \frac{1}{h} D_{ij}[f(x)] + O(h).\label{eq:finite-difference-approximation}
  \end{align}
  It is easy to see that
  \begin{align}
  \sum_{i,j=1}^d e_{ij} e_{ij}^T = 2d 1_{V_0^{d-1}}, \label{eq:id-V}
  \end{align}
  and hence, for the Laplacian $\Delta = \tr(H(\cdot))$ on $V_0^{d-1}$ with $H(\cdot)$ the Hessian matrix, we have
  \begin{align}
  \Delta f(x) &= \tr( H(f)(x))\\
  &= \tr\left(1_{V_0^{d-1}}H(f)(x)\right)\\
  &= \frac{1}{2d} \sum_{i,j=1}^{d} e_{ij}^T H(f)(x) e_{ij}\\
  &= \frac{1}{2d} \sum_{i,j=1}^{d} \partial_{ij}^2 f(x).\label{eq:laplacian-overcomplete}
  \end{align}
  Similarly, denoting by $\langle \cdot,\cdot\rangle_{V_0^{d-1}}$ the inner product on $V_0^{d-1}$, we have
  \begin{align}
  \langle \nabla f(x) ,\nabla f(x) \rangle_{V_0^{d-1}} &= \frac{1}{2d} \sum_{i,j=1}^d \langle \nabla f(x) ,e_{ij} e_{ij}^T \nabla f(x) \rangle_{V_0^{d-1}}\\
  &= \frac{1}{2d} \sum_{i,j=1}^d \left(e_{ij}^T \nabla f(x)\right)^2\\
  &= \frac{1}{2d} \sum_{i,j=1}^d (\partial_{ij}f(x))^2.\label{eq:nabla-overcomplete}
  \end{align}

  We now calculate, abbreviating $\sum\nolimits_{y\in\Lambda}' = \sum_{y\in\Lambda} \omega(y)$:
  \begin{align}
  &\int_{\OS_{d-1}} f(g_N(x))(-\Delta f)(g_N(x)) \D x \\
  &\qquad\qquad {} = \ssum_{y\in\Lambda} f(y) (-\Delta f)(y)\\
  &\qquad\qquad {} = -  \frac{1}{2d} \ssum_{y\in\Lambda} \sum_{i,j=1}^{d} f(y) \partial_{ij}^2 f(y)\\
  &\qquad\qquad {} = \frac{1}{2d} \ssum_{y\in\Lambda} \sum_{i,j=1}^d (\partial_{ij}f(y))^2 - \frac{1}{2d} \ssum_{y\in\Lambda} \sum_{i,j=1}^{d} \partial_{ij} \left[ f(y) \partial_{ij} f(y) \right]\\
  &\qquad\qquad {} = \ssum_{y\in\Lambda} \langle \nabla f(y), \nabla f(y) \rangle_{V_0^{d-1}} - \frac{1}{2dh} \sum_{i,j=1}^{d} \ssum_{y\in\Lambda} D_{ij}[f(y)\partial_{ij} f(y)] + O(h), \label{eq:summation-by-parts}
  \end{align}
  where we used \eqref{eq:laplacian-overcomplete} in the second equality, and \eqref{eq:nabla-overcomplete} and \eqref{eq:finite-difference-approximation} in the last equality.

  For the first term in \eqref{eq:summation-by-parts}, we have
  \begin{align}
  \ssum_{y\in\Lambda} \langle \nabla f(y), \nabla f(y) \rangle_{V_0^{d-1}} &= \int_{\OS_{d-1}} \langle \nabla f(g_N(x)), \nabla f(g_N(x)) \rangle_{V_0^{d-1}} \D x\\
  &= \int_{\OS_{d-1}} \langle \nabla f(x), \nabla f(x) \rangle_{V_0^{d-1}} \D x + O(N^{-1})\\
  &= \int_{\OS_{d-1}} f(x) (-\Delta f(x)) \D x + O(N^{-1}),
  \end{align}
  where the second equality follows from (ii), and the third equality is ordinary integration by parts.
  For the second term in \eqref{eq:summation-by-parts}, we use the definition of $D_{ij}$ to obtain
  \begin{multline}
  \frac{1}{2dh} \sum_{i,j=1}^d \ssum_{y\in\Lambda} D_{ij}[f(y)\partial_{ij} f(y)] + O(h) \\
  = \frac{1}{2dh} \sum_{i,j=1}^d \ssum_{y\in\Lambda} \left(f(y+\tfrac{h}{2}e_{ij})\partial_{ij} f(y+\tfrac{h}{2}e_{ij}) - f(y-\tfrac{h}{2}e_{ij})\partial_{ij} f(y-\tfrac{h}{2}e_{ij}) \right) + O(h).\label{eq:extra-term}
  \end{multline}
  We choose $h=O(N^{-1})$ such that $y\pm \tfrac{h}{2} e_{ij}\in \Lambda$ for all $y\in\Lambda$ sufficiently far away from the boundary of $\Lambda$.
  Then all terms in \eqref{eq:extra-term} cancel except for those terms involving evaluations of $f$ on $\partial_{h}\OS_{d-1}$ or outside $\OS_{d-1}$.
  But these terms in turn are $O(h)=O(N^{-1})$, which can be seen using the same arguments as those in the proof of (ii).
  It follows that, with the above choice of $h=O(N^{-1})$,
  \begin{align}
  \frac{1}{2dh} \sum_{i,j=1}^d \ssum_{y\in\Lambda} D_{ij}[f(y)\partial_{ij} f(y)] + O(h) = O(N^{-1}).
  \end{align}
  In summary, we have shown that
  \begin{align}
  \int_{\OS_{d-1}} f(g_N(x))(-\Delta f)(g_N(x)) \D x = \int_{\OS_{d-1}} f(x) (-\Delta f(x)) \D x + O(N^{-1}),
  \end{align}
  which is what we set out to prove.
\end{proof}

We are now ready to prove \Cref{thm:thefrenchguys}:

\begin{proof}[Proof of Thm.~\ref{thm:thefrenchguys}] Fix a dimension $d$, and let $a\in C^2(\OS_{d-1})$ be twice continuously differentiable\footnote{The second derivative is continuous and its limit for the argument approaching the boundary exists.} such that $a|_{\partial\OS_{d-1}}\equiv 0$, $a(x)\ge 0$ for all $x\in\OS_{d-1}$, and $\|a\|_2=1$, where $\|\cdot\|_2$ is the $L_2$-norm on $\OS_{d-1}$. As $d$ is fixed throughout the proof, we omit indicating any dependence on $d$ except when we would like to emphasize the dimension of an object. Note that clearly $a\in L_2(\OS_{d-1})$ as $a$ is continuous and $\OS_{d-1}$ is compact.

  We use the square of a scaled version of $a$ as a candidate probability distribution $q$ on Young diagrams $\mu$ with $N$ boxes and at most $d$ rows,
  \begin{align}
  q(\mu)=\frac{\eta_{N}}{N^{d-1}}a^2\left(\frac{\mu}{N}\right).\label{eq:q}
  \end{align}
  Here $\eta_{N}$ is a normalization constant which is close to one. Roughly speaking, this is due to the fact that the normalization condition for $q(\mu)$ is essentially proportional to a Riemann sum for the integral that calculates the $L_2$-norm of $a$, which is equal to unity by assumption.
  Indeed, since $a^2$ satisfies the assumptions of \Cref{lem:lattice-integration} with $p=1$, we have
  \begin{align}
  1&=\sum_{\mu\vdash_d N}q(\mu)\\
  &=\frac{\eta_{N}}{N^{d-1}}\sum_{\mu\in\Z^d\cap N\OS_{d-1}}a^2\left(\frac{\mu}{N}\right)\\
  &=\frac{\eta_{N}}{N^{d-1}}\sum_{y\in\left(\frac 1 N \Z^d\right)\cap \OS_{d-1}}a^2\left(y\right)\\
  &= \eta_N\left( \int_{ \OS_{d-1} } a^2(g_N(y)) \D y + O(N^{-3})\right)\\
  &= \eta_N\left( \int_{ \OS_{d-1} } a^2(y) \D y + O(N^{-1})\right)\\
  &= \eta_N \left(1 + O(N^{-1}) \right),
  \end{align}
  where the fourth and fifth equality follow from \Cref{lem:lattice-integration}(i) and (ii), respectively, and the last equality follows from $\|a\|_2=1$.
  Hence, $\eta_N = 1 + O(N^{-1})$.

  Before we proceed, we restate the fidelity formula in \eqref{eq:cambridgeII} for the optimal deterministic protocol for the reader's convenience:
  \begin{align}
  F^{*}_d(N) =d^{-N-2}\max_{c_\mu}\sum_{\alpha\vdash_d N-1}\left(\sum_{\mu=\alpha+\square}\sqrt{c_\mu d_\mu m_{d,\mu}}\right)^2.\label{eq:cambridgeII-restated}
  \end{align}
  We bound this expression from below by choosing $c_\mu=q(\mu)/p(\mu)$, where $q(\mu)$ is defined as in \eqref{eq:q} and $p(\mu)=\frac{d_\mu m_{d,\mu}}{d^N}$ is the Schur-Weyl distribution.
  The choice of $c_\mu$ in \eqref{eq:cambridgeII-restated} corresponds to a particular PBT protocol whose entanglement fidelity we denote by $F_a$ in the following.
  It will be convenient to rewrite the sums over Young diagrams $\alpha\vdash_d N-1$ and $\mu=\alpha+\square$ in \eqref{eq:cambridgeII-restated} as a sum over Young diagrams $\mu\vdash_d N$ and $i,j=1,\dots, d$, requiring that both $\mu+e_i-e_j$ and $\mu - e_j$ be Young diagrams themselves.
  Using this trick, the quantity $\frac{d^2}{\eta_N} F_a$ can be expressed as
  \begin{align}
  \frac{d^2}{\eta_N} F_a &=N^{-d+1} \sum_{\mu\vdash_{d}N}a\left(\frac{\mu}{N}\right)\sum_{i,j=1}^d\mathds 1_{\mathrm{YD}}(\mu+e_i-e_j) \mathds{1}_{\mathrm{YD}}(\mu-e_j) a\left(\frac{\mu+e_i-e_j}{N}\right)\\
  &= N^{-d+1} \sum_{\mu\vdash_{d}N}a\left(\frac{\mu}{N}\right)\sum_{i,j=1}^da\left(\frac{\mu+e_i-e_j}{N}\right)\\
  &\qquad {} + N^{-d+1} \sum_{\mu\vdash_{d}N}a\left(\frac{\mu}{N}\right)\sum_{i,j=1}^d \mathds 1_{\mathrm{YD}}(\mu+e_i-e_j) \mathds{1}_{\mathrm{YD}}(\mu-e_j) a\left(\frac{\mu+e_i-e_j}{N}\right)\\
  &\qquad {} - N^{-d+1} \sum_{\mu\vdash_{d}N}a\left(\frac{\mu}{N}\right)\sum_{i,j=1}^da\left(\frac{\mu+e_i-e_j}{N}\right).\label{eq:splitting-up}
  \end{align}

  We first argue that up to order $N^{-2}$ we only need to consider the first term in the above expression.
  To this end, we rewrite the sum in the second term as an integral,
  \begin{multline}
  N^{-d+1} \sum_{\mu\vdash_{d}N}a\left(\frac{\mu}{N}\right)\sum_{i,j=1}^d f_{i,j}(\mu)a\left(\frac{\mu+e_i-e_j}{N}\right) \\ = \int_{ \OS_{d-1} } h_N(x) a(g_N(x))\sum_{i,j=1}^d f_{i,j}(x) a\left(g_N(x) + \frac{e_i-e_j}{N}\right) \D x,
  \end{multline}
  where $f_{i,j}(x)\coloneqq \mathds 1_{\mathrm{YD}}(Ng_N(x)+e_i-e_j) \mathds{1}_{\mathrm{YD}}(Ng_N(x)-e_j)$.
  The function $h_N(x)\in [0,1]$ takes care of normalization around the boundaries of $\OS_{d-1}$, that is, $h_N(x)=1$ except in a region $\partial_{c_1/N}\OS_{d-1}$ for some constant $c_1$ that only depends on $d$.
  Note that the same statement is true for the function $f_{i,j}(x)$, and therefore, this also holds for the product $h_N(x)f_{i,j}(x)$.
  Using \Cref{lem:lattice-integration}(i) for the third term in \eqref{eq:splitting-up} gives
  \begin{multline}
  N^{-d+1} \sum_{\mu\vdash_{d}N}a\left(\frac{\mu}{N}\right)\sum_{i,j=1}^da\left(\frac{\mu+e_i-e_j}{N}\right) \\ = \int_{\OS_{d-1}} a(g_N(x))\sum_{i,j=1}^d a\left(g_N(x) + \frac{e_i-e_j}{N} \right) \D x + O(N^{-3}).
  \end{multline}
  Hence, for the difference of the second and third term in \eqref{eq:splitting-up}, we obtain
  \begin{align}
  &N^{-d+1} \sum_{\mu\vdash_{d}N}a\left(\frac{\mu}{N}\right)\sum_{i,j=1}^d \left[\mathds 1_{\mathrm{YD}}(\mu+e_i-e_j) \mathds{1}_{\mathrm{YD}}(\mu-e_j) -1 \right] a\left(\frac{\mu+e_i-e_j}{N}\right)\\
  &\qquad{} = \int_{\OS_{d-1}}a(g_N(x))\sum_{i,j=1}^d\left[h_N(x)f_{i,j}(x)-1\right] a\left(g_N(x)+\frac{e_i-e_j}{N}\right)\D x + O(N^{-3})\\
  &\qquad{} \le \frac{c_2}{N^2}\int_{\partial_{c_1/N}\OS_{d-1}}\left[ h_N(x)\mathds 1_{\mathrm{YD}}(Ng_N(x)+e_i-e_j) \mathds{1}_{\mathrm{YD}}(Ng_N(x)-e_j) -1\right] \D x + O(N^{-3})\\
  &\qquad{} \le \frac{c_3}{N^2}\vol(\partial_{c_1/N}\OS_{d-1}) + O(N^{-3})\\
  &\qquad{} = O(N^{-3})\label{eq:boundary-int}
  \end{align}
  for some constants $c_2$ and $c_3$.
  Here, the first inequality is obtained by a Taylor expansion of the different occurrences of $a$ around the respective closest boundary point and using the fact that $a$ vanishes on the boundary by assumption.
  The second inequality follows since $h_N$ is bounded uniformly in $N$.\footnote{Observe that a constant fraction of $U_N(y)$ of each lattice point $y\in\OS_{d-1}\cap\frac{1}{N}\Z^d$ lies inside $\OS_{d-1}$. This fraction is not uniformly bounded in $d$, as the solid angle of the vertices of $\OS_{d-1}$ decreases with $d$. However, this does not concern us, since we are only interested in the limit $N\to\infty$ for fixed $d$.}

  We now turn to the first term in \eqref{eq:splitting-up}, applying \Cref{lem:lattice-integration}(i) once more to obtain
  \begin{multline}
  N^{-d+1} \sum_{\mu\vdash_{d}N}a\left(\frac{\mu}{N}\right)\sum_{i,j=1}^da\left(\frac{\mu+e_i-e_j}{N}\right) =\\  \int_{ \OS_{d-1} } a(g_N(x)) \sum_{i,j=1}^d a\left(g_N(x)+ \frac{e_i-e_j}{N}\right)\D x + O(N^{-3}).
  \end{multline}
  Expanding $a\left(g_N(x)+\frac{e_i-e_j}{N}\right)$ into a Taylor series gives
  \begin{align}
  a\left(g_N(x)+\frac{e_i-e_j}{N}\right)&=a(g_N(x))+\frac 1 N \langle e_i-e_j, (\nabla a)(g_N(x))\rangle_{V_0^{d-1}}\\
  &\quad\quad+\frac{1}{2N^2}\tr\left[(e_i-e_j)(e_i-e_j)^T (H(a))(g_N(x))\right]+O(N^{-3}),
  \end{align}
  where $\langle\cdot,\cdot\rangle_{V_0^{d-1}}$ is the standard inner product on $V_0^{d-1}$ and $H(a)$ denotes the Hessian of $a$ on $V_0^{d-1}$.
  Summing over $i$ and $j$ yields
  \begin{align}
  \sum_{i,j=1}^d e_i-e_j &=0 & \sum_{i,j=1}^d(e_i-e_j)(e_i-e_j)^T &=2d  1_{V_0^{d-1}}.\label{eq:summation-identities}
  \end{align}
  It follows that
  \begin{align}
  &\int_{\OS_{d-1}}a(g_N(x))\sum_{i,j=1}^da\left(g_N(x)+\frac{e_i-e_j}{N}\right)\D x +O(N^{-3})\\
  &\qquad {} =\int_{\OS_{d-1}}a(g_N(x))\left(d^2a(g_N(x))+\frac{d}{N^2}(\Delta a)(g_N(x))\right)\D x+O(N^{-3})\\
  &\qquad {} = \frac{d^2}{N^{d-1}} \sum_{y\in \OS_{d-1}\cap \frac{1}{N}\mathbb{Z}^d} a^2(y) -\frac{d}{N^2}\int_{\OS_{d-1}}a(g_N(x))(-\Delta a)(g_N(x))\D x+O(N^{-3})\\
  &\qquad {} =\frac{d^2}{\eta_{N}}-\frac{d}{N^2}\int_{\OS_{d-1}}a(g_N(x))(-\Delta a)(g_N(x))\D x+O(N^{-3})\\
  &\qquad {} =\frac{d^2}{\eta_{N}}-\frac{d}{N^2}\int_{\OS_{d-1}}a(x)(-\Delta a)(x)\D x+O(N^{-3}),
  \end{align}
  where in the first equality the $N{-1}$ term vanishes due to \eqref{eq:summation-identities}, and we defined the Laplace operator $\Delta(a)=\tr H(a)$ on $V_0^{d-1}$.
  In the second equality we used \Cref{lem:lattice-integration}(i) to switch back to discrete summation, in the third equality we used the normalization of $a$, and in the fourth equality we used \Cref{lem:lattice-integration}(iii).

  Putting together everything we have derived so far, we obtain
  \begin{align}
  F_a=1-\frac{1}{dN^2}\int_{\OS_{d-1}}a(x)(-\Delta a)(x)\D x+O(N^{-3}).
  \end{align}
  In equation \cref{eq:cambridgeII}, the fidelity is maximized over all densities $c_\mu$. The above expression shows, that restricting to the set of densities $c_\mu$ that stem from a function $a$ on $\OS_{d-1}$ makes the problem equivalent to minimizing the expression
  \begin{align}
  \int_{\OS_{d-1}}a(x)(-\Delta a)(x)\D x.
  \end{align}
  When taking the infimum over $a\in H^2(\OS_{d-1})$, where $H^2(\OS_{d-1})$ is the Sobolev space of twice weakly differentiable functions, instead of $a\in C^2(\OS_{d-1})$, this is exactly one of the variational characterizations of the first Dirichlet eigenvalue of the Laplace operator on $\OS_{d-1}$. This is because the eigenfunction corresponding to the first eigenvalue of the Dirichlet Laplacian can be chosen positive (see, e.g.,~\cite{Grebenkov2013}). But $C^2(\OS_{d-1})$ is dense in $H^2(\OS_{d-1})$, which implies that
  \begin{align}\label{eq:Fidfromeigenvalue}
  \sup_a F_a&=1-\frac{\lambda_1(\OS_{d-1})}{dN^2}+O(N^{-3}),
  \end{align}
  where the supremum is taken over all non-negative functions $a\in C^2(\OS_{d-1})$.
\end{proof}

Upper and lower bounds for the first Dirichlet eigenvalue of the Laplacian on a sufficiently well-behaved domain readily exist.
\begin{thm}[\cite{Krahn1926,Freitas2008}]\label{thm:eigenvaluebounds}
  For the first Dirichlet eigenvalue $\lambda_1(\Omega)$ on a bounded convex domain $\Omega\subset \R^d$, the following inequalities hold,
  \begin{align}
  \lambda_1(\Omega)&\ge\lambda_1(B_1)\left(\frac{\vol(B_1)}{\vol(\Omega)}\right)^\frac{2}{d}\text{, and}\\
  \lambda_1(\Omega)&\le \lambda_1(B_1)\frac{\vol(\partial\Omega)}{dr_\Omega\vol(\Omega)},
  \end{align}
  where $B_1\subset \R^d$ is the unit ball and $r_\Omega$ is the inradius of $\Omega$.
\end{thm}

The inradius of $\OS_{d-1}$ is equal to $1/d^2$. This can be seen by guessing the center of the inball $\hat x=((2d-1)/d^2,(2d-3)/d^2,\ldots,1/d^2)$ and checking that the distance to each facet is $1/d^2$. Therefore we get the following lower bound on the optimal PBT fidelity.
This theorem is stated in \Cref{sec:summary}, and restated here for convenience.
\begin{thm-frenchguys-summary}[restated]
	\restateFrenchguysSummary
\end{thm-frenchguys-summary}
\begin{proof}
  \Cref{thm:thefrenchguys} gives us the bound
  \begin{align}
  F^*_d(N)\ge 1-\frac{\lambda_1(\OS_{d-1})}{dN^2}+O(N^{-3}).
  \end{align}
  Using \cref{thm:eigenvaluebounds} and \cref{lem:simpvol} we bound
  \begin{align}
  \lambda_1(\OS_{d-1})&\le \lambda_1(B_1^{d-1})\frac{\vol(\partial\Omega)}{dr_\Omega\vol(\Omega)}\\
  &\le \lambda_1(B_1^{d-1})d^2\left(\frac{d(d-1)}{\sqrt 2}+\sqrt{d(d-1)}+\sqrt 2\right).
  \end{align}
  The first eigenvalue of the Dirichlet Laplacian on the $(d-1)$-dimensional Ball is given by
  \begin{align}
  \lambda_1(B_1^{d-1})=j^2_{\frac{d-3} 2,1},
  \end{align}
  where $j_{\nu,l}$ is the $l$th root of the Bessel function of the first kind with parameter $\nu$. This is, in turn, bounded as~\cite{Chambers1982}
  \begin{align}
  j_{\nu,1}\le \sqrt{\nu+1}(\sqrt{\nu+2}+1).
  \end{align}
  Putting the inequalities together we arrive at
  \begin{align}
  \lambda_1(B_1^{d-1})&\le \frac{d-1}{2}\left(\sqrt{\frac{d+1}{2}}+1\right)^2,\\
  \lambda_1(\OS_{d-1})&\le \frac{d-1}{2}\left(\sqrt{\frac{d+1}{2}}+1\right)^2d^2\left(\frac{d(d-1)}{\sqrt 2}+\sqrt{d(d-1)}+\sqrt 2\right), \text{ and hence}\\
  F^*_d(N)&\ge 1-\frac{ \frac{d-1}{2}\left(\sqrt{\frac{d+1}{2}}+1\right )^2d\left(\frac{d(d-1)}{\sqrt 2}+\sqrt{d(d-1)}+\sqrt 2\right)}{N^2}+O(N^{-3})\\
  &=1-\frac{ d^5+O(d^{9/2})}{4\sqrt 2 N^2}+O(N^{-3}).
  \end{align}
\end{proof}
In the appendix, we provide a concrete protocol in \cref{thm:PrettyGoodProtocol} that achieves the same asymptotic dependence on $N$ and $d$, with a slightly worse constant.

Intuitively it seems unlikely that a ``wrinkly'' distribution, i.e.\ a distribution that does not converge against an $L_1$ density on $\OS$, is the optimizer in \cref{eq:cambridgeII}. Supposing that the optimizer comes from a function $a$ as described above, we can also derive a converse bound for the asymptotics of the entanglement fidelity $F^*_d(N)$ using \cref{thm:eigenvaluebounds}.
\begin{rem}
  Let $P^N_a$ be the PBT protocol with $c_\mu=N^{d-1}a^2(\mu/N)/P(\mu)$ for some function $a\in L_2(\OS_{d-1})$. For the asymptotic fidelity of such protocols for large $N$ the following converse bound holds,
  \begin{align}
  F_{a}(N)\le 1-\frac{\pi d^4+O(d^{3})}{8e^3N^2}+O(N^{-3}).
  \end{align}
  This can be seen as follows. From \cref{thm:thefrenchguys} we have that
  \begin{align}
  F_a\le 1-\frac{\lambda_1(\OS_{d-1})}{dN^2}+O(N^{-3}).
  \end{align}
  \Cref{thm:eigenvaluebounds} together with \cref{lem:simpvol} yields
  \begin{align}
  \lambda_1(\OS_{d-1})&\ge \lambda_1(B_1)\left(\frac{\vol(B_1^{d-1})}{\vol(\OS_{d-1})}\right)^\frac{2}{d}\\
  &=\lambda_1(B_1)\left(\frac{\pi^{\frac{d-1}{2}}\sqrt d((d-1)!)^2}{\Gamma(\frac {d-1} 2+1)}\right)^\frac{2}{d}\\
  &\ge \pi^{1-1/d}\lambda_1(B_1)\left(\frac{((d-1)!)^2}{\Gamma(\frac {d-1} 2+1)}\right)^\frac{2}{d}
  \end{align}
  where in the second line we have used the volume of the $(d-1)$-dimensional Ball,
  \begin{align}
  \vol(B_1^{d-1})=\frac{\pi^{\frac{d-1}{2}}}{\Gamma(\frac {d-1} 2+1)},
  \end{align}
  and $\Gamma(x)$ is the gamma function. Using bound versions of Stirling's approximation we obtain
  \begin{align}
  \lambda_1(\OS_{d-1})&\ge O(1)\lambda_1(B_1)\left(\frac{d-1}{e}\right)^{3(1-1/d)}
  \end{align}
  Using a  lower bound for the first zero of the Bessel function of the first kind~\cite{Breen1995} we bound
  \begin{align}
  \lambda_1(B_1^{d-1})&\ge\left(\frac{d} 2+c\right)^2
  \end{align}
  for some constant $c$, so we finally arrive at
  \begin{align}
  F_a
  &=1-\frac{\pi d^4+O(d^{3})}{8e^3N^2}+O(N^{-3})
  \end{align}
  This bound has the nice property that $N\propto d^2$ if the error of the PBT protocol is fixed, which is what we expect from information theoretic insights (see \cref{sec:converse}).
\end{rem}

\section{Converse Bound}\label{sec:converse}
We begin by deriving a lower bound on the communication requirements for approximate quantum teleportation of any kind, i.e., not only for PBT. Such a result could be called folklore, but has, to the best of our knowledge, not appeared elsewhere.\footnote{Except in the PhD thesis of one of the authors~\cite{Majenz2017a}.}

For the proof we need the converse bound for one-shot quantum state splitting that was given in~\cite{Berta2011} in terms of the smooth max-mutual information $I_{\max}^{\varepsilon}(E:A)_{\rho}$.
To define this quantity, let $D_{\max}(\rho\|\sigma)=\min\left\{\lambda\in\R\big|2^\lambda\sigma\ge\rho\right\}$ be the max-relative entropy~\cite{Datta2009}, and let $P(\rho,\sigma)\coloneqq \sqrt{1-F(\rho,\sigma)}$ be the purified distance.
Furthermore, let $B_\varepsilon(\rho)\coloneqq \lbrace \bar{ \rho}\colon \bar{ \rho}\geq 0, \tr\bar{ \rho}\leq 1, P(\rho,\bar{ \rho})\leq \varepsilon\rbrace$ be the $\varepsilon$-ball of subnormalized states around $\rho$ with respect to the purified distance.
The smooth max-mutual information is defined as
\begin{align}
I_{\max}^{\varepsilon}(E:A)_{\rho}&\coloneqq \min_{\bar{\rho}\in B_\varepsilon(\rho)}I_{\max}(E:A)_{\bar{\rho}},\label{eq:Imax1}
\end{align}
where $I_{\max}(E:A)_{\bar{\rho}}\coloneqq\min_{\sigma_A} D_{\max}(\bar\rho_{AE}\|\sigma_A\otimes\bar{\rho}_E)$ with the minimization over normalized quantum states $\sigma_A$.

\begin{lem}\label{lem:Imaxbound}
  Let
  \begin{align}
  \ket{\phi^+}_{AB}=\frac{1}{\sqrt d}\sum_{i=0}^{d-1}\ket{ii}_{AB}\in\hi_A\otimes \hi_B
  \end{align}
  be the $d\times d$-dimensional maximally entangled state. Then
  \begin{align}
  2\log\left\lceil d(1-\varepsilon^2)\right\rceil \ge I_{\max}^{\varepsilon}(A:B)_{\phi^+}\ge 2\log\left( d(1-\varepsilon^2)\right).
  \end{align}
\end{lem}

\begin{proof}

  Let $\rho\in B(\hi_A\otimes\hi_{B})$ be a quantum state such that $I_{\max}^{\varepsilon}(A:B)_{\phi^+}=I_{\max}(A:B)_{\rho}$, and let $\ket{\gamma}_{ABE}$ be a purification of $\rho$. Uhlmann's Theorem ensures that there exists a pure quantum state $\ket{\alpha}_E$ such that
  \begin{align}
  \sqrt{1-\varepsilon^2}\le \sqrt F(\phi^+,\rho)=\bra{\phi^+}_{AB}\bra{\alpha}_E\ket{\gamma}_{ABE}.\label{eq:uhlmann}
  \end{align}
  This holds without taking the absolute value because any phase can be included in $\ket{\alpha}$.
  Let
  \begin{align}
  \ket{\gamma}_{ABE}=\sum_{i=0}^{d-1}\sqrt{p_i}\ket{\phi_i}_{A}\otimes\ket{\psi_i}_{BE}\label{eq:gamma-schmidt}
  \end{align}
  be the Schmidt decomposition of $\ket{\gamma}$ with respect to the bipartition $A:BE$. Let further $U_A$ be the unitary matrix such that $U_A\ket{i}_A=\ket{\phi_i}_{A}$. Using the Mirror Lemma~\ref{lem:mirror-lemma} we get
  \begin{align}
  \ket{\phi^+}_{AB}&=U_AU_A^\dagger \ket{\phi^+}_{AB}\\
  &=U_A\bar U_B\ket{\phi^+}_{AB}\\
  &=\frac{1}{\sqrt d}\sum_{i=0}^{d-1}\ket{\phi_i}_{A}\ket{\xi_i}_{B},
  \end{align}
  where $\bar{U}$ is the complex conjugate in the computational basis and $\ket{\xi_i}_{B}=\bar U_B\ket{i}_B$. With this we obtain from \eqref{eq:uhlmann} that
  \begin{align}
  1-\varepsilon^2&\le (\bra{\phi^+}_{AB}\bra{\alpha}_E\ket{\gamma}_{ABE})^2\\
  &= (\Re \bra{\phi^+}_{AB}\bra{\alpha}_E\ket{\gamma}_{ABE} )^2\\
  &=\left(\sum_{i=0}^{d-1}\sqrt{\frac{p_i}{d}}\Re \bra{\xi_i}_B\bra{\alpha}_E\ket{\psi_i}_{BE}\right)^2\\
  &\le\frac 1 d\sum_{i=0}^{d-1}\left(\Re \bra{\xi_i}_B\bra{\alpha}_E\ket{\psi_i}_{BE}\right)^2\\
  &\le\frac 1 d\sum_{i=0}^{d-1}\Re\bra{\xi_i}_B\bra{\alpha}_E\ket{\psi_i}_{BE}.
  \end{align}
  The second inequality is the Cauchy-Schwarz inequality and the third inequality follows from $\Re\bra{\xi_i}_B\bra{\alpha}_E\ket{\psi_i}_{BE}\le 1 $.

  The next step is to bound the max-mutual information of $\rho$. Let
  \begin{align}
  \lambda=I_{\max}(A:B)_\rho=I_{\max}^{\varepsilon}(A:B)_{\phi^+}.
  \end{align}
  By the definition of $I_{\max}$ there exists a quantum state $\sigma_B$ such that
  \begin{align}
  2^\lambda=\left\|\rho_A^{-{\frac 1 2}}\otimes\sigma_B^{-{\frac 1 2}}\rho_{AB} \, \rho_A^{-{\frac 1 2}}\otimes\sigma_B^{-{\frac 1 2}}\right\|_\infty.
  \end{align}
  Here, $X^{-1}$ denotes the pseudo-inverse of a matrix $X$, i.e., $X^{-1}X=XX^{-1}$ is equal to the projector onto the support of $X$.
  Let $\ket{\phi_\sigma}=\sqrt d\sigma_B^{1/2}\ket{\phi^+}$ be the standard purification of $\sigma$. We bound
  \begin{align}
  2^\lambda&=\left\|\rho_A^{-{\frac 1 2}}\otimes\sigma_B^{-{\frac 1 2}}\rho_{AB} \, \rho_A^{-{\frac 1 2}}\otimes\sigma_B^{-{\frac 1 2}}\right\|_\infty\\
  &\ge\bra{\phi_\sigma}\rho_A^{-{\frac 1 2}}\otimes\sigma_B^{-{\frac 1 2}}\rho_{AB} \, \rho_A^{-{\frac 1 2}}\otimes\sigma_B^{-{\frac 1 2}}\ket{\phi_\sigma}\\
  &=\tr\bra{\phi_\sigma}\rho_A^{-{\frac 1 2}}\otimes\sigma_B^{-{\frac 1 2}}\proj{\gamma}_{ABE}\rho_A^{-{\frac 1 2}}\otimes\sigma_B^{-{\frac 1 2}}\ket{\phi_\sigma}\\
  &\ge\bra{\phi_\sigma}_{AB}\bra{\alpha}_E\rho_A^{-{\frac 1 2}}\otimes\sigma_B^{-{\frac 1 2}}\proj{\gamma}_{ABE}\rho_A^{-{\frac 1 2}}\otimes\sigma_B^{-{\frac 1 2}}\ket{\phi_\sigma}_{AB}\ket{\alpha}_E\\
  &=d\left|\bra{\phi^+}_{AB}\rho_A^{-{\frac 1 2}}\bra{\alpha}_E\ket{\gamma}_{ABE}\right|^2\\
  &=\left|\sum\nolimits_i\bra{\xi_i}_B\bra{\alpha}_E\ket{\psi_i}_{BE}\right|^2\\
  &\ge\left(\sum\nolimits_i\Re\bra{\xi_i}_B\bra{\alpha}_E\ket{\psi_i}_{BE}\right)^2\\
  &\ge d^2(1-\varepsilon^2)^2,
  \end{align}
  where we used the particular form of $\ket{\phi_\sigma}$ in the third equality, and \eqref{eq:gamma-schmidt} in the fourth equality, together with the fact that $\lbrace p_i\rbrace_i$ are the eigenvalues of $\rho_A$.
  This proves the claimed up upper bound on $I_{\max}^{\varepsilon}(A:B)_{\phi^+}$.

  In order to prove the lower bound, let ${r}=\lceil d(1-\varepsilon^2)\rceil$ and
  \begin{align}
  \ket{\phi^+_{r}}=\frac{1}{\sqrt{{r}}}\sum_{i=0}^{{r}-1}\ket{ii}_{AB}\in\hi_A\otimes \hi_B.
  \end{align}
  Then we have
  \begin{align}
  {I_{\max}(A:B)_{\phi^+_{r}}}&= 2\log {r}=2\log\lceil d(1-\varepsilon^2)\rceil\\
  |\langle\phi^+|\phi^+_{r}\rangle|^2&= {r}/d\ge 1-\varepsilon^2.
  \end{align}
  The observation that $\proj{\phi^+_{r}}$ is a point in the minimization over $\sigma$ finishes the proof.
\end{proof}

Using the special case of state merging/splitting with trivial side information and the converse bound from~\cite{Berta2011}, we can bound the necessary quantum communication for simulating the identity channel with a given entanglement fidelity.

\begin{cor}\label{cor:sim-ident}
  Let $\mathcal E_{AA'\to B}$, $\mathcal D_{BB'\to A}$ be quantum (encoding and decoding) channels with $\dim\hi_A=d$ and $\dim\hi_B= {d'}$ such that there exists a resource state $\rho_{A'B'}$ achieving
  \begin{align}
  F(\mathcal D\circ \mathcal E((\cdot)\otimes\rho_{A'B'}))=1-\varepsilon^2.
  \end{align}
  Then the following inequality holds:
  \begin{align}
  {d'}\ge d\left(1-\varepsilon^2\right).
  \end{align}
\end{cor}
\begin{proof}
  Using \cref{lem:Imaxbound}, this follows from applying the lower bound on the communication cost of one-shot state splitting from~\cite{Berta2011} to the special case where Alice and the reference system share a maximally entangled state.
\end{proof}

Together with superdense coding this implies a lower bound on approximate teleportation.

\begin{cor}\label{cor:telebound}
  If in the above corollary $\mathcal E$ is a $qc$-channel, then
  \begin{align}
  {d'}\ge d^2\left(1-\varepsilon^2\right)^2.
  \end{align}
\end{cor}
\begin{proof}
  This follows as any protocol with a lower classical communication in conjunction with superdense coding would violate \cref{cor:sim-ident}.
\end{proof}

For the special case of port-based teleportation, this implies a lower bound on the number of ports.

\begin{cor}\label{cor:porttelebound}
  Any port-based teleportation protocol with input dimension $d$ and $N$ ports has entanglement fidelity at most
  \begin{align}
  F^*_d(N) \le\frac{\sqrt{N}}{d}.
  \end{align}
\end{cor}
\begin{proof}
  In port-based teleportation, the only information that is useful to the receiver is which port to select. More precisely, given a protocol $P$ for PBT in which Alice sends a message that is not a port number, we can construct a modified protocol $P$ where Alice applies the procedure that Bob uses in $P$ to deduce the port to select and then sends the port number instead. For a given entanglement fidelity $F$, having fewer than $\left({d}{F}\right)^2$ ports would therefore violate the bound from \cref{cor:telebound}.
\end{proof}

The converse bound on the amount of quantum communication in \Cref{cor:sim-ident} holds for arbitrary protocols implementing a simulation of the identity channels, and \Cref{cor:telebound} puts a lower bound on the classical communication of any (approximate) teleportation scheme.
We continue to derive a converse bound specifically for port-based teleportation that is nontrivial for all combinations of $d$ and $N$.
Let us consider a general port-based teleportation scheme, given by POVMs $\{E_{A^N}^{(i)}\}$ and a resource state $\rho_{A^NB^N}$, where $A_0\cong\mathbb C^d$ and $B_1,\dots,B_N\cong\mathbb C^d$.
We would like to upper-bound the entanglement fidelity
\begin{align}
  F^*_d(N) = F\left(\sum_{i=1}^N (I_{B_0} \ot I_{B_i\to B_1}) \tr_{(B_0B_i)^c}[ (E_A^{(i)} \ot I_B) (\rho_{A^NB^N} \ot \phi^+_{A_0B_0}) ] (I_{B_0} \ot I_{B_i\to B_1}^\dagger), \phi^+_{B_0B_1}\right),\label{eq:entanglement-fidelity}
\end{align}
where $B_0\cong\mathbb C^d$ and $F(\rho,\sigma)=\lVert\sqrt\rho\sqrt\sigma\rVert_1^2$ is the fidelity.  This fidelity corresponds to the special case of Alice using an arbitrary PBT protocol to teleport half of a maximally-entangled state to Bob, who already possesses the other half. An upper bound for this fidelity then directly implies an upper bound for the entanglement fidelity of the PBT protocol.
We prove the following
\begin{thm}\label{thm:converse-bound}
  For any port-based teleportation scheme, the entanglement fidelity \eqref{eq:entanglement-fidelity} can be bounded from above as
  \begin{align}
  F^*_d(N) &\leq 1-\frac{d^2-1}{8N^2}\frac{1}{1+\frac{d^2-2}{2N}}.\label{eq:nonasymptotic-converse}
  \intertext{Asymptotically, this bound becomes}
  F^*_d(N) &\leq 1 - \frac{d^2-1}8 \frac 1{N^2} + O(N^{-3}).\label{eq:asymptotic-converse}
  \end{align}
\end{thm}

\begin{proof}
  Note first that for a pure state $|\psi\rangle$ we have $F(\psi,\tau) = \langle \psi|\tau|\psi\rangle$ for any mixed state $\tau$, and hence $\tau\mapsto F(\psi,\tau)$ is linear for any $\tau$.
  Since $\phi^+_{B_0B_1}$ is pure, the entanglement fidelity \eqref{eq:entanglement-fidelity} can hence be rewritten as
\begin{align}
  F^*_d(N)
  &= \sum_{i=1}^N p(i) F\left(\frac1{p(i)}\tr_{(B_0B_i)^c}[((E^{(i)})^{1/2}_A \ot I_B) (\rho_{A^NB^N} \ot \phi^+_{A_0B_0}) ((E^{(i)})^{1/2}_A \ot I_B)], \phi^+_{B_0B_i}\right) \\
  &= \sum_{i=1}^N p(i) F\left(\frac1{p(i)}((E^{(i)})^{1/2}_A \ot I_B) (\rho_{A^NB^N} \ot \phi^+_{A_0B_0}) ((E^{(i)})^{1/2}_A \ot I_B), \phi^+_{B_0B_i}\ot\sigma^{(i)}_{(B_0B_i)^c}\right)
\end{align}
for suitable $\sigma^{(i)}_{(B_0B_i)^c}$ whose existence is guaranteed by Uhlmann's Theorem.
Here we have introduced $p(i)=\tr[(E^{(i)})^{1/2}_A (\rho_{A^NB^N} \ot \tau_{A_0}) (E^{(i)})^{1/2}_A]$.
Abbreviating $\sqrt{F}(\cdot,\cdot) \equiv \sqrt{F(\cdot,\cdot)}$, we now have for any $j\in\{1,\ldots,N\}$ that
\begin{align}
  F^*_d(N)
  &\leq \sum_{i=1}^N p(i) \sqrt{F}\left(\frac1{p(i)}((E^{(i)})^{1/2}_A \ot I_B) (\rho_{A^NB^N} \ot \phi^+_{A_0B_0}) ((E^{(i)})^{1/2}_A \ot I_B), \phi^+_{B_0B_i}\ot\sigma^{(i)}_{(B_0B_i)^c}\right)\\
  &\leq \sqrt{F}\left(\rho_{B_j} \ot \tau_{B_0}, p(j) \phi^+_{B_0B_j} + (1-p(j)) \tau_{B_0} \ot \sigma_{B_j}\right)
\end{align}
where the second step uses joint concavity of the root fidelity, and we trace out all systems but $B_0B_j$, with $\sigma_{B_j}$ being some appropriate state.
Now, the fact that $\bra\phi^+_ {AB} \bigl(X_A \otimes \tau_B\bigr)\ket \phi^+_ {AB} = \frac1{d^2}\tr(X_A)$ for any operator $X_A$ and data processing inequality with respect to the binary measurement $\{\phi^+_{B_0B_j}, I-\phi^+_{B_0B_j}\}$ gives
\begin{align}
  F^*_d(N) &\leq \sqrt{f}\left(\frac1{d^2}, p(j) + (1-p(j)) \frac1{d^2}\right),
\end{align}
where $\sqrt{f}(x,y)=\sqrt{xy}+\sqrt{(1-x)(1-y)}$ is the binary root fidelity.
Note that $f(q,p + (1-p)q)$ is monotonically increasing as $p$ decreases from $1$ to $0$.
Now, one of the $N$ probabilities $p(j)$ has to be~$\geq1/N$.
Thus,
\begin{align}
  F^*_d(N) &\leq \sqrt{f}\left(\frac1{d^2}, \frac1N + \left(1-\frac1N \right) \frac1{d^2}\right) .
  \label{eq:classical-fidelity-bound}
\end{align}
To derive the non-asymptotic bound \eqref{eq:nonasymptotic-converse}, Equation \eqref{eq:classical-fidelity-bound} can be rearranged as
\begin{align}\label{eq:fbound}
  F^*_d(N) &\le\frac{1}{d^2}\left[\left(d^2-1\right) \left(1-\frac{1}{2
    N}\right)\sqrt{1-\frac{1}{(1-2 N)^2}} +\left(\frac{d^2-1}{2 N}+1\right)
    \sqrt{1-\frac{\left(d^2-1\right)^2}{\left(d^2+2 N-1\right)^2}}\right].
\end{align}
We bound the square roots using $\sqrt{1+a}\le 1+a/2$ for any $a\ge -1$ to obtain
\begin{align}
F^*_d(N) &\le\frac{1}{d^2}\left[\left(d^2-1\right)  \left(1-\frac{1}{2
    N}\right)\left(1-\frac{1}{2 (1-2 N)^2}\right)+\left(\frac{d^2-1}{2 N}+1\right) \left(1-\frac{\left(d^2-1\right)^2}{2
    \left(d^2+2 N-1\right)^2}\right)\right]\\
  &=1-\frac{d^2-1}{8N^2}\frac{1}{(1-\frac{1}{2N}) \left(1+\frac{d^2-1}{2N}\right)}\\
  &\le 1-\frac{d^2-1}{8N^2}\frac{1}{1+\frac{d^2-2}{2N}},
\end{align}
which is \eqref{eq:nonasymptotic-converse}.
For $N\to\infty$ this implies
\begin{align}
F^*_d(N) &\leq  1 - \frac{d^2-1}8 \frac 1{N^2} + O(N^{-3}),
\end{align}
which is \eqref{eq:asymptotic-converse} and concludes the proof.
\end{proof}

Combining \Cref{thm:converse-bound} with \cref{cor:porttelebound} above yields a simplified bound as a corollary, that we stated as \Cref{cor:converse} in \Cref{sec:summary} as one of our main results.
We restate it below for convenience, and in \cref{fig:converse-comparison} we compare the quality of this bound for $N>d^2/2$ with the converse bound \eqref{eq:lowebound-ishi} derived in~\cite{Ishizaka2015}.

\begin{cor-converse}[restated]
	\restateConverse
\end{cor-converse}

\begin{figure}
  \centering
  \begin{tikzpicture}
  \begin{axis}[
  scale = 1.5,
  xlabel=$N$,
  xmode=log,
  ymin = 0.99,
  ymax = 1.0005,
  ytick={0.99,0.992,0.994,0.996,0.998,1},
  legend style = {font=\footnotesize,at = {(0.95,0.05)},anchor = south east},
  legend cell align = left,
  tick label style={/pgf/number format/fixed, /pgf/number format/precision=3},
  ]
  \addlegendimage{semithick,dashed,color=black}
  \addlegendimage{semithick,color=black}
 
  \addplot[smooth,mark=square*,mark size=1pt,semithick,dashed,domain=3:100,color=blue] {1-1/(4*x^2)};
  \addplot[smooth,mark=square*,mark size=1pt,semithick,dashed,domain=5:100,color=red] {1-1/(4*2*x^2)};
  \addplot[smooth,mark=square*,mark size=1pt,semithick,dashed,domain=9:100,color=orange] {1-1/(4*3*x^2)};
  \addplot[smooth,mark=square*,mark size=1pt,semithick,dashed,domain=13:100,color=green] {1-1/(4*4*x^2)};

  \addplot[smooth,mark=*,mark size=1pt,semithick,domain=3:100,color=blue] {1-3/(16*x^2)};
  \addplot[smooth,mark=*,mark size=1pt,semithick,domain=5:100,color=red] {1-8/(16*x^2)};
  \addplot[smooth,mark=*,mark size=1pt,semithick,domain=9:100,color=orange] {1-15/(16*x^2)};
  \addplot[smooth,mark=*,mark size=1pt,semithick,domain=13:100,color=green] {1-24/(16*x^2)};

  \legend{$1-\frac{1}{4(d-1)N^2}$~\cite{Ishizaka2015},$1-\frac{d^2-1}{16N^2}$ (\cref{cor:porttelebound}),,,,,$d=2$,$d=3$,$d=4$,$d=5$};
  \end{axis}
  \end{tikzpicture}
  \caption{Comparison of the converse bound $F_d^*(N) \leq 1-\frac{1}{4(d-1)N^2}$ derived in~\cite{Ishizaka2015} and the converse bound $F_d^*(N)\leq 1-\frac{d^2-1}{16N^2}$ derived in \cref{cor:porttelebound}, valid for $N>d^2/2$.}
  \label{fig:converse-comparison}
\end{figure}
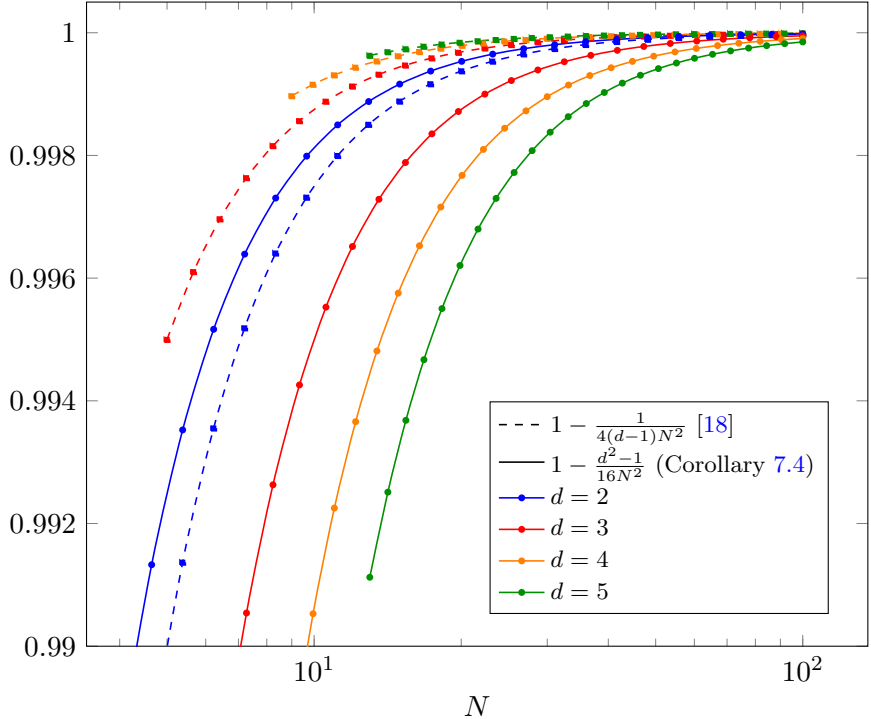

\section{Conclusion}\label{sec:conclusion}
In this paper, we completed the picture of the asymptotic performance of port-based teleportation~(PBT) in the important regime when the input dimension is fixed while the number of ports tends to infinity.
In particular, we determined the asymptotic performance of deterministic PBT in the fully optimized setting, showing that the optimal infidelity decays as $\Theta(1/N^2)$ with the number of ports~$N$.
We also determined the precise asymptotics of the standard protocol for deterministic PBT (which uses EPR pairs and the `pretty good' measurement) as well as probabilistic PBT using EPR pairs.
The asymptotics for probabilistic PBT in the fully optimized setting had been determined previously in~\cite{Mozrzymas2017}.

While our work closes a chapter in the study of PBT, it opens several interesting avenues for further investigation, both in the finite and in the asymptotic regime.
Note that the limit $d\to\infty$ for fixed $N$ is not very interesting, as the error tends to one in this regime.
However, it would be natural to consider limits where both $N$ and $d$ tend to infinity.
In particular, the fidelity~$F_d^*(N)$ plausibly has a nontrivial limit when the ratio $N/d^2$ remains fixed.
Given the import of PBT to, e.g., instantaneous non-local quantum computation, it would be desirable to determine the limiting value.
Finally, we also mention the problem of determining the exact functional dependence on~$d$ of the leading order coefficient $\lim_{N\to\infty}N^2(1-F_d^*(N))$ in fully optimized deterministic PBT.

\paragraph{Acknowledgements.}
We acknowledge interesting discussions with Charles Bordenave, Beno\^{i}t Collins, Marek Mozrzymas, M\=aris Ozols, Jan Philip Solovej, Sergii Strelchuk, and Micha\l{} Studzi\'{n}ski.
MC and FS acknowledge financial support from the European Research Council (ERC Grant Agreement No.~337603) and VILLUM FONDEN via the QMATH Centre of Excellence (Grant No.~10059).
MC further acknowledges the Quant-ERA project Quantalgo and the hospitality of the Center for Theoretical Physics at MIT, where part of this work was done.
FL and GS are supported by National Science Foundation (NSF) Grant No.~PHY 1734006.
FL appreciates the hospitality of QuSoft, CWI, and the University of Amsterdam, where part of this work was done.
CM was supported by a Netherlands Organisation for Scientific Research (NWO) VIDI grant (639.022.519).
GS is supported by the NSF Grant No.~CCF 1652560.
MW thanks JILA for hospitality, where this work was partly initiated.
MW acknowledges financial support by the NWO through Veni Grant No.~680-47-459.

\appendix

\section{Proof of Lemma~\ref{lem:singlet bound}}\label{app:single bound}

The following lemma was first derived in~\cite{studzinski2016port}.
In this section we give an alternative proof.
Our proof is elementary and only uses the Schur-Weyl duality and the Pieri rule.

\begin{lem-singlet-bound}[restated]
\restateSingletBound
\end{lem-singlet-bound}

\begin{proof}
We note that the operator $T(N)$ commutes with the action of $U(d)$ by $\bar{U} \ot U^{\ot N}$ as well as with the action of $S_N$ that permutes the systems $B_1,\dots,B_N$.
Let us work out the corresponding decomposition of $(\C^d)^{1+N}$:
We first consider the action of $U(d) \times U(d)$ by $\bar{U} \ot V^{\ot N}$ together with the $S_N$. By Schur-Weyl duality,
\begin{align}
  (\C^d)^{1+N} \cong (\C^d)^* \ot \bigoplus_{\mu\vdash_d N} V_\mu^d \ot W_\mu.
\end{align}
The notation means that $\mu$ runs over all Young diagrams with $N$ boxes and no more than $d$ rows (i.e., $\mu_1\geq\dots\geq\mu_d\geq0$ and $\sum_j \mu_j=N$). We write $V_\mu^d$ for the irreducible $U(d)$-representation with highest weight $\mu$, and $W_\mu$ for the irreducible $S_N$-representation corresponding to the partition $\mu$.

The dual representation $(\C^d)^*$ is not polynomial; its highest weight is $(-1,0,\dots,0)$.
However, $(\C^d)^* \cong V_{(1,\dots,1,0)}^d \ot {\det}^{-1}$.
The (dual) Pieri rule tells us that $V_{(1,\dots,1,0)}^d \ot V_\mu^d$ contains all irreducible representations whose highest weight can be obtained by adding $1$'s to all but one of the rows (with multiplicity one). Tensoring with the determinant amounts to subtracting $(-1,\dots,-1)$, so the result of tensoring with $(\C^d)^*$ amounts to subtracting $1$ from one of the rows:
\begin{align}
(\C^d)^* \ot V_\mu^d = \bigoplus_{i : \mu_i > \mu_{i+1}} V_{\mu - \epsilon_i}^d,
\end{align}
where we write $\epsilon_i$ for the $i$-th standard basis vector.
(We stress that $\mu-\epsilon_i$ is always a highest weight, but does not need to be a Young diagram.)
Thus, we obtain the following multiplicity-free decomposition into $U(d) \times S_N$-representations:
\begin{align}
  (\C^d)^{1+N} \cong \bigoplus_{\mu\vdash_d N} \bigoplus_{i: \mu_i > \mu_{i+1}} V_{\mu - \epsilon_i}^d \ot W_\mu.
\end{align}
The operator $T(N)$ can be decomposed accordingly:
\begin{align}
  T(N) = \bigoplus_{\mu,i} t_{\mu,i} \cdot I_{V_{\mu - \epsilon_i}^d} \ot I_{W_\mu}
\end{align}
for some $t_{\mu,i}\geq0$.
To determine the $t_{\mu,i}$, let us denote by $P_\mu$ the isotypical projectors for the $S_N$-action on $(\C^d)^{\ot N}$ and by $Q_\alpha$ the isotypical projectors for the $U(d)$ action by ${\bar U} \ot U^{\ot N}$ (they commute).
Then:
\begin{align}\label{eq:otoneh}
\tr T(N) (I_A \ot P_\mu) Q_{\mu - \epsilon_i} = t_{\mu,i} \, \dim(V_{\mu-\epsilon_i}^d) \, \dim(W_\mu) \,.
\end{align}
On the other hand:
\begin{align}
\tr T(N) (I_A \ot P_\mu) Q_{\mu - \epsilon_i} = \tr \phi^+_{AB_1} (I_A \ot P_\mu) Q_{\mu - \epsilon_i}.
\end{align}
The maximally entangled state $\phi^+_{AB_1}$ is invariant under ${\bar U} \ot U$.
This means that on the range of the projector $\phi^+_{AB_1}$, the actions of ${\bar U} \ot U^{\ot N}$ and $I_{AB_1} \ot U^{\ot (N-1)}$ agree!
Explicitly:
\begin{align}
\C\ket{\phi^+}_{AB_1} \ot (\C^d)^{\ot (N-1)} \cong \C\ket{\phi^+}_{AB_1} \ot \bigoplus_{\alpha \vdash_d N-1} V_{\alpha}^d \ot [\alpha] \cong \bigoplus_{\alpha \vdash_d N-1} V_{\alpha}^d \ot [\alpha].
\end{align}
It follows that
\begin{align}
\phi^+_{AB_1} Q_\alpha = \phi^+_{AB_1} (I_{AB_1} \ot Q'_\alpha)
\end{align}
where $Q'_\alpha$ refers to the action of $U(d)$ by $U^{\ot (N-1)}$ on $B_2\dots{}B_n$, and so
\begin{align}
\tr \phi^+_{AB_1} (I_A \ot P_\mu) Q_{\mu - \epsilon_i} = \begin{cases}
\tr \phi^+_{AB_1} (I_A \ot P_\mu) (I_{AB_1} \ot Q'_\alpha) & \text{if } \alpha \coloneqq \mu-\epsilon_i \text{ is a partition}, \\
0 & \text{otherwise}.
\end{cases}
\end{align}
We can now trace over the A-system:
\begin{align}
  \tr \phi^+_{AB_1} (I_A \ot P_\mu) (I_{AB_1} \ot Q'_\alpha)
  = \frac 1d \tr P_\mu (I_{B_1} \ot Q'_\alpha) \,.
\end{align}
The remaining trace is on $(\C^d)^{\ot N}$.
The operator $P_\mu$ refers to the $S_N$-action, while $Q'_\alpha$ refers to the $U(d)$-action by $U^{\ot(N-1)}$ on $B_2\dots{}B_n$.
Equivalently, we can define $Q'_\alpha$ with respect to the $S_{N-1}$ action by permuting the last $N-1$ tensor factors.
Using Schur-Weyl duality and the branching rule for restricting $S_N$ to $S_1 \times S_{N-1}$:
\begin{align}
(\C^d)^{\ot N}
= \bigoplus_\mu V_\mu^d \ot W_\mu
= \bigoplus_\mu V_\mu^d \ot \bigoplus_{i : \alpha = \mu - \epsilon_i \text{ partition}} W_\alpha
\end{align}
And hence
\begin{align}
\frac 1d \tr P_\mu (I_{B_1} \ot Q'_\alpha) = \frac 1d \dim(V^d_\mu) \dim(W_\alpha)
\end{align}
in the case of interest.
Comparing this with \cref{eq:otoneh}, we obtain the following result:
\begin{align}
  t_{\mu,i} = \frac 1d \frac {\dim(V^d_\mu)} {\dim(V^d_\alpha)} \frac {\dim(W_\alpha)} {\dim(W_\mu)}
\end{align}
if $\alpha = \mu - \epsilon_i$ is a partition, and otherwise zero.
These are the desired eigenvalues of $T(N)$.
\end{proof}

\section{A family of explicit protocols for deterministic PBT}

Guessing a good candidate density $c_\mu$ with a simple functional form for the optimization in \cref{eq:cambridgeII} yields a protocol with performance close to the achievability bound \cref{thm:thefrenchguys-summary}.
\begin{thm}\label{thm:PrettyGoodProtocol}
  For fixed but arbitrary dimension $d$, there exists a concrete protocol for deterministic PBT with entanglement fidelity
  \begin{align}
  F\ge 1-\frac{d^4(d+3)}{2N^{2}}-O(N^{-3})
  \end{align}
\end{thm}
\begin{proof}
  Assume that $N/d^2$ is an integer (otherwise use only the first $d^2\left\lfloor\frac{N}{d^2}\right\rfloor$ ports). Let $c_\mu$ be defined such that
  \begin{align}
  q(\mu)=c_\mu p(\mu)=\begin{cases}
  \eta_{N}\left(R^2-r(\mu)^2\right)^2 & r\le R\\
  0&\text{else},
  \end{cases}
  \end{align}
  with
  \begin{align}
  r(\mu)&=\|\mu-\hat\mu\|_2,\\
  \hat\mu&=\left((2d-1)\frac{N}{d^2},(2d-3)\frac{N}{d^2},\ldots,\frac{N}{d^2}\right),\\
  R&=\sqrt{2}\frac{N}{d^2},
  \end{align}
  and
  \begin{align}
  \eta_{N}=\left(\sum_{\substack{\mu\in\hat\mu+\Lambda_d\\ r(\mu)\le R}}\left(R^2-r(\mu)^2\right)^2\right)^{-1}
  \end{align}
  is a normalization factor that ensures that $q$ is a probability distribution. $\hat\mu$ has Euclidean distance $R$ from the boundary of the set of Young diagrams, i.e.\ all vectors $\mu\in\hat\mu+\Lambda_d$ such that $\|\mu-\hat\mu\|_2\le R$  are Young diagrams. We extend the probability distribution $q$ to be defined on all $v\in\hat\mu+\Lambda_d$ for convenience. Let $B^{\Lambda_d}_L(v_0)=\{v\in v_0+\Lambda_d|\|v-v_0\|_2\le L\}$. We now look at the PBT-fidelity for the protocol using the density $c_\mu$. First note that the formula \cref{eq:cambridgeII} can be rearranged in the following way,
  \begin{align}
  d^2F&=\sum_{\alpha \vdash_d N-1}\left(\sum_{\mu=\alpha+\square}\sqrt{q(\mu)}\right)^2\\
  &=\sum_{\alpha \vdash_d N-1}\sum_{\mu,\mu'=\alpha+\square}\sqrt{q(\mu)q(\mu')}\\
  &=\sum_{\mu\vdash_d N}\sum_{\mu'=\mu+\square-\square}\sqrt{q(\mu)q(\mu')}.
  \end{align}
  In the last line, the notation $\mu'=\mu+\square-\square$ means summing over all possibilities to remove a square from $\mu$ and adding one, including removing and adding the same square. Noting that all vectors in $B^{\Lambda_d}_R(\hat\mu)$ are Young diagrams, we can write
  \begin{align}
  d^2F&=\sum_{\mu\in B^{\Lambda_d}_R(\hat\mu)}\sum_{i,j=1}^d\mathds 1_{B_R(\hat\mu)}(\mu+e_i-e_j)\sqrt{q(\mu)q(\mu+e_i-e_j)}\\
  &=\sum_{\mu\in B^{\Lambda_d}_R(\hat\mu)}q(\mu)\sum_{i,j=1}^d\mathds 1_{B_R(\hat\mu)}(\mu+e_i-e_j)\sqrt{\frac{q(\mu+e_i-e_j)}{q(\mu)}}\\
  &=\sum_{\mu\in B^{\Lambda_d}_R(\hat\mu)}q(\mu)\sum_{i,j=1}^d\mathds 1_{B_R(\hat\mu)}(\mu+e_i-e_j)\left(1+2\frac{g_{ij}(\mu)-1}{\sqrt{f(\mu)}}\right)\\
  &=\sum_{\mu\in B^{\Lambda_d}_{R-\sqrt 2}(\hat\mu)}q(\mu)\sum_{i,j=1}^d\left(1+2\frac{g_{ij}(\mu)-1}{\sqrt{f(\mu)}}\right)\\
  &+\sum_{\mu\in B^{\Lambda_d}_{R}(\hat\mu)\setminus B^{\Lambda_d}_{R-\sqrt{2}}(\hat\mu)}q(\mu)\sum_{i,j=1}^d\mathds 1_{B_R(\hat\mu)}(\mu+e_i-e_j)\left(1+2\frac{g_{ij}(\mu)-1}{\sqrt{f(\mu)}}\right).\label{eq:shellsplitoff}
  \end{align}
  Here we have defined the functions
  \begin{align}
  g_{ij}(\mu)=\mu_j-\hat\mu_j-\mu_i+\hat\mu_i
  \end{align}
  and
  \begin{align}
  f(\mu)=\left(R^2-r(\mu)^2\right)^2.
  \end{align}
  The last equation holds because $\|e_i-e_j\|_2=(1-\delta_{ij})\sqrt 2$, i.e.\ for all $\mu\in B^{\Lambda_d}_{R-\sqrt{2}}(\hat{\mu})$ and all $1\le i,j\le d$ we have $\mu+e_i-e_j\in B^{\Lambda_d}_{R}(\hat{\mu})$.

  We can bound the normalization constant as follows. Denote by $\cP(\Lambda_d)$ the unit cell of $\Lambda_d$ with smallest diameter, $\ell$
  .
  The volume of the unit cell is $\sqrt d$, which can be seen as follows.\footnote{

    For a general lattice $\mathcal{L}_d\subset\mathbb{R}^d$ with basis $B=\lbrace b_1,\dots,b_m\rbrace$ (where $m\leq d$), the volume of the unit cell of $\mathcal{L}^d$ is equal to $\det(\mathcal{L}_d) = \sqrt{\det(B^TB)}$.
  }
  A basis for the lattice $\Lambda_d = \lbrace v\in\mathbb{Z}^d\colon \sum_{i=1}^d v_i = 0\rbrace$ is given by $B=\lbrace b_i\rbrace_{i=1}^{d-1}$, where $b_i=e_1-e_{i+1}$.
  It follows that $M=B^T B$ is a $(d-1)\times (d-1)$-matrix with all diagonal elements equal to $2$ and all off-diagonal elements equal to $1$.
  The matrix $M$ has one eigenvalue $d$ corresponding to the eigenvector $\sum_{i=1}^{d-1} e_i$, and $d-2$ eigenvalues $1$ corresponding to the eigenvectors $e_1-e_{i+1}$, respectively.
  Hence, $\det(\Lambda_d) = \sqrt{d}$.

  Let further $g:\R\Lambda_d\to\Lambda_d$ be the function such that for all $x\in\R\Lambda^d$ there exist $\gamma_i\in(-1/2,1/2]$, $i=1,\ldots,d-1$ such that
  \begin{align}
  x=g_N(x)+\sum_{i=1}^{d-1}\gamma_i a_i.
  \end{align}
  Heuristically, $g$ is the function that maps every point in the $(d-1)$-dimensional subspace $\Lambda^d$ lives in to the lattice point $v$ in whose surrounding unit cell  it lies, where the surrounding unit cell is here the set $\{v+\sum_{i=1}^{d-1}\gamma_i a_i|\gamma_i\in(-1/2,1/2]\}$, i.e.\ the point lies in the center of the cell. As $f$ is nonnegative, we have with $l$ as defined above that
  \begin{align}
  \eta_{N}^{-1}&=\sum_{\mu\in B^{\Lambda_d}_R(\hat \mu)}f(\mu)\\
  &\le\frac{1}{\sqrt d} \int_{ B^{\Lambda_d}_{R+\ell/2}(\hat \mu)}f(g_N(x))\mathrm dx\\
  &\le\frac{1}{\sqrt d} \int_{ B^{\Lambda_d}_{R+\ell/2}(\hat \mu)}f(x)\mathrm dx\\
  &+\frac{1}{\sqrt d} \int_{ B^{\Lambda_d}_{R+\ell/2}(\hat \mu)}\frac{\ell}{2}\max_{x':\|x-x'\|_2\le\ell/2}\left\|(\nabla f)(x')\right\|_2\mathrm dx\label{eq:intappr}
  \end{align}
  The gradient of $f$ is given by
  \begin{align}
  (\nabla f)(x)=-4(R^2-\|x\|_2^2)x.
  \end{align}
  We can bound
  \begin{align}
  |4(R^2-(r\pm l)^2)(r\pm l)|\le 4(R^2-(r-l)^2)(r+l),
  \end{align}
  so
  \begin{align}
  &\frac{1}{\sqrt d} \int_{ B^{\Lambda_d}_{R+\ell/2}(\hat \mu)}\frac{\ell}{2}\max_{x':\|x-x'\|_2\le\ell/2}\left\|(\nabla f)(x')\right\|_2\mathrm dx\\
  &\le \frac{1}{\sqrt d} \int_{ B^{\Lambda_d}_{R+\ell/2}(\hat \mu)}\frac{\ell}{2}4(R^2-(r(x)-\ell/2)^2)(r(x)+\ell/2)\mathrm dx\\
  &= \frac{2\ell\vol(\mathbb S_{d-2})}{\sqrt d} \int_{ 0}^{R+\ell/2}r^{d-2}(R^2-(r-\ell/2)^2)(r+\ell/2)\mathrm dr\\
  &= \frac{2\ell\vol(\mathbb S_{d-2})}{\sqrt{d}}\left(\frac{1}{d}-\frac{1}{d+2}\right)\left(R+\ell/2\right)^{d+2}+O(R^{d+1}).
  \end{align}
  Here, we changed into spherical coordinates with origin in $\hat\mu$ in the third line, and $\vol(\mathbb S_{d-2})$ is the  volume of the $(d-2)$-dimensional sphere.
  Turning to the first term in \cref{eq:intappr}, we calculate
  \begin{align}
  \int_{ B^{\Lambda_d}_{R+\ell/2}(\hat \mu)}f(x)\mathrm dx&=\int_{ B^{\Lambda_d}_{R+\ell/2}(\hat \mu)}(R^2-r(x)^2)^2\mathrm dx\\
  &=\vol(\mathbb S_{d-2})\int_{ 0}^{R+\ell/2}r^{d-2}(R^2-r^2)^2\mathrm dr\\
  &=\vol(\mathbb S_{d-2})\left(R+\ell/2\right)^{d+3}\left(\frac{1}{d-1}-\frac{2}{d+1}+\frac{1}{d+3}\right)\\
  &=\frac{8\vol(\mathbb S_{d-2})\left(R+\ell/2\right)^{d+3}}{d^3+3d^2-d-3}
  \end{align}
  Combining the last two equations, expanding the polynomials of the form $\left(R+\ell/2\right)^{k}$ and using the power series expansion of $1/(1+x)$ we finally arrive at
  \begin{align}
  \eta_{N}&=\frac{\sqrt{d}\left(d^3+3d^2-d-3\right)}{8\vol(\mathbb S_{d-2})}R^{-(d+3)}+O(R^{-(d+4)}).
  \end{align}
  Returning to equation \cref{eq:shellsplitoff}, let us first bound the magnitude of the last term. To this end, observe that for $r(\mu)\ge R-\sqrt{2}$, we have
  \begin{align}
  \sqrt{f(\mu)}&=(R^2-r(\mu)^2)\\
  &\le 2\sqrt 2 R.
  \end{align}
  Furthermore we have that
  \begin{align}
  \mathds 1_{B_R(\hat\mu)}(\mu+e_i-e_j)\left(1+2\frac{g_{ij}(\mu)-1}{\sqrt{f(\mu)}}\right)&\le 1+2\frac{2R}{\sqrt{f(\mu)}},
  \end{align}
  and hence
  \begin{align}
  &\sum_{\mu\in B^{\Lambda_d}_{R}(\hat\mu)\setminus B^{\Lambda_d}_{R-\sqrt{2}}(\hat\mu)}q(\mu)\sum_{i,j=1}^d\mathds 1_{B_R(\hat\mu)}(\mu+e_i-e_j)\left(1+2\frac{g_{ij}(\mu)-1}{\sqrt{f(\mu)}}\right)\\
  &\le d^2\eta_{N}\sum_{\mu\in B^{\Lambda_d}_{R}(\hat\mu)\setminus B^{\Lambda_d}_{R-\sqrt{2}}(\hat\mu)}\left(f(\mu)+2\sqrt{f(\mu)}R\right)\\
  &\le d^2\eta_{N}\sum_{\mu\in B^{\Lambda_d}_{R}(\hat\mu)\setminus B^{\Lambda_d}_{R-\sqrt{2}}(\hat\mu)}\left(\left(2\sqrt 2R\right)^2+4\sqrt 2 R^2\right)\\
  &\le 4(2+\sqrt 2)d^2R^2\eta_{N}\Bigl|B^{\Lambda_d}_{R}(\hat\mu)\setminus B^{\Lambda_d}_{R-\sqrt{2}}(\hat\mu)\Bigr|
  \end{align}
  To bound the number of lattice points in the spherical shell $B^{\Lambda_d}_{R}(\hat\mu)\setminus B^{\Lambda_d}_{R-\sqrt{2}}(\hat\mu)$, note that i) each lattice point is surrounded by its own unit cell, and ii) these cells have diameter $\ell$. Therefore all these unit cells are disjoint subsets of a shell of width $\sqrt{2}+\ell$, and hence we have the bound
  \begin{align}
  \Bigl|B^{\Lambda_d}_{R}(\hat\mu)\setminus B^{\Lambda_d}_{R-\sqrt{2}}(\hat\mu)\Bigr|\le \vol{\mathbb S_{d-2}}(R+\ell)^{d-2}(\ell+\sqrt 2).
  \end{align}
  Combining the bounds we arrive at
  \begin{align}
  &\sum_{\mu\in B^{\Lambda_d}_{R}(\hat\mu)\setminus B^{\Lambda_d}_{R-\sqrt{2}}(\hat\mu)}q(\mu)\sum_{i,j=1}^d\mathds 1_{B_R(\hat\mu)}(\mu+e_i-e_j)\left(1+2\frac{g_{ij}(\mu)-1}{\sqrt{f(\mu)}}\right)=O(R^{-3})
  \end{align}
  Turning to the first expression on the right hand side of \cref{eq:shellsplitoff}, we observe that both the set $B^{\Lambda_d}_{R-\sqrt 2}$ and the distribution $q$ are invariant under the map $\mu\mapsto 2\mu-\hat\mu$, i.e.\ central reflection about $\hat\mu$. Therefore the sum over $g_{ij}(\mu)$, which is linear in $\mu-\hat\mu$, vanishes, i.e.
  \begin{align}
  &\sum_{\mu\in B^{\Lambda_d}_{R-\sqrt 2}(\hat\mu)}q(\mu)\sum_{i,j=1}^d\left(1+2\frac{g_{ij}(\mu)-1}{\sqrt{f(\mu)}}\right)\\
  &=\sum_{\mu\in B^{\Lambda_d}_{R-\sqrt 2}(\hat\mu)}q(\mu)\sum_{i,j=1}^d\left(1-2\frac{1}{\sqrt{f(\mu)}}\right)\\
  &=\sum_{\mu\in B^{\Lambda_d}_{R}(\hat\mu)}q(\mu)\sum_{i,j=1}^d\left(1-2\frac{1}{\sqrt{f(\mu)}}\right)-\sum_{\mu\in B^{\Lambda_d}_{R}(\hat\mu)\setminus B^{\Lambda_d}_{R-\sqrt{2}}(\hat\mu)}q(\mu)\sum_{i,j=1}^d\left(1-2\frac{1}{\sqrt{f(\mu)}}\right)\\
  &\ge d^2-\eta_{N}\left(2\sum_{\mu\in B^{\Lambda_d}_{R}(\hat\mu)}\sqrt{f(\mu)}+\sum_{\mu\in B^{\Lambda_d}_{R}(\hat\mu)\setminus B^{\Lambda_d}_{R-\sqrt{2}}(\hat\mu)}\sqrt{f(\mu)}\left(\sqrt{f(\mu)}-2\right)\right)
  \end{align}
  Using the same argument as for bounding $\eta_{N}$, we find
  \begin{align}
  \sum_{\mu\in B^{\Lambda_d}_{R}(\hat\mu)}\sqrt{f(\mu)}&\le\frac{2\ell\vol(\mathbb S_{d-2})}{\sqrt d} \int_{ 0}^{R+\ell/2}r^{d-2}(R^2-(r-\ell/2)^2)\mathrm dr\\
  &=\frac{\vol(\mathbb S_{d-2})R^{d+1}}{(d^2-1)\sqrt d}+O(R^d) .
  \end{align}
  The second term is bounded in the same way as the spherical shell sum above, yielding
  \begin{align}
  \eta_{N}d^2\sum_{\mu\in B^{\Lambda_d}_{R}(\hat\mu)\setminus B^{\Lambda_d}_{R-\sqrt{2}}(\hat\mu)}\sqrt{f(\mu)}\left(\sqrt{f(\mu)}-2\right)=O(R^{-3}).
  \end{align}
  Combining all bounds, we arrive at
  \begin{align}
  F&\ge 1-\frac{d^3+3d^2-d-3}{d^2-1}R^{-2}+O(R^{-3})=(d+3)R^{-2}+O(R^{-3}).
  \end{align}
  Using $R=\frac{N}{d^2}$ we obtain the final bound
  \begin{align}
  F&\ge 1- \frac{d^4(d+3)}{2N^{2}}+O(R^{-3}).
  \end{align}
\end{proof}

\section{\texorpdfstring{The maximal eigenvalue of a $2\times 2$ GUE${}_0$ matrix}{The maximal eigenvalue of a 2x2 GUE\_0 matrix}}\label{app:maxboltz}
The maximal eigenvalue $\lambda_{\max}(\mathbf G)$ of a $2\times 2$ GUE${}_0$ matrix $\mathbf G$ can be easily analyzed, as $\lambda_{\max}(\mathbf G)=\sqrt{\frac 1 2 \tr \mathbf G^2}$.
\begin{lem}\label{lem:maxboltz}
  For $\mathbf{X}\sim \mathrm{GUE}_0(2)$, $\sqrt 2\lambda_{\max}(\mathbf G)\sim\chi_3$, where $\chi_3$ is the chi-distribution with three degrees of freedom.\footnote{This distribution is also known as the Maxwell-Boltzmann distribution.} Consequently, $\mathbb{E}\left[\lambda_{\max}(\mathbf G)\right]=\frac{2}{\sqrt \pi}$.
\end{lem}
\begin{proof}
  By definition, the probability density of $\mathrm{GUE}(d)$ is
  \begin{align}
    p_{\mathrm{GUE}}(M)=\left(2\pi\right)^{-\frac{d^2}{2}}\exp\left(-\frac{\tr M^2}{2}\right),
  \end{align}
  and therefore we get
  \begin{align}
  p_{\mathrm{GUE}_0}(G)=\left(2\pi\right)^{-\frac{d^2-1}{2}}\exp\left(-\frac{\tr G^2}{2}\right),
  \end{align}
  for the density of GUE${}_0$. Writing $\mathbf G=\sum_{i=1}^3\mathbf x_i \sigma_i$ with the Pauli matrices $\sigma_i, i=1,2,3$, we see that the $\mathbf x_i$ are independent normal random variables with variance $1/2$, and
  \begin{align}
    \lambda_{\max}(\mathbf G)=\sqrt{\frac{\tr \mathbf G^2}{2}}=\sqrt{\sum_{i=1}^3\mathbf x_i^2 },
  \end{align}
  proving the claim.
\end{proof}

\section{Technical lemmas}\label{app:technical}
The following ``mirror lemma'', also called ``transpose trick'', is well known in the literature, and can be proven in a straightforward way:

\begin{lem}[Mirror lemma, transpose trick]\label{lem:mirror-lemma}
  Let $\lbrace |i\rangle\rbrace_{i=1}^d$ be a basis and $|\gamma\rangle=\sum_{i=1}^d |i\rangle|i\rangle$ be the unnormalized maximally entangled state.
  For any operator $X$,
  \begin{align}
  I\otimes X|\gamma\rangle = X^T\otimes I |\gamma\rangle,
  \end{align}
  where $X^T$ denotes transposition of $X$ with respect to the basis $\lbrace |i\rangle\rbrace_{i=1}^d$.
\end{lem}

The maximization in the definition of the diamond norm can be carried out explicitly for the distance of two unitarily covariant channels.
This is the statement of the following lemma, which is a special case of a more general result about generalized divergences proven in~\cite{LKDW17}.
\begin{lem}[\cite{LKDW17}]\label{lem:covdiamond}
  Let $\Lambda^{(i)}_{A\to A}$ for $i=1,2$ be unitarily covariant maps. Then the maximally entangled state $|\phi^+\rangle_{AA'}$ is a maximizer for their diamond norm distance, i.e.,
  \begin{align}
  \left\|\Lambda^{(1)}_{A\to A}-\Lambda^{(2)}_{A\to A}\right\|_\diamond=\left\|\left(\Lambda^{(1)}_{A\to A}-\Lambda^{(2)}_{A\to A}\right)(\phi^+_{AA'})\right\|_1.
  \end{align}
\end{lem}

The following Lemma from Ref.~\cite{Pirandola2018} shows that the entanglement fidelity and the diamond norm distance to the identity channel are  even in a 1-1 relation for unitarily covariant channels.

\begin{lem}[\cite{Pirandola2018}]\label{lem:pirandola}
  For a unitarily covariant channel $\Lambda\colon A\to A$,
  \begin{align}
  \|\idch_A-\Lambda\|_\diamond=2\left(1-\sqrt F(\Lambda)\right).
  \end{align}
\end{lem}

We need an explicit limit of certain Riemann sums. The proof of the following can, e.g., be found in~\cite{stack-riemann}.
\begin{lem}\label{lem:Riemann-madness}
  Let $f:\R_+\to\R_+$ be nonincreasing such that the (proper or improper) Riemann integral
  \begin{align}
  \intop_a^b f(x)\mathrm dx
  \end{align}
  exists for all $a,b\in[0,\infty]$ with $a<b$. Then
  \begin{align}
  \lim_{n\to\infty}\frac{1}{n}\sum_{i=1}^{g n}f\left(\frac{c+i}{n}\right)=\intop_0^gf(x)\mathrm dx
  \end{align}
  for all $c\ge 0$ and $g\in[0,\infty]$.
\end{lem}

The following lemma provides the volume of the simplex of ordered probability distributions as well as the volume of its boundary.
\begin{lem}\label{lem:simpvol}
  Let
  \begin{align}
    \OS_{d-1}=\left\{x\in\R^{d}\Bigg|\sum_i x_i=0, x_i\ge x_{i+1}, x_{d}\ge 0\right\}
  \end{align}
  be the simplex of ordered probability distributions. The volume of this simplex, and the volume of its boundary, are given by
  \begin{align}
    \vol(\OS_{d-1})&=\frac{1}{\sqrt d ((d-1)!)^2}\text{, and}\label{eq:simpvol}\\
    \vol(\partial\OS_{d-1})&=\vol(\OS_{d-1})\left(\frac{d(d-1)^2}{\sqrt 2}+\sqrt d(d-1)^{3/2}+\sqrt 2(d-1)\right),
  \end{align}
  respectively.
\end{lem}
\begin{proof}
  $\OS_{d-1}$ is given in its dual description above, let us therefore begin by finding its extremal points. These are clearly given by
  \begin{align}
    v_i=\left(\frac 1 i,\ldots,\frac 1 i,0,\ldots,0\right),
  \end{align}
  i.e.\ the $i$th extremal point has $i$ entries $\frac 1 i$ and $d-i$ entries $0$. The supporting (affine) hyperplanes $H_i$ of the facets $F_i$, $i=1,\ldots,d$ of $\OS_{d-1}$ in $V_0^{(d-1)}=\left\{x\in \R^d|\sum_i x_i=0\right\}$ are given by the normalized normal vectors
  \begin{align}
    n_i&=\frac{e_i-e_{i+1}}{\sqrt 2},\ i=1,\ldots,d-1,\text{ and}\\
    n_d&=\frac{1}{\sqrt{d(d-1)}}(1,\ldots,1,-d+1).
  \end{align}
  Now note that the facet $F_d=\{x\in\OS_{d-1}|x_d=0\}$ is equal to $\OS_{d-2}$, and the volume of a $(d-1)$-dimensional pyramid is given by the product of the volume of its base and its height, divided by $d-1$. Therefore we get the recursive formula
  \begin{align}
  \vol(\OS_{d-1})=\frac{1}{d-1}\vol(\OS_{d-2})h_{d},
  \end{align}
  where we have defied the distance $h_{i}$ between $v_i$ and $H_i$. Let us calculate $h_d$. This can be done by taking the difference of $v_d$ and any point in $H_i$ and calculating the absolute value of its inner product with $n_d$. We thus get
  \begin{align}
    h_d&=|\langle n_d,v_d-v_1\rangle|\\
    &=\frac{1}{\sqrt{d(d-1)}}\left|-\frac{d-1}{d}+(d-2)\frac{1}{d}-\frac{d-1}{d}\right|\\
    &=\frac{1}{\sqrt{d(d-1)}}.
  \end{align}
  The recursion therefore becomes
    \begin{align}\label{eq:rec}
  \vol(\OS_{d-1})=\sqrt{\frac{d-1}{d}}\frac{1}{(d-1)^{2}}\vol(\OS_{d-2}).
  \end{align}
  The claimed formula for the volume is now proven by induction. $\OS_2$ is just the line from $(1,0)$ to $(1/2,1/2)$, so its volume is clearly
  \begin{align}
    \vol(\OS_2)=\frac{1}{\sqrt 2}=\frac{1}{\sqrt 2(1!)^2},
  \end{align}
  proving \cref{eq:simpvol} for $d=2$. For the induction step, assume that the formula \cref{eq:simpvol} holds for $d=k-1$. Then we have
  \begin{align}
    \vol(\OS_{k-1})&=\sqrt{\frac{k-1}{k}}\frac{1}{(k-1)^{2}}\vol(\OS_{k-2})\\
    &=\sqrt{\frac{k-1}{k}}\frac{1}{(k-1)^{2}}\frac{1}{\sqrt{k-1}((k-2)!)^2}\\
    &=\frac{1}{\sqrt{k}((k-1)!)^2}.\\
  \end{align}
  For the boundary volume, we can use the pyramid volume formula again to obtain
    \begin{align}
    \vol(\OS_{d-1})=\frac{1}{d-1}\vol(F_i)h_{i},
    \end{align}
    i.e.\ we obtain the formula
    \begin{align}
      \vol(\partial\OS_{d-1})&=\sum_{i=1}^d\vol(F_i)\\
      &=(d-1)\vol(\OS_{d-1})\sum_{i=1}^d\frac{1}{h_i}.
    \end{align}
    We calculate the heights $h_i$ for $i\neq d$. For $1<i<d$ we get in the same way as above for $i=d$,
    \begin{align}
    h_i&=|\langle n_i,v_i-v_1\rangle|\\
    &=\frac{1}{i\sqrt{2}}.
    \end{align}
    for $i=1$ we calculate
    \begin{align}
    h_1&=|\langle n_1,v_1-v_2\rangle|\\
    &=\frac{1}{2\sqrt{2}}.
    \end{align}
    Therefore we get the boundary volume
    \begin{align}
    \vol(\partial\OS_{d-1})&=(d-1)\vol(\OS_{d-1})\left(2\sqrt 2+\sqrt{d(d-1)}+\sqrt 2\sum_{i=2}^{d-1}i\right)\\
    &=(d-1)\vol(\OS_{d-1})\left(\sqrt 2+\sqrt{d(d-1)}+\frac{d(d-1)}{\sqrt 2}\right)\\
    &=\vol(\OS_{d-1})\left(\frac{d(d-1)^2}{\sqrt 2}+\sqrt d(d-1)^{3/2}+\sqrt 2(d-1)\right).
    \end{align}
\end{proof}

\printbibliography[title={References},heading=bibintoc]

\end{document}